%% file: main.tex
\documentclass[a4paper,UKenglish,cleveref, autoref, thm-restate]{lipics-v2021}
\RequirePackage[normalem]{ulem} % for command \sout
\pdfoutput=1
\hideLIPIcs

\title{Concrete Branching Bisimilarity for Processes with Time-outs}

\author{Gaspard Reghem}{ENS Paris-Saclay, Université Paris-Saclay, France}{gaspard.reghem@ens-paris-saclay.fr}{}{}

\author{Rob J. van Glabbeek}{School of Informatics, University of Edinburgh, UK \and School of Computer Science and Engineering, University of New South Wales, Sydney, Australia \and \url{http://theory.stanford.edu/~rvg/}}{rvg@cs.stanford.edu}{https://orcid.org/0000-0003-4712-7423}{}

\authorrunning{G. Reghem and R.\,J. van Glabbeek}

\Copyright{Gaspard Reghem and Robert J. van Glabbeek}



\relatedversiondetails[cite=RG24]{A variant of this paper featuring eliding time-out actions appears in CONCUR'24}{https://doi.org/10.4230/LIPIcs.CONCUR.2024.36}

\funding{Supported by Royal Society Wolfson Fellowship RSWF{\textbackslash}R1{\textbackslash}221008}

\nolinenumbers
\usepackage{macros}
\usepackage{stmaryrd}
\usepackage{multicol}
\usepackage{bussproofs}
\usepackage{tikz}

\begin{document}

\maketitle

\begin{abstract}
This paper provides an adaptation of branching bisimilarity to reactive systems with time-outs that does not enable eliding of time-out transitions. Multiple equivalent definitions are procured, along with a modal characterisation and a proof of its congruence property for a standard process algebra with recursion. The last section presents a complete axiomatisation for guarded processes without infinite sequences of unobservable actions.
\end{abstract}

\section{Introduction} \label{sec:intro}

\emph{Strong bisimilarity} \cite{Mi90ccs} is the default semantic equivalence on labelled transition systems (LTSs), modelling systems that move from state to state by performing discrete, uninterpreted actions. In \cite{strongreactivebisimilarity}, it has been generalised, under the name \emph{strong reactive bisimilarity},
to LTSs that feature, besides the hidden action $\tau$ \cite{Mi90ccs}, an unobservable \emph{time-out} action $\rt$ \cite{vG21}, modelling the end of a time-consuming activity from which we abstract. This addition significantly increases the expressiveness of the model \cite{vG23a,strongreactivebisimilarity}.

Applied to the verification of realistic distributed systems, strong bisimilarity is too fine an equivalence, especially because it does not cater to abstraction from internal activity. \emph{Branching bisimilarity} \cite{branching} is a variant that does abstract from internal activity, and lies at the basis of many verification toolsets \cite{BGKLNVWWW19,GLMS11}. The present paper, as well as \cite{RG24}, generalises branching bisimilarity to LTSs with time-outs, thereby combining the virtues of \cite{strongreactivebisimilarity} and \cite{branching}. The resulting notion of \emph{branching reactive bisimilarity}, proposed in \cite{RG24}, elides time-outs, in the sense that---using the process algebra notation to be formally introduced in Section~\ref{sec:process algebra}---the processes $a.\rt.b.0$ and $a.\rt.\rt.b.0$ (as well as $a.\rt.\tau.\rt.b.0$) are branching reactive bisimilar. Both require an unquantified positive but finite amount of rest between the actions $a$ and $b$. In this paper we propose a \emph{\tb reactive bisimilarity} that instead treats time-outs more like visible actions. We support this notion through a modal characterisation, congruence results for a standard process algebra with recursion, and a complete axiomatisation.

The addition of the time-out action $\rt$ aims at modelling the passage of time while staying in the realm of \emph{untimed} process algebra. Here, ``untimed'' means that our framework does not facilitate measuring time, even though it models whether a system can pause in some state or not. We assume that the execution of any action is instantaneous; thus, time elapses in states only. The amount of time spent in a state is dictated by the interaction of the system with an external entity called its \emph{environment}.

We call a system \emph{reactive} if it interacts with an environment able to allow or disallow visible actions. The environment represents a user or other systems, running in parallel, which has no control over $\tau$ or $\rt$ actions. If $X$ is the set of visible actions currently allowed by the environment and the system can perform any transition labelled by an element of $X \cup \{\tau\}$ then it will perform one of those transitions immediately. When a visible action is performed, it triggers the environment to choose a new set of allowed actions. If the environment is allowing $X$ and the system cannot perform any transition labelled by $\tau$ or any allowed action, then the system is said to be \emph{idling}. When the system idles, time-outs become executable, but the environment can also get impatient and choose a new $X$ before any time-out occurs. 
\advance\textheight 3pt

\advance\textheight -3pt
We have supposed that the environment cannot synchronise with the execution of a time-out, thus implying that, right after executing a time-out, the environment is still allowing the same set of allowed actions as before this execution. For example, the process $a.P + \rt.(a.Q + \tau.R)$ will never reach $Q$ because, for the time-out to happen, the environment has to block $a$ and so $a.Q + \tau.R$ can only be reached when the environment blocks $a$. In this case, the $\tau$-transition is always executed before the environment can allow $a$ again. 

Similarly, strong and [concrete] branching reactive bisimilarity satisfy the process algebraic law $\tau.P + \rt.Q = \tau.P$, essentially giving $\tau$ priority over $\rt$. Whereas this could have been formalised through an operational semantics in which the process $\tau.P + \rt.Q$ lacks an outgoing $\rt$-transition, here, and in \cite{strongreactivebisimilarity}, we derive an LTS for a standard process algebra with time-outs in a way that treats $\rt$ just like any other action. Instead, the priority of $\tau$ over $\rt$ is implemented in the reactive bisimilarity: its says that even though the transition
$\tau.P + \rt.Q \step\rt Q$ is present in our LTS, it will never be taken. This approach is not only simpler, it also generalises better to choices like $b.P + \rt.Q$, where the priority of $b$ over $\rt$ is conditional on the environment in which the system is placed, namely on whether or not this environment allows the $b$-action to occur.

From the system's perspective, the environment can be in two kinds of states: either allowing a specific set of actions, or being triggered to change. Our model does not stipulate how much time the environment takes to choose a new set of allowed actions once triggered, or even if it will ever make such a choice. Thus, the system could perform some transitions while the environment is triggered, especially those labelled $\tau$. In our view, the most natural way to see the environment is as another system executed in parallel, while enforcing synchronisation on all visible actions. This implies that the environment allows a set $X$ of actions when it idles in a state whose set of initial actions is $X$, and the environment is triggered when it is not idling, especially when it can perform a $\tau$-transition. In this paradigm, while the environment is triggered, any action can be allowed for a brief amount of time. However, there is no reason to believe that it will necessarily settle down on a specific set. For instance, this can happen if the environment reaches a \emph{divergence}: an infinite sequence of $\tau$-transitions.

In \cite{vG93}, seven (or nine) forms of branching bisimilarity are classified; they differ only in the treatment of divergence. In the present paper we are chiefly interested in divergence-free processes, on grounds that in the intuition of \cite{strongreactivebisimilarity} any sequence of $\tau$-transitions could be executed in time zero; yet we do wish to allow infinite sequences of $\rt$-transitions.
For divergence-free process all these forms of branching bisimilarity coincide.
Nevertheless, we do not formally exclude divergences, and in their presence our \tb reactive bisimilarity generalises the \emph{stability respecting branching bisimilarity} of \cite{vG93}, which differs from the default version from \cite{branching} through the presence of Clause 2.e of Definition~\ref{def:intuitive}. There does not exist a plausible reactive generalisation of the default version.

Section \ref{sec:brb} supplies the formal definition of \tb reactive bisimilarity as well as its rooted version, which will be shown to be its congruence closure. It also provides equivalent definitions that reduce our bisimilarity to a non-reactive one and illustrate that \tb reactive bisimilarity coincides with stability respecting branching bisimilarity in the absence of time-outs.

Section \ref{sec:modal} gives a modal characterisation of \tb reactive bisimilarity and its rooted version on an extension of the Hennessy-Milner logic. Section \ref{sec:congruence} introduces the process algebra $\ccsp$ along with an alternative characterisation of \tb reactive bisimilarity that will be used to prove that rooted \tb reactive bisimilarity is a full congruence for $\ccsp$.

Section \ref{sec:axiom} displays a complete axiomatisation of our bisimilarity on different fragments of $\ccsp$. Most completeness proofs rely on standard techniques like equation merging, but the very last one uses a relatively new method called ``canonical representatives''.

\section{Branching Reactive Bisimilarity}\label{sec:brb}

A \emph{labelled transition system} (LTS) is a triple $(\closed,Act,\rightarrow)$ with $\closed$ a set (of \emph{states} or \emph{processes}), $Act$ a set (of \emph{actions}) and ${\rightarrow}\in\closed\times Act\times\closed$. In this paper we consider LTSs with $Act:= A\uplus\{\tau,\rt\}$, where $A$ is a set of \emph{visible actions}, $\tau$ is the \emph{hidden or invisible action}, and $\rt$ the \emph{time-out action}. Let $A_\tau := A \cup \{\tau\}$. $P \step{\alpha} P'$ stands for $(P,\alpha,P') \in {\rightarrow}$ and these triplets are called \emph{transitions}. Moreover, $P \step{\opt{\alpha}} P'$ denotes that either $\alpha = \tau$ and $P = P'$, or $P \step{\alpha} P'$. Furthermore, \emph{paths} are sequences of connected transitions and $\pathtau$ is the reflexive-transitive closure of $\steptau$. The set of \emph{initial} actions of a process $P \in\closed$ is $\init{P}:=\{\alpha\in A_\tau \mid P{\step \alpha}\}$. Here $P{\step \alpha}$ means that there is a $Q$ with $P \step\alpha Q$.

\begin{definition}\rm \label{def:intuitive}
    A \emph{\tb reactive bisimulation} is a symmetric\footnote{meaning that $(P,Q)\in \R \Leftrightarrow (Q,P)\in \R$ and $(P,X,Q)\in \R \Leftrightarrow (Q,X,P)\in \R$} relation $\R \subseteq (\closed\times\closed) \cup (\closed\times\Pow(A)\times\closed)$ such that, for all $P,Q \in \closed$ and $X \subseteq A$,
    \begin{enumerate}
        \item if $\R(P,Q)$ then
        \begin{enumerate}
            \item if $P \step{\alpha} P'$ with $\alpha \in A_\tau$ then there is a path $Q \pathtau Q_1 \step{\opt{\alpha}} Q_2$ with $\R(P,Q_1)$ and $\R(P',Q_2)$,
            \item for all $Y \subseteq A$, $\R(P,Y,Q)$;
        \end{enumerate}
        \item if $\R(P,X,Q)$ then
        \begin{enumerate} 
            \item if $P \steptau P'$ then there is a path $Q \pathtau Q_1 \step{\opt{\tau}} Q_2$ with $\R(P,X,Q_1)$ and $\R(P',X,Q_2)$,
            \item if $P \step{a} P'$ with $a \in X$ then there is a path $Q \pathtau Q_1 \step{a} Q_2$ with $\R(P,X,Q_1)$ and $\R(P',Q_2)$,
            \item if $\deadend{P}{X}$ then there is a path $Q \pathtau Q_0$ with $\R(P,Q_0)$,
            \item if $\deadend{P}{X}$ and $P \step{\rt} P'$ then there is a path $Q \pathtau Q_1 \step{\rt} Q_2$ with $\R(P',X,Q_2)$,
            \item if $P \nsteptau$ then there is a path $Q \pathtau Q_0 \nsteptau$.
        \end{enumerate}
    \end{enumerate}
For $P,Q \mathbin\in \closed$, if there exists a \tb reactive bisimulation $\R$ with $\R(P,Q)$ (resp.\ $\R(P,X,Q)$) then $P$ and $Q$ are said to be \emph{\tb reactive bisimilar} (resp.\ \emph{\tb $X$-bisimilar}), which is denoted $P \bisimtbrc Q$ (resp.\ $P \bisimtbrc[X] Q$).
\end{definition}

\noindent
To build the above definition, the definition of a strong reactive bisimulation \cite{strongreactivebisimilarity} was modified in a branching manner \cite{branching}. Intuitively, a triplet $\R(P,X,Q)$ affirms that $P$ and $Q$ behave similarly when the environment allows (only) the set of actions in $X$ to occur, whereas a couple $\R(P,Q)$ says that $P$ and $Q$ behave in the same way when the environment has been triggered to change. As said before, the environment can be seen as a system executed in parallel while enforcing the synchronisation of all visible actions. 

Clause 1 captures the scenario of a triggered environment: if $P$ can perform a visible or invisible action then $Q$ has to be able to match it; and the environment can settle on a set $Y$ of allowed actions at any moment. Time-outs are not considered because these can occur only when the system idles, and idling can happen only when the environment has stabilised on a set of allowed actions. One might notice that, in \cite{strongreactivebisimilarity}, the first clause was only required for invisible actions. However, there the case $\alpha\neq\tau$ is actually implied by the other clauses. If in our definition Clause 1.a were restricted to invisible actions then $\bisimtbrc$ would not be a congruence for the parallel operator, as shown in Appendix \ref{app:examples}.

Clause 2 depicts the scenario of an environment allowing $X$. $\tau$-transitions have to be matched
since the environment cannot disallow them, and their execution does not trigger the environment to
change. Visible actions have to be matched only if they are allowed, and their execution triggers
the environment. Triggering the environment or not explains why Clause 2a matches $Q_2$ in a triplet
and Clause 2b in a couple. If $P$ idles (i.e.\ $\deadend{P}{X}$) then the environment can be
triggered, thus, $Q$ has to be able to instantaneously reach a state $Q_0$ related to $P$ in a
triggered environment.\footnote{By Lemma~\ref{lem:obvious}.4 we can even choose $Q_0$ such that $Q_0 \nsteptau$, so that $\init{Q_0}=\init{P}$.} If $P$ idles and has an outgoing time-out transition then $Q$ has to be able to match it in a branching manner. Unlike the branching reactive bisimulation, this bisimulation does not allow to elide time-out transitions, thus, the matching resembles the one with visible action except that the execution of a time-out transition does not trigger the environment.\footnote{It is not necessary to match $P$ and $Q_1$ as it is implied by other clauses (see Lemma \ref{lem:obvious}.1), nor to require that $\deadend{Q_1}{X}$ because it is also implied by the other clauses (see Lemma \ref{lem:obvious}.4).}
Lastly, a stability respecting clause \cite{vG93} was added for practical reasons. In Appendix \ref{app:examples}, an example shows that without it $\bisimtbrc$ would not even be an equivalence. For the important class of \emph{divergence-free} systems, without infinite sequences $Q_0 \steptau Q_{1} \steptau \dots$, Clause 2.e is easily seen to be redundant.

\begin{lemma}\label{lem:obvious}
  Let $\R$ be a \tb reactive bisimulation.
  \begin{enumerate}
  \item If $\R(P,X,Q)$, $P \nsteptau$ and $Q \pathtau Q'$ then also $\R(P,X,Q')$.
  \item If $\R(P,Q)$ or $\R(P,X,Q)$, $P \nsteptau$ and $Q \nsteptau$ then $\init{Q}=\init{P}$.
  \item If $\R(P,X,Q)$, $\deadend{P}{X}$ and $Q \nsteptau$ then $\R(P,Q)$.
  \item If $\R(P,X,Q)$ and $\deadend{P}{X}$ then there is a path $Q \pathtau Q_0$ with $\R(P,Q_0)$, $Q_0\nsteptau$ and $\init{Q_0}=\init{P}$.
  \end{enumerate}
\end{lemma}

\begin{proof}
\begin{enumerate}
    \item This is an immediate consequence of the symmetric counterpart of Clause 2.a (where $Q$ takes a $\tau$-step).
    When that clause yields $P \pathtau P_1 \step{\opt{\tau}} P_2$ we have $P_2=P$.
    \item This is a direct consequence of Clause 1.a or 2.b and its symmetric counterpart.
    \item By Clause 2.e there is path $Q \pathtau Q_0$ with $Q_0 \nsteptau$. By Claim 1 of this lemma, $\R(P,X,Q_0)$. Thus, by Clause 2.c there is a path $Q_0 \mathbin{\pathtau} Q_1$ with $\R(\hspace{-1pt}P,\hspace{-1pt}Q_1\hspace{-1pt})$, but $Q_1\mathop{=}Q_0\mathop{=}Q$ since $Q \!\nsteptau$.
    \item By Clause 2.e there is path $Q \pathtau Q_0$ with $Q_0 \nsteptau$. By Claim 1 of this lemma, $\R(P,X,Q_0)$. That $\init{Q_0}=\init{P}$ and $\R(P,Q_0)$ follows by Claims 2 and 3 of this lemma.
\popQED
\end{enumerate}
\end{proof}

\noindent
In \cite{branching}, branching bisimilarity is expressed in multiple equivalent ways. For practical purposes, our definition uses the semi-branching format, which is equivalent to the branching format thanks to the following lemma.

\begin{lemma}[Stuttering Lemma] \label{lem:stuttering}
    Let $P, P^\dag, P^\ddag, Q \in \closed$, if $P \bisimtbrc Q$, $P^\ddag \bisimtbrc Q$ (resp.\ $P \bisimtbrc[X] Q$, $P^\ddag \bisimtbrc[X] Q$) and $P \steptau P^\dag \steptau P^\ddag$ then $P^\dag \bisimtbrc Q$ (resp.\ $P^\dag \bisimtbrc[X] Q$).
\end{lemma}

\begin{proof}
    Let $\R$ be a \tb reactive bisimulation. Let's define $\R' := \R \cup \{(P^\dag,Q),(Q,P^\dag) \mid \exists P,P^\ddag \in \closed, P \pathtau P^\dag \pathtau P^\ddag \wedge \R(P,Q) \wedge \R(P^\ddag,Q)\} \cup \{(P^\dag,X,Q),\linebreak[3](Q,X,P^\dag) \mid \exists P,P^\ddag \mathbin\in \closed, P \pathtau P^\dag \pathtau P^\ddag \wedge \R(P,X,Q) \wedge \R(P^\ddag,X,Q)\}$. $\R'$ is symmetric by definition and $\R'$ is a \tb reactive bisimulation, as proven in Appendix \ref{app:intro}.
\end{proof}

\begin{proposition} \label{prop:equivalence}
    $\bisimtbrc$ and $(\bisimtbrc[X])_{X \subseteq A}$ are equivalence relations.
\end{proposition}

\begin{proof}
    Reflexivity and symmetry are trivial following the definition. For transitivity, consider two \tb reactive bisimulations $\R_1$ and $\R_2$. Let's define $\R := (\R_1 \circ \R_2) \cup (\R_2 \circ \R_1)$. Here $\R_1 \circ \R_2 := \{(P,Q) \mid \exists R.~\R(P,R) \wedge \R(R,Q)\} \cup \{(P,X,Q) \mid \exists R.~\R(P,X,R) \wedge \R(R,X,Q)\}$.\linebreak[3] $\R$ is symmetric by definition and $\R$ is a \tb reactive bisimulation, as proven in Appendix \ref{app:intro}.
\end{proof}

\subsection{Rooted Version}

A well-known limitation of branching bisimilarity $\bisimb$ is that it fails to be a congruence for the choice operator $+$. For example, $a \bisimb \tau.a$ but $a + b \,\not\!\bisimb \tau.a + b$. Since the objective is to define a congruence, instead of $\bisimtbrc$ we use the \emph{congruence closure} of $\bisimtbrc$\,, which is the coarsest congruence included in $\bisimtbrc$\,. 

\begin{definition}\rm\label{def:rooted intuitive}
    A \emph{rooted \tb reactive bisimulation} is a symmetric relation $\R \subseteq (\closed\times\closed)\cup(\closed\times\Pow(A)\times\closed)$ such that, for all $P,Q \in \closed$ and $X \subseteq A$,
    \begin{enumerate}
        \item if $\R(P,Q)$
        \begin{enumerate}
            \item if $P \step\alpha P'$ with $\alpha \in A_\tau$ then there is a transition $Q \step\alpha Q'$ with $P' \bisimtbrc Q'$,
            \item for all $Y \subseteq A$, $\R(P,Y,Q)$;
        \end{enumerate}
        \item if $\R(P,X,Q)$
        \begin{enumerate}
            \item if $P \steptau P'$ then there is a transition $Q \steptau Q'$ with $P' \bisimtbrc[X] Q'$,
            \item if $P \step{a} P'$ with $a \in X$ then there is a transition $Q \step{a} Q'$ with $P' \bisimtbrc Q'$,
            \item if $\deadend{P}{X}$ then $\R(P,Q)$,
            \item if $\deadend{P}{X}$ and $P \step{\rt} P'$ then there is a transition $Q \step{\rt} Q'$ with $P' \bisimtbrc[X] Q'$.\vspace{3pt}
        \end{enumerate}
    \end{enumerate}
For $P,Q \in \closed$, if there exists a rooted \tb reactive bisimulation $\R$ with $\R(P,Q)$ (resp.\ $\R(P,X,Q)$) then $P$ and $Q$ are said to be \emph{rooted \tb reactive bisimilar} (resp.\ rooted \tb $X$-bisimilar), which is denoted $P \bisimrtbrc Q$ (resp.\ $P \bisimrtbrc[X] Q$).
\end{definition}
A rooted version of a bisimulation consists in enforcing a stricter matching on the first transition of a system. In the branching case, the first transition is matched in the strong manner. The stability respecting clause can be removed, as it is now implied by the other clauses. Rooting the bisimilarity is the standard technique to obtain its congruence closure; later $\bisimrtbrc$ will be proven to be a congruence. As any \tb reactive bisimulation relating $P+b$ and $Q+b$, for a fresh action $b$, induces a rooted \tb reactive bisimulation relating $P$ and $Q$, it then follows that $\bisimrtbrc$ is the coarsest congruence included in $\bisimtbrc$. Since $\bisimtbrc$ is an equivalence, the proof of Proposition \ref{prop:equivalence} can be adapted to $\bisimrtbrc$ in a straightforward way.

\begin{proposition} \label{prop:rooted equivalence}
    $\bisimrtbrc$ and $(\bisimrtbrc[X])_{X \subseteq A}$ are equivalence relations.
\end{proposition}

\subsection{Alternative Forms of Definition~\ref{def:intuitive}}

Definition~\ref{def:intuitive} can be rephrased in various ways. First of all, using Requirements 1.b and 2.c, one can move Requirement 2.d from Clause 2 (dealing with triples $(P,X,Q)$) to Clause~1 (dealing with pairs $(P,Q)$), now adding a universal quantifier over $X$ to the requirement. Next, Requirement 2.e can be copied under Clause~1. This makes Clause 1.b unnecessary, thereby obtaining a definition in which the triples $(P,X,Q)$ are encountered only after taking a $\rt$-transition. In this form it is obvious that \tb reactive bisimilarity reduces to the classical stability respecting branching bisimilarity for systems without $\rt$-transitions. We have chosen the form of Definition~\ref{def:intuitive} over the above alternatives, because we believe it comes with more natural intuitions for its plausibility.

In Appendix~\ref{app:gbrb} a further modification of Definitions~\ref{def:intuitive} and~\ref{def:rooted intuitive} is proposed, called \emph{generalised [rooted] \tb reactive bisimulation}. We show that each [rooted] \tb reactive bisimulation is a generalised [rooted] \tb reactive bisimulation, and two systems are [rooted] \tb reactive bisimilar iff they are related by a generalised [rooted] \tb reactive bisimulation.
This characterisation of $\bisimtbrc$ and $\bisimrtbrc$ will be used in the proofs of Theorem~\ref{thm:modal characterisation} and Proposition~\ref{prop:time-out bisim}.

In \cite{Pohlmann}, Pohlmann introduces an encoding which maps strong reactive bisimilarity to strong bisimilarity where time-outs are considered as any visible action. This encoding in essence places a given process in a most general environment, one that features environment time-out actions $\rt_\varepsilon$, as well as actions $\varepsilon_X$ for settling in a state that allows exactly the actions in $X$. This proves that reactive equivalences can be expressed as non-reactive ones at the cost of increasing the processes' size. Thus, any tool set able to work on strong bisimulation could theoretically deal with its reactive counterpart. 

In Appendix \ref{app:Pohlmann}, this encoding is slightly modified to yield a similar result for \tb reactive bisimulation and its rooted version, for the latter result also employing actions $\rt_X$. It appears that these modifications do not impact its effect on strong reactive bisimilarity. This encoding maps our bisimilarity to the traditional stability respecting branching bisimilarity and $\bisimrtbrc$ to rooted stability respecting branching bisimilarity \cite{vG93,modalstab}.

\begin{definition}\rm \label{def:non-reactive}
    A \emph{stability respecting branching bisimulation} is a symmetric relation $\R \subseteq \closed\times\closed$ such that, for all $P,Q \in \closed$, if $\R(P,Q)$ then
    \begin{enumerate}
        \item if $P \step{\alpha} P'$ with $\alpha \in Act \cup\{\rt_\epsilon,\epsilon_X \mid X \subseteq A\}$ then there is a path $Q \pathtau Q_1 \step{\opt{\alpha}} Q_2$ with $\R(P,Q_1)$ and $\R(P',Q_2)$,
        \item if $P \nsteptau$ then there is a path $Q \pathtau Q_0 \nsteptau$.
\vspace{3pt}

\end{enumerate}
For $P,Q \in \closed$, if there exists a stability respecting branching bisimulation $\R$ with $\R(P,Q)$ then $P$ and $Q$ are said to be \emph{stability respecting branching bisimilar}, which is denoted $P \bisimb Q$.
\end{definition}

\begin{definition}\rm \label{def:rooted non-reactive}
    A \emph{rooted stability respecting branching bisimulation} is a symmetric relation $\R \subseteq \closed\times\closed$ such that, for all $P,Q \in \closed$, if $\R(P,Q)$ then
    \begin{enumerate}
        \item if $P \mathord{\step{\alpha}} P'$ with $\alpha \mathord\in Act \mathop\cup\{\rt_\epsilon,\epsilon_X \mathbin{\mid} X \mathord\subseteq A\}$ then there is a transition $Q \mathop{\step{\alpha}} Q'$ with $P' \bisimb Q'\!$.\vspace{3pt}
    \end{enumerate}
For $P,Q \in \closed$, if there exists a rooted stability respecting branching bisimulation $\R$ with $\R(P,Q)$ then $P$ and $Q$ are said to be \emph{rooted stability respecting branching bisimilar}, which is denoted $P \bisimrb Q$.
\end{definition}

\section{Modal Characterisation} \label{sec:modal}

The Hennessy-Milner logic \cite{HM85} expresses properties of the behaviour of processes in an LTS\@. In \cite{strongreactivebisimilarity}, the modality $\langle X\rangle\varphi$ was added to obtain a modal characterisation of strong reactive bisimilarity ($\rbis{}{r}$). 

\begin{definition}\rm \label{def:Hennessy-Milner}
    The class $\logic$ of \emph{reactive Hennessy-Milner formulas} is defined as follows, where $I$ is an index set, $\alpha \in A_\tau$, $a \in A$ and $X \subseteq A$,
    \begin{center}
        $\varphi ::= \top \mid \bigwedge\limits_{i \in I} \varphi_i \mid \neg\varphi \mid \langle\alpha\rangle\varphi \mid \langle X\rangle\varphi$
    \end{center}
    \vspace{-6pt}
\pagebreak[3]
\end{definition}

\begin{table}[ht]
\begin{center}
    \begin{tabular}{l c l l c l}
        $P \models \top$ &&& $P \models_Y \top$ \\
        $P \models \bigwedge_{i\in I}\varphi_i$ & iff & $\forall i \in I, P \models \varphi_i$ & $P \models_Y \bigwedge_{i\in I}\varphi_i$ & iff & $\forall i \in I, P \models_Y \varphi_i$ \\
        $P \models \neg\varphi$ & iff & $P \not\models \varphi$ & $P \models_Y \neg\varphi$ & iff & $P \not\models_Y \varphi$ \\
        $P \models \langle\alpha\rangle\varphi$ & iff & $\exists P \step{\alpha} P', P' \models \varphi$ & $P \models_Y \langle\tau\rangle\varphi$ & iff & $\exists P \step{\tau} P', P' \models_Y \varphi$ \\
        $P \models_Y \langle a\rangle\varphi$ & iff & \multicolumn{4}{l}{$(a \in Y \vee \deadend{P}{Y}) \wedge \exists P \step{a} P', P' \models \varphi$} \\
        $P \models \langle X\rangle\varphi$ & iff & \multicolumn{4}{l}{$\deadend{P}{X} \wedge \exists P \step{\rt} P', P' \models_X \varphi$} \\
        $P \models_Y \langle X\rangle\varphi$ & iff & \multicolumn{4}{l}{$\deadend{P}{X \cup Y} \wedge \exists P \step{\rt} P', P' \models_X \varphi$} \\
    \end{tabular}
\end{center}
    \caption{Semantics of $\models$ and $(\models_Y)_{Y \subseteq A}$}
    \label{tab:formulas semantics}
\vspace{-2.5ex}
\end{table}

\noindent
The satisfaction rules of $\logic$ are given in Table \ref{tab:formulas semantics}. $P \models \varphi$ means that $P$ satisfies $\varphi$ when the environment is triggered, and $P \models_Y \varphi$ indicates that $P$ satisfies $\varphi$ when the environment allows $Y$. The modality $\langle X\rangle\varphi$ expresses that a process can idle in its current state during a period in which the environment allows the actions in $X$ and from which it can perform a time-out transition to a state which satisfies $\varphi$ while the environment keeps allowing $X$. The definition above captures that $P \models_Y \varphi$ whenever $\init{P}\cap (Y\cup\{\tau\}) = \emptyset$ and $P \models \varphi$. This is because the environment may choose to change during a period of idling.
The modal characterisation theorem of \cite{strongreactivebisimilarity} says \( P \rbis{}{r} Q ~~\Leftrightarrow~~ \forall \varphi\in\logic.~(P \models \varphi \Leftrightarrow Q \models \varphi)\;. \)

To obtain a modal characterisation of [rooted] \tb relative bisimilarity, we need a few other derived modalities. First of all, $\langle\epsilon\rangle\varphi := \bigvee_{i \in \nat}\langle\tau\rangle^i\varphi$. To lessen the notations, for all $\alpha \in A_\tau$, $\langle\hat{\alpha}\rangle\varphi$ denotes $\varphi \vee \langle\tau\rangle\varphi$ if $\alpha = \tau$, $\langle\alpha\rangle\varphi$ otherwise. Moreover, $\varphi \wedge \langle\hat{\alpha}\rangle\varphi'$ is shortened to $\varphi\langle\hat{\alpha}\rangle\varphi'$.
The satisfaction rules of these new modalities can be derived from the basic ones: see Table \ref{tab:operator semantics}.

\begin{table}[h]
\begin{center}
    \begin{tabular}{l c l l c l}
        $P \models \langle\hat{\alpha}\rangle\varphi$ & iff & $\exists P \step{\opt{\alpha}} P', P' \models \varphi$ & $P \models_Y \langle\hat{\tau}\rangle\varphi$ & iff & $\exists P \step{\opt{\tau}} P', P' \models_Y \varphi$ \\
        $P \models \langle\epsilon\rangle\varphi$ & iff & $\exists P \pathtau P', P' \models \varphi$ & $P \models_Y \langle\epsilon\rangle\varphi$ & iff & $\exists P \pathtau P', P' \models_Y \varphi$ \\
    \end{tabular}
\end{center}
\caption{Semantics of $\models$ and $(\models_Y)_{Y \subseteq A}$ for the derived modalities}
\label{tab:operator semantics}
\vspace{-2ex}
\end{table}

\begin{definition}\rm \label{def:subclass}
    The sub-classes $\logic_b^c$ and $\logic_b^{cr}$ are defined as follows, where $I$ is an index set, $\alpha \in A_\tau$, $X \subseteq A$, $\varphi, \varphi' \in \logic_b^c$ and $\psi \in \logic_b^{cr}$,
    \begin{align*}
        \varphi &::= \top \mid \bigwedge_{i \in I} \varphi_i \mid \neg\varphi \mid \langle\epsilon\rangle(\varphi\langle\hat{\alpha}\rangle\varphi') \mid \langle\epsilon\rangle\langle X\rangle\varphi' \mid \langle\epsilon\rangle\neg\langle\tau\rangle\top \tag{$\logic_b^c$} \\
        \psi &::= \top \mid \bigwedge_{i \in I} \psi_i \mid \neg\psi \mid \langle\alpha\rangle\varphi \mid \langle X\rangle\varphi \tag{$\logic_b^{cr}$}
    \end{align*}
\end{definition}

\noindent
The last option for $\logic_b^c$, inspired by \cite{modalstab}, is used to encompass the stability respecting Clause 2.e of Definition~\ref{def:intuitive}.

\begin{theorem} \label{thm:modal characterisation}
    Let $P,Q \in \closed$. For all $X \subseteq A$,
    \begin{itemize}
        \item $P \bisimtbrc Q$ iff $\forall \varphi \in \logic_b^c, P \models \varphi \Leftrightarrow Q \models \varphi$,
        \item $P \bisimtbrc[X] Q$ iff $\forall \varphi \in \logic_b^c, P \models_X \varphi \Leftrightarrow Q \models_X \varphi$,
        \item $P \bisimrtbrc Q$ iff $\forall \psi \in \logic_b^{cr}, P \models \psi \Leftrightarrow Q \models \psi$,
        \item $P \bisimrtbrc[X] Q$ iff $\forall \psi \in \logic_b^{cr}, P \models_X \psi \Leftrightarrow Q \models_X \psi$.
    \end{itemize}
\end{theorem}

\begin{proof}
    $(\Rightarrow)$ The four propositions are proven simultaneously by structural induction on $\logic_b^c$ and $\logic_b^{cr}$ in Appendix \ref{app:modal}.

    $(\Leftarrow)$ Let $\equiv \; := \{(P,Q) \mid \forall \varphi \in \logic_b^c, P \models \varphi \Leftrightarrow Q \models \varphi\} \cup \{(P,X,Q) \mid \forall \varphi \in \logic_b^c, P \models_X \varphi \Leftrightarrow Q \models_X \varphi\}$, and $\equiv^r \; := \{(P,Q) \mid \forall \psi \in \logic_b^{cr}, P \models \psi \Leftrightarrow Q \models \psi\} \cup \{(P,X,Q) \mid \forall \psi \in \logic_b^{cr}, P \models_X \psi \Leftrightarrow Q \models_X \psi\}$. It suffices to check that $\equiv$ [resp.\ $\equiv^r$] is a generalised [rooted] \tb reactive bisimulation. This is done in Appendix \ref{app:modal}.
\end{proof}

\section{Process Algebra and Congruence} \label{sec:congruence}
\label{sec:process algebra}

The process algebra $\ccsp$ is composed of classical operators from the well-known process algebras CCS \cite{Mi90ccs}, CSP \cite{BHR84,OH86} and ACP \cite{BW90,Fok00}, as well as the time-out action $\rt$ and two \emph{environment operators} from \cite{strongreactivebisimilarity}, that were added in order to enable a complete axiomatisation.

\begin{definition}\rm \label{def:ccsp syntax}
    Let $V$ be a countable set of variables, the \emph{expressions} of $\ccsp$ are recursively defined as follows:
    \begin{align*}
        E ::= 0 \mid x \mid \alpha.E \mid E + F \mid E \parallel_S F \mid \tau_I(E) \mid \rename(E) \mid \theta_L^U(E) \mid \psi_X(E) \mid \langle y | \equa \rangle
    \end{align*}
    where $x \in V$, $\alpha \in Act$, $S, I, L, U, X \subseteq A$, $L \subseteq U$, $\rename \subseteq A\times A$, $\equa$ is a \emph{recursive specification}: a set of equations $\{x\mathbin=\equa_x \mid x \mathbin\in V_\equa\}$ with $V_\equa \subseteq V$ and each $\equa_x$ a $\ccsp$ expression, and $y\mathbin\in V_\equa$.\linebreak[3] We require that all sets ${\{b\mid (a,b)\in \rename\}}$ are finite.
\end{definition}

\noindent
$0$ stands for a system which cannot perform any action. The expression $\alpha.E$ represents a system that first performs $\alpha$ and then $E$. The expression $E + F$ represents a choice to behave like $E$ or $F$. The parallel composition $E \parallel_S F$ synchronises the execution of $E$ and $F$, but only when performing actions in $S$. $\tau_I(E)$ represents the system $E$ where all actions $a \mathbin\in I$ are transformed into~$\tau$. The operator $\rename$ renames a given action $a\mathbin\in A$ into a choice between all actions $b$ with $(a,b)\mathbin\in \rename$. $\langle y | \equa\rangle$ is the $y$-component of a solution of $\equa$. 

$\ccsp$ also has two environment operators that help to develop a complete axiomatisation (like the left merge for ACP). $\theta_L^U(E)$ is the expression $E$ plunged into an environment $X$ such that $L \subseteq X \subseteq U$. $\theta_X^X(E)$ is denoted $\theta_X(E)$. $\psi_X(E)$ plunges $E$ into the environment $X$ if a time-out occurs, but, has no effect if any other action is performed. The operational semantics of $\ccsp$ is given in Figure~\ref{fig:ccsp semantics}. All operators except the environment ones follow the semantics of CCS, CSP or ACP\@. As $\theta_L^U(E)$ simulates the expression $E$ plunged in an environment $L \subseteq X \subseteq U$, it has no effect on $\tau$-transitions, which do not trigger the environment. Moreover, $\theta_L^U$ restricts the ability to perform visible actions to those allowed by the environment (i.e.\ included in $U$) and performing these actions triggers the environment. However, if the expression idles (i.e.\ $\deadend{E}{L}$) then it might trigger the environment and $\theta_L^U(E)$ acts like $E$. $\psi_X(E)$ supposes that time-outs are performed while the environment allows $X$, thus, it has no effect on actions that are not $\rt$. However, if $E$ can perform a time-out while the environment allows $X$ (i.e. $\deadend{E}{X}$) then $\psi_X(E)$ can perform the time-out while plunging the expression in the environment $X$.

\begin{figure}[ht]
\centering
    \begin{prooftree}
        \AxiomC{\textcolor{white}{$x \step{\alpha} y$}}
        \UnaryInfC{$\alpha.x \step{\alpha} x$}
        \DisplayProof
        \hskip 2em
        \AxiomC{$x \step{\alpha} x'$}
        \UnaryInfC{$x + y \step{\alpha} x'$}
        \DisplayProof
        \hskip 2em
        \AxiomC{$y \step{\alpha} y'$}
        \UnaryInfC{$x + y \step{\alpha} y'$}
    \end{prooftree}
    
    \begin{prooftree}
        \AxiomC{$x \step{a} x' \wedge \rename(a,b)$}
        \UnaryInfC{$\rename(x) \step{b} \rename(x')$}
        \DisplayProof
        \hskip 2em
        \AxiomC{$x \steptau x'$}
        \UnaryInfC{$\rename(x) \steptau \rename(x')$}
        \DisplayProof
        \hskip 2em
        \AxiomC{$x \step{\rt} x'$}
        \UnaryInfC{$\rename(x) \step{\rt} \rename(x')$}
    \end{prooftree}

    \begin{prooftree}
        \AxiomC{$x \step{\alpha} x' \wedge \alpha \not\in S$}
        \UnaryInfC{$x\parallel_S y \step{\alpha} x' \parallel_S y$}
        \DisplayProof
        \hskip 2em
        \AxiomC{$y \step{\alpha} y' \wedge \alpha \not\in S$}
        \UnaryInfC{$x \parallel_S y \step{\alpha} x \parallel_S y'$}
        \DisplayProof
        \hskip 2em
        \AxiomC{$x \step{a} x' \wedge y \step{a} y' \wedge a \in S$}
        \UnaryInfC{$x \parallel_S y \step{a} x' \parallel_S y'$}
    \end{prooftree}

    \begin{prooftree}
        \AxiomC{$x \step{\alpha} x' \wedge \alpha \not\in I$}
        \UnaryInfC{$\tau_I(x) \step{\alpha} \tau_I(x')$}
        \DisplayProof
        \hskip 2em
        \AxiomC{$x \step{a} x' \wedge a \in I$}
        \UnaryInfC{$\tau_I(x) \steptau \tau_I(x')$}
        \DisplayProof
        \hskip 2em
        \AxiomC{$\langle \equa_x | \equa \rangle \step{\alpha} x'$}
        \UnaryInfC{$\langle x | \equa\rangle \step{\alpha} x'$}
    \end{prooftree}

    \begin{prooftree}
        \AxiomC{$x \steptau x'$}
        \UnaryInfC{$\theta_L^U(x) \steptau \theta_L^U(x')$}
        \DisplayProof
        \hskip 2em
        \AxiomC{$x \step{a} x' \wedge a \in U$}
        \UnaryInfC{$\theta_L^U(x) \step{a} x'$}
        \DisplayProof
        \hskip 2em
        \AxiomC{$x \step{\alpha} x' \wedge \alpha \neq t$}
        \UnaryInfC{$\psi_X(x) \step{\alpha} x'$}
    \end{prooftree}

    \begin{prooftree}
        \AxiomC{$x \step{\alpha} x' \wedge \deadend{x}{L}$}
        \UnaryInfC{$\theta_L^U(x) \step{\alpha} x'$}
        \DisplayProof
        \hskip 2em
        \AxiomC{$x \step{\rt} x' \wedge \deadend{x}{X}$}
        \UnaryInfC{$\psi_X(x) \step{\rt} \theta_X(x')$}
    \end{prooftree}
    \caption{Operational semantics of $\ccsp$}
    \label{fig:ccsp semantics}
\end{figure}

All $\equa_x$ are considered to be sub-expressions of $\langle y | \equa\rangle$. An occurrence of a variable $x$ is \emph{bound} in $E \in \ccsp$ iff it occurs in a sub-expression $\langle y |\equa\rangle$ of $E$ such that $x \in V_\equa$; otherwise it is \emph{free}. An expression $E$ is \emph{invalid} if it has a sub-expression $\theta_L^U(F)$ or $\psi_X(F)$ such that a variable occurrence is free in $F$, but bound in $E$. An example justifying this condition can be found in \cite{strongreactivebisimilarity}. The set of valid expressions of $\ccsp$ is denoted $\expr$. If an expression is valid and all of its variable occurrences are bound then it is \emph{closed} and we call it a \emph{process}; the set of processes is denoted $\closed$.

A \emph{substitution} is a partial function $\rho: V \rightharpoonup E$. The application $E[\rho]$ of a substitution $\rho$ to an expression $E \in \expr$ is the result of the simultaneous replacement, for all $x \in \dom{\rho}$, of each free occurrence of $x$ by the expression $\rho(x)$, while renaming bound variables to avoid name clashes. We write $\langle E|\equa\rangle$ for the expression $E$ where any $y \in V_\equa$ is substituted by $\langle y | \equa\rangle$. 

\subsection{Time-out Bisimulation}

Thanks to the environment operator $\theta_L^U$, it is possible to express our bisimilarity in a much more succinct way. Indeed, $\theta_X$ was defined so that $P \bisimtbrc[X] Q$ if and only if $\theta_X(P) \bisimtbrc \theta_X(Q)$.

\begin{definition}\rm \label{def:time-out bisim}
    A \emph{\tb time-out bisimulation} is a symmetric relation ${\tbisim} \subseteq \closed\times\closed$ such that, for all $P,Q \in \closed$, if $P \tbisim Q$ then
    \begin{enumerate}
        \item if $P \step{\alpha} P'$ with $\alpha \in A_\tau$ then there is a path $Q \pathtau Q_1 \step{\opt{\alpha}} Q_2$ with $P \tbisim Q_1$ and $P' \tbisim Q_2$
        \item if $\deadend{P}{X}$ and $P \step{\rt} P'$ then there is a path $Q \pathtau Q_1 \step{\rt} Q_2 $ with $\theta_X(P') \tbisim \theta_X(Q_2)$
        \item if $P \nsteptau$ then there is a path $Q \pathtau Q_0 \nsteptau$.
    \end{enumerate}
\end{definition}

\noindent
Note that in Condition~2 above one also has $P \tbisim Q_1$ and consequently $\deadend{Q_{1}}{X}$. A rooted version of \tb time-out bisimulation can be defined in the same vein.

\begin{definition}\rm \label{def:rooted time-out bisim}
    A \emph{rooted \tb time-out bisimulation} is a symmetric relation ${\tbisim} \subseteq \closed\times\closed$ such that, for all $P,Q \in \closed$ such that $P \tbisim Q$,
    \begin{enumerate}
        \item if $P \step{\alpha} P'$ with $\alpha \in A_\tau$ then there is a step $Q \step{\alpha} Q'$ such that $P' \bisimtbrc Q'$
        \item if $\init{P}\cap(X\cup\{\tau\}) \mathbin= \emptyset$ and $P \mathbin{\step{\rt}} P'$ then there is a step $Q \mathbin{\step{\rt}} Q'$ such that $\theta_X(P') \mathbin{\bisimtbrc} \theta_X(Q').$
    \end{enumerate}
\end{definition}

\begin{proposition} \label{prop:time-out bisim}
    Let $P,Q \in \closed$, 
    \begin{enumerate}
        \item $P \bisimtbrc Q$ (resp.\ $P \bisimtbrc[X] Q$) iff there exists a \tb time-out bisimulation $\tbisim$ with $P \tbisim Q $ (resp.\ $(\theta_X(P) \tbisim \theta_X(Q)$),
        \item $P \bisimtbrc[X] Q$ if and only if $\theta_X(P) \bisimtbrc \theta_X(Q)$,\label{corr}
        \item $P \bisimrtbrc Q$ (resp.\ $P \bisimrtbrc[X] Q$) iff there exists a rooted \tb time-out bisimulation $\tbisim$ with $P \tbisim Q $ (resp.\ $(\theta_X(P) \tbisim \theta_X(Q)$).
    \end{enumerate}
\end{proposition}

\begin{proof}
  Note that Proposition~\ref{prop:time-out bisim}.\ref{corr} is a trivial corollary of \ref{prop:time-out bisim}.1.
  
  Let $\R$ be a [generalised rooted] \tb reactive bisimulation, let's define ${\tbisim} := \{(P,Q) \mid \R(P,Q)\} \cup \{(\theta_X(P),\theta_X(Q)) \mid \R(P,X,Q)\}$. $\tbisim$ is a [rooted] \tb time-out bisimulation, as proven in Appendix \ref{app:time-out}. Let $\tbisim$ be a [rooted] \tb time-out bisimulation, let's define $\R = \{(P,Q) \mid P \tbisim Q\} \cup \{(P,X,Q) \mid \theta_X(P) \tbisim \theta_X(Q)\}$. $\R$ is a [rooted] generalised \tb reactive bisimulation, as proven in Appendix \ref{app:time-out}.
\end{proof}

\noindent
Time-out bisimulations are very practical as there are no triplets to deal with anymore.

\subsection{Congruence}

Until now, bisimilarity was only defined between closed expressions, but any relation ${\sim} \subseteq \closed\times\closed$ can be extended to $\expr\times\expr$ in the following way: $E \sim F$ iff $\forall \rho: V \rightarrow \closed,\; E[\rho] \sim F[\rho]$. It can be extended further to substitutions $\rho, \nu \in V \rightharpoonup \expr$ by $\rho \sim \nu$ iff $\dom{\rho} = \dom{\nu}$ and $\forall x \in \dom{\rho},\; \rho(x) \sim \nu(x)$.

\begin{definition}\rm \label{def:congruence}
    An equivalence ${\sim} \subseteq \expr\times\expr$ is a congruence for an $n$-ary operator $f$ if $P_i \sim Q_i$ for all $i=0,\dots,n{-}1$ implies $f(P_0,...,P_{n-1}) \sim f(Q_0,...,Q_{n-1})$. It is a \emph{lean congruence} if, for all $E \in \expr$ and all $\rho, \nu \in V \rightharpoonup \expr$ such that $\rho \sim \nu$, $E[\rho] \sim E[\nu]$. It is a \emph{full congruence} if 
    \begin{enumerate}
        \item it is a congruence for all operators in the language, and
        \item for all recursive specifications $\equa, \equa'$ with $V_\equa = V_{\equa'}$ and $x \in V_\equa$ such that $\langle x|\equa\rangle,\langle x|\equa'\rangle \in \closed$, if $\forall y \in V_\equa,\; \equa_y \sim \equa'_y$ then $\langle x|\equa\rangle \sim \langle x|\equa'\rangle$.
    \end{enumerate}
\end{definition}

\noindent
To show that $\sim$ is a lean congruence it suffices to restrict attention to closed substitutions $\rho, \nu \in V \rightarrow \closed$, because the general property will then follow by composition of substitutions. A full congruence is a lean congruence, and a lean congruence is a congruence for all operators in the language, but both implications are strict, as shown in \cite{vG17b}.

To show that $\bisimrtbrc$ and $\bisimrb$ are full congruences, it is first necessary to prove that $\bisimtbrc$ and $\bisimb$ are congruences for some of the operators of $\ccsp$.

\begin{proposition} \label{prop:stability}
    $\bisimtbrc$ and $\bisimb$ are congruences for action prefixing, parallel composition, abstraction, renaming and the environment operator $\theta_L^U$, for all $L \subseteq U \subseteq A$.
\end{proposition}

\begin{proof}
    Let $\tbisim$ be the smallest relation such that, for all $P,Q \in \closed$,
    \begin{itemize}
        \item if $P \bisimtbrc Q$ then $P \tbisim Q$;
        \item if $P \tbisim Q$ then, for all $\alpha \in Act$, $I\subseteq A$, $\rename \in A\times A$ and $L \subseteq U \subseteq A$, $\alpha.P \tbisim \alpha.Q$, $\tau_I(P) \tbisim \tau_I(Q)$, $\rename(P) \tbisim \rename(Q)$ and $\theta_L^U(P) \tbisim \theta_L^U(Q)$; 
        \item if $P_1 \tbisim Q_1$, $P_2 \tbisim Q_2$ and $S \subseteq A$ then $P_1 \parallel_S P_2 \tbisim Q_1 \parallel_S Q_2$.
    \end{itemize}
    It suffices to show that $\tbisim$ is a \tb time-out bisimulation up to $\bisim$\,, which implies ${\tbisim}\subseteq{\bisimtbrc}\,$. A bisimulation ``up to'' is a notion introduced by Milner in \cite{Mi90ccs}; it is commonly used when proving congruence properties. The proof uses some lemmas which were obtained in \cite{strongreactivebisimilarity}. Details can be found in Appendix \ref{app:stability}. A similar proof yields the result for $\bisimb$\,. \end{proof}

\begin{theorem} \label{thm:congruence}
    $\bisimrtbrc$ and $\bisimrb$ are full congruences.
\end{theorem}

\begin{proof}
    Let ${\tbisim} \subseteq \closed\times\closed$ be the smallest relation such that 
    \begin{itemize}
        \item if $P \bisimrtbrc Q$ then $P \tbisim Q$;
        \item if $P_1 \tbisim Q_1$ and $P_2 \tbisim Q_2$ then $P_1 + P_2 \tbisim Q_1 + Q_2$ and $\forall S \subseteq A,\; P_1 \parallel_S P_2 \tbisim Q_1 \parallel_S Q_2$;
        \item if $P \tbisim Q$ then $\forall \alpha \mathbin\in Act,\; \alpha.P \tbisim \alpha.Q$, $\forall I \subseteq A,\; \tau_I(P) \tbisim \tau_I(Q)$, $\forall \rename \subseteq A\times A,\; \rename(P) \tbisim \rename(Q)$, $\forall L \mathbin\subseteq U \mathbin\subseteq A,\;\theta_L^U(P) \tbisim \theta_L^U(Q)$ and $\forall X \subseteq A,\; \psi_X(P) \tbisim \psi_X(Q)$;
        \item if $\equa$ is a recursive specification with $z \in V_\equa$ and $\rho, \nu \in V\setminus V_\equa \rightarrow \closed$ are substitutions such that $\forall x \in V\setminus V_\equa,\; \rho(x) \tbisim \nu(x)$, then $\langle z|\equa\rangle[\rho] \tbisim \langle z|\equa\rangle[\nu]$;
        \item if $\equa$ and $\equa'$ are recursive specifications and $x \in V_\equa=V_{\equa'}$ with $\langle x|\equa\rangle, \langle x|\equa'\rangle \in \closed$ such that $\forall y \in V_\equa,\; \equa_y \bisimrtbrc \equa'_y$, then $\langle x|\equa\rangle \tbisim \langle x|\equa'\rangle$.
    \end{itemize}
    Since ${\bisimrtbrc} \subseteq {\tbisim}$, it suffices to prove that $\tbisim$ is a rooted \tb time-out bisimulation up to $\bisimtbrc$\,, as done in Appendix \ref{app:congruence}. This implies ${\tbisim} = {\bisimrtbrc}$\, and the definition will then give us that $\bisimrtbrc$ is a lean congruence. Moreover, the last condition of $\tbisim$ adds that it is a full congruence. A similar proof yields the result for $\bisimrb$\,.
\end{proof}

\section{Axiomatisation} \label{sec:axiom}

We will provide complete axiomatisations for $\bisimrtbrc$ and $\bisimrb$ on various fragments of $\ccsp$.

\subsection{Recursive Principles} \label{subsec:recursive}

The expression $\langle x |\equa\rangle$ is intuitively defined as the $x$-component of the solution of $\equa$. However, $\equa$ could perfectly well have multiple solutions that are not bisimilar to each other. For instance, take $\equa = \{x \mathbin= x\}$; any expression is an $x$-component of a solution of $\equa$. For our complete\linebreak[3] axiomatisation, we need to restrict attention to recursive specifications which have a unique solution with respect to our notion of bisimilarity. This property can be decomposed into two principles \cite{BW90,Fok00}: the \emph{recursive definition principle} (RDP) states that a system of recursive equations has at least one solution and the \emph{recursive specification principle} (RSP) that it has at most one solution. The latter holds under a condition traditionally called \emph{guardedness}.

\begin{definition}\rm \label{def:solution}
    Let $\equa$ be a recursive specification and ${\sim} \subseteq \closed\times\closed$, a \emph{solution up to $\sim$} of $\equa$ is a substitution $\rho \in \expr^{V_\equa}$ such that $\rho \sim \equa[\rho]$. Here $\rho$ and ${\equa} \in\expr^{V_\equa}$ are seen as $V_\equa$-tuples.
\end{definition}

\noindent
In \cite{BW90,Fok00} RDP was proven for the classical notion of strong bisimilarity $\bisim$. Since $\bisimrtbrc$ and $\bisimrb$ are included in $\bisim$, it holds for both of these relations as well.

\begin{proposition}[RDP] \label{prop:rdp}
    Let $\equa$ be a recursive specification. The substitution $\rho: x \mapsto \langle x|\equa\rangle$ for all $x \mathbin\in V_\equa$ is a solution of $\equa$ up to $\bisim$. It is called the \emph{default solution} of $\equa$.
\end{proposition}

\noindent
An occurrence of a variable $x$ in an expression $E$ is \emph{well-guarded} if $x$ occurs in a subexpression $a.F$ of $E$, with $a\mathbin\in A\cup\{\rt\}$. Here we do not allow $\tau$ as a guard, but unlike in \cite{RG24} $\rt$ can be a guard. An expression $E$ is \emph{well-guarded} if no operator $\tau_I$ occurs in $E$ and all free occurrences of variables in $E$ are well-guarded.\linebreak[3] A recursive specification $\equa$ is \emph{manifestly well-guarded} if no operator $\tau_I$ occurs in $\equa$ and for all $x,y\in V_\equa$ all occurrences of $x$ in the expression $\equa_y$ are well-guarded; it is \emph{well-guarded} if it can be made manifestly well-guarded by repeated substitution of $\equa_y$ for $y$ within terms $\equa_x$.\linebreak[3] A $\ccsp$ process $P\in\closed$ is \emph{guarded} if each recursive specification occurring in $E$ is well-guarded. It is \emph{strongly guarded} if moreover there is no infinite path of $\tau$-transitions starting in a state $P_0$ reachable from $P\!$.

\begin{proposition}[RSP] \label{prop:rsp}
    Let $\equa$ be a well-guarded recursive specification and $\rho, \nu \mathbin\in \expr^{V_\equa}\!\!$. If $\rho$ and $\nu$ are solutions of $\equa$ up to $\bisimrtbrc$ (or $\bisimrb$) then $\rho \bisimrtbrc \nu$ (resp.\ $\rho \bisimrb \nu$).
\end{proposition}

\begin{proof}
    Modifying $\equa$ by substituting $\equa_y$ for $y$ within terms $\equa_x$ with $x,y\in V_\equa$ does not affect the set of its solutions. Hence we can restrict attention to manifestly well-guarded $\equa$.

    Thanks to the composition of substitutions, it suffices to prove the proposition when $\rho, \nu \in \closed^{V_\equa}$ and only variables of $V_\equa$ can occur in $\equa_x$ for $x \in V_\equa$. It suffices to show that the symmetric closure of \({\tbisim} := \{(H[\equa[\rho]],H[\equa[\nu]]) \mid H \in \expr\) is without $\tau_I$ operators and with free variables from $V_\equa \mbox{ only}\}$ is a rooted \tb time-out bisimulation up to $\bisimtbrc$\,. Here $\equa[\rho]\in\closed^{V_\equa}$ is seen as a substitution. Details can be found in Appendix \ref{app:RSP}. An almost identical strategy can be applied to get RSP for $\bisimrb$\,.
\end{proof}

\subsection{Axioms and Soundness} \label{subsec:soundness}

The set of axioms provided is composed of the axiomatisation of $\bisim_r$ \cite{strongreactivebisimilarity}, together with the \emph{branching axiom} which is well-known since it is used in the axiomatisation of rooted branching bisimilarity \cite{branching}.

\renewcommand{\arraystretch}{1.2}
\begin{table}[ht]
    \centering
    \begin{tabular}{|l l l|}
        \hline
        $x+(y+z) = (x+y)+z$ & $\tau_I(x+y) = \tau_I(x) + \tau_I(y)$ & $\rename(x+y) = \rename(x) + \rename(y)$ \\
        $x+y = y+x$ & $\tau_I(\alpha.x) = \alpha.\tau_I(x)$ if $\alpha \not\in I$ & $\rename(\tau.x) = \tau.\rename(x)$\\
        $x+x = 0$ & $\tau_I(\alpha.x) = \tau.\tau_I(x)$ if $\alpha \in I$ & $\rename(\rt.x) = \rt.\rename(x)$ \\
        $x+0 = x$ & & $\rename(a.x) = \sum_{\{b \mid \rename(a,b)\}}b.\rename(x)$ \\
        \multicolumn{3}{|l|}{\textbf{Expansion Theorem:} if $P = \sum\limits_{i \in I}\alpha_i.P_i$ and $Q = \sum\limits_{j \in J} \beta_j.Q_j$ then} \\
        \multicolumn{3}{|c|}{$P \parallel_S Q = \sum\limits_{i \in I,\alpha_i \not\in S}(\alpha_i.P_i \parallel_S Q) + \sum\limits_{j \in J, \beta_j \not\in S}(P \parallel_S \beta_j.Q_j) + \sum\limits_{i \in I, j \in J, \alpha_i = \beta_j \in S}\alpha_i.(P_i \parallel_S Q_j)$} \\
        \multicolumn{3}{|c|}{$\alpha.(\tau.(x+y)+x) = \alpha.(x+y)$ \quad\textbf{(Branching Axiom)}} \\
        $\langle x |\equa\rangle = \langle \equa_x |\equa\rangle$ \quad\textbf{(RDP)} & \multicolumn{2}{c|}{$\equa \Rightarrow x = \langle x | \equa\rangle$ \quad with $\equa$ well-guarded\quad\textbf{(RSP)}} \\
        \hline
        \hline
        \multicolumn{2}{|l}{$\theta_L^U(\sum_{i \in I}\alpha_i.x_i) = \sum_{i\in I} \alpha_i.x_i$} & $(\forall i \in I, \alpha_i \not\in L\cup\{\tau\})$ \\
        \multicolumn{2}{|l}{$\theta_L^U(x + \alpha.y + \beta.z) = \theta_L^U(x + \alpha.y)$} & $(\alpha \in L\cup\{\tau\} \wedge \beta \not\in U\cup\{\tau\})$ \\
        \multicolumn{2}{|l}{$\theta_L^U(x + \alpha.y + \beta.z) = \theta^U_L(x + \alpha.y) + \theta_L^U(\beta.z)$} & $(\alpha \in L\cup\{\tau\} \wedge \beta\in U\cup\{\tau\})$ \\
        \multicolumn{2}{|l}{$\theta_L^U(\alpha.x) = \alpha.x$} & $(\alpha \neq \tau)$ \\
        $\theta_L^U(\tau.x) = \tau.\theta_L^U(x)$ & & \\
        \multicolumn{2}{|l}{$\psi_X(x + \alpha.y) = \psi_X(x) + \alpha.y$} & $(\alpha \not\in X\cup\{\tau,t\})$ \\
        \multicolumn{2}{|l}{$\psi_X(x + \alpha.y + \rt.z) = \psi_X(x + \alpha.y)$} & $(\alpha \in X\cup\{\tau\})$ \\
        \multicolumn{2}{|l}{$\psi_X(x + \alpha.y + \beta.z) = \psi_X(x + \alpha.y) + \beta.z$} & $(\alpha, \beta \in X\cup\{\tau\})$ \\
        \multicolumn{2}{|l}{$\psi_X(\alpha.x) = \alpha.x$} & $(\alpha \neq \rt)$ \\
        $\psi_X(\sum_{i \in I}\rt.y_i) = \sum_{i \in I}\rt.\theta_X(y_i)$ & & \\
        \hline
        \hline
        \multicolumn{3}{|c|}{$(\forall X \subseteq A,\; \psi_X(x) = \psi_X(y)) \Rightarrow x= y$ \quad\textbf{(Reactive Approximation Axiom)}} \\
        \hline
    \end{tabular}
\vspace{2ex}
    \caption{Axiomatisation of $\bisimrtbrc$ and $\bisimrb$}
    \label{tab:axioms}
\end{table}

Let $Ax^\infty$ be the set of all axioms in the first two rectangles in Table \ref{tab:axioms} and $Ax := Ax^\infty \setminus \{\mbox{RDP},\mbox{RSP}\}$. Let $Ax^\infty_r$ be the set of all axioms in Table \ref{tab:axioms} and $Ax_r := Ax^\infty_r \setminus \{\mbox{RDP},\mbox{RSP}\}$. The law \hypertarget{Lt}{$\mbox{\bf L}\tau$: $\tau.x + \rt.y = \tau.x$} can be derived from the reactive approximation axiom \cite{strongreactivebisimilarity}.

\begin{proposition} \label{prop:soundness}
    Let $P,Q$ be two $\ccsp$ processes.
    \begin{itemize}
        \item If $Ax^\infty \vdash P = Q$ then $P \bisimrb Q$.
        \item If $Ax^\infty_r \vdash P = Q$ then $P \bisimrtbrc Q$.
    \end{itemize}
\end{proposition}

\begin{proof}
    Since $\bisimrb$ and $\bisimrtbrc$ are congruences, it suffices to prove that each axiom is sound, meaning that replacing, in each axiom, $=$ by the desired bisimilarity and each variable by a process produces a true statement. Most of these axioms were proven to be sound for the classical notion $\bisim$ of strong bisimilarity \cite{Mi90ccs} in \cite{strongreactivebisimilarity}. Thus, since both $\bisimrb$ and $\bisimrtbrc$ are included in $\bisim$, most of them are sound for $\bisimrb$ and $\bisimrtbrc$\,. 
    
    Only the branching axioms, RSP and the reactive approximation axiom remain to be proven sound. The soundness of the branching axioms is trivial and the soundness of RSP is exactly Proposition \ref{prop:rsp}. For the reactive approximation axiom, it suffices to show that ${\tbisim} := \; \bisimrtbrc \cup \{(P,Q),(Q,P) \mid \forall X \subseteq A, \psi_X(P) \bisimrtbrc \psi_X(Q)\}$ is a rooted \tb time-out bisimulation, as done in Appendix~\ref{app:RA}.
\end{proof}

\subsection{Completeness}

A well-known feature of most process algebras is that the standard collection of axioms allows one to bring any {guarded} process expression in the following normal form \cite{BW90,Fok00}.

\begin{definition}\rm \label{def:head-normal form}
    Let $P$ be a {guarded} $\ccsp$ process. The \emph{head-normal form} of $P$ is $\hat{P} := \sum_{\{(\alpha,Q) \mid P\step{\alpha}Q\}}\alpha.Q$.
\end{definition}

\noindent
In \cite{strongreactivebisimilarity}, it is proven that the axiomatisation of $\bisim_r$ enables one to equate any {guarded} process with its head-normal form (using a definition of guardedness that is more liberal than the one employed here, with $\tau$ allowed as a guard). Since the axiomatisation of $\bisim_r$ is included in $Ax^\infty$ and $Ax^\infty_r$, this yields the property for them as well.

\begin{lemma} \label{lem:head-normal form}
    Let $P$ be a guarded $\ccsp$ process. Then $Ax^\infty \vdash P = \hat{P}$ and $Ax^\infty_r \vdash P = \hat{P}$. Moreover, $Ax$ or $Ax_r$ are sufficient if $P$ is recursion-free.
    \qed
\end{lemma}

\noindent
This lemma is used extensively in the proof of the following completeness results.

\begin{proposition} \label{prop:collapse}
    Let $P,Q$ be two recursion-free $\ccsp$ processes. If $P \bisimtbrc Q$ (resp.\ $P \bisimb Q$) then, for all $\alpha \in Act$, $Ax_r \vdash \alpha.\hat{P} = \alpha.\hat{Q}$ (resp.\ $Ax \vdash \alpha.\hat{P} = \alpha.\hat{Q}$).
\end{proposition}

\begin{proof}
    The \emph{depth} $d(p)$ of a process $P$ is the length of the longest path starting from $P$. Note that it is properly defined for recursion-free processes only. The proof proceeds by induction on $\max(d(P),d(Q))$. The technique is fairly standard and the details can be found in Appendix \ref{app:completeness finite}.
\end{proof}

\begin{theorem} \label{thm:completeness finite}
    Let $P,Q$ be two recursion-free $\ccsp$ processes. If $P \bisimrtbrc Q$ (resp.\ $P \bisimrb Q$) then $Ax_r \vdash P = Q$ (resp.\ $Ax \vdash P = Q$).
\end{theorem}

\begin{proof}
    It suffices to express both processes in their head-normal form and then to equate each pair of matching branches using Proposition \ref{prop:collapse}. Details are in Appendix \ref{app:completeness finite}.
\end{proof}

\noindent
The following theorem lifts this result for $\bisimrb$ from finite (recursion-free) processes to arbitrary (infinite) ones, subject to the restriction of strong guardedness.

\begin{theorem} \label{thm:completeness}
    Let $P,Q$ be strongly guarded $\ccsp$ processes.\\ If $P \bisimrb Q$ then $Ax^\infty \vdash {P} = {Q}$.
\end{theorem}

\begin{proof}
A well-known technique called \emph{equation merging} can be applied. Details can be found in Appendix \ref{app:completeness}.
\end{proof}

\subsection{Canonical Representative} \label{subsec:canonical}

Unfortunately, equation merging does not work on reactive bisimulations \cite{strongreactivebisimilarity}. Thus, another technique is used \cite{GF20,LY20}, called \emph{canonical representatives}. The idea is to build the simplest process for each equivalence class of $\bisimrtbrc$ and use them as intermediary to equate processes.

Let us denote with $\closed^g$ the strongly guarded fragment of $\closed$. For all $P \in \closed^g$, $[P] := \{Q \in \closed^g \mid P \bisimtbrc Q\}$ is the $\bisimtbrc$-equivalence class of $P$. $[\closed^g]$ denotes the set of all $\bisimtbrc$-equivalence classes. Using the axiom of choice, a choice function $\chi: [\closed^g] \rightarrow \closed^g$ can be defined such that $\forall R \in [\closed^g], \chi(R) \in R$. A transition relation can be defined between $\bisimtbrc$-equivalence classes:
\begin{align*}
    \forall \alpha \in A_\tau, (R \step{\alpha} R' \Leftrightarrow\; & \chi(R) \pathtau P_1 \step{\alpha} P_2 \wedge P_1\in R \wedge P_2\in R' \wedge (\alpha \in A \vee R \ne R')) \\
    R \step{\rt} R' \Leftrightarrow\; & \chi(R) \pathtau P_1 \step{\rt} P_2 \wedge P_1 \in R \wedge {P_1\!\nsteptau} \wedge P_2 \in R' \\
\end{align*}

\noindent
All bisimulations can be extended to $\bisimtbrc$-equivalence classes. It suffices to consider the set of states $\closed^g \uplus [\closed^g] \uplus \{\theta_X([P]) \mid X \subseteq A \wedge P\in \closed^g\}$.

\begin{proposition} \label{prop:class}
    Let $P \in \closed^g$, $P \bisimtbrc[] [P]$.
\end{proposition}

\begin{proof}
    It suffices to prove that $\tbisim := \{(P,[P]),([P],P) \mid P \in \closed^g\}$ is a \tb time-out bisimulation up to $\bisimtbrc$\,. Details can be found in Appendix \ref{app:canonical rep}.
\end{proof}

\begin{definition}\rm \label{def:canonical representative}
    Let $P, Q \in \closed^g$, the \emph{canonical representative} of $P$ and $Q$ is a recursive specification $\equa$ such that $V_{\equa} := \{x_P,x_Q\} \cup \{x_R \mid R \in \bigcup_{P' \in \reach{P} \cup \reach{Q}}\reach{[P']}\}$, and $\forall R \in \bigcup_{P' \in \reach{P}\cup \reach{Q}}\reach{[P']}$,
    \begin{align*}
        \equa_{x_P} :=  \hspace{-10pt}\sum_{\{(\alpha,P') \mid P \step{\alpha} P'\}} \hspace{-10pt}\alpha.x_{[P']} \mbox{ ; } \equa_{x_Q} := \hspace{-10pt}\sum_{\{(\alpha,Q') \mid Q \step{\alpha} Q'\}} \hspace{-10pt}\alpha.x_{[Q']} \mbox{ and } \equa_{x_R} :=  \hspace{-10pt}\sum_{\{(\alpha,R') \mid R \step{\alpha} R'\}} \hspace{-10pt}\alpha.x_{R'}
    \end{align*}
\end{definition}

\noindent
The canonical representative is well-defined since $P$, $Q$, as well as processes $[P']\in[\closed^g]$ are finitely branching~\cite{strongreactivebisimilarity}. Additionally, $\bigcup_{P' \in \reach{P}\cup\reach{Q}}\reach{[P']}$ is countable. Moreover, $\equa$ is strongly guarded. Furthermore, \hypertarget{recall}{by construction $\langle x_R |\equa\rangle \bisim R$} for all $R \in \bigcup_{P' \in \reach{P}\cup\reach{Q}}\reach{[P']}$. 

\begin{proposition} \label{prop:canonical}
    Let $P,Q \in \closed^g$ and $\equa$ be the canonical representative of $P$ and $Q$. $Ax^\infty_r \vdash P = \langle x_P |\equa\rangle$.
\end{proposition}

\begin{proof}
  It suffices to show that $P$ and $\langle x_P|\equa\rangle$ are $y_P$-components of solutions of
  $\{y_{P^\dag} = \sum_{\{(\alpha,P^\ddag) \mid P^\dag\step{\alpha} P^\ddag\}}\alpha.y_{P^\ddag} \mid P^\dag \in \reach{P}\}$.
  Details can be found in Appendix \ref{app:canonical}.
\end{proof}

\begin{theorem} \label{thm:canonical}
    Let $P,Q \in \closed^g$. If $P \bisimrtbrc Q$ then $Ax^\infty_r \vdash P = Q$.
\end{theorem}

\begin{proof}
    It suffices to equate $\langle x_P|\equa\rangle$ and $\langle x_Q|\equa\rangle$ using RDP and the reactive approximation axiom. Details can be found in Appendix \ref{app:canonical}.
\end{proof}

\section*{Conclusion}

This paper defined a form of branching bisimilarity for processes with time-out transitions, and provided a modal characterisation, congruence results, and a complete axiomatisation for strongly {guarded} processes. Whereas the bisimilarity presented here treats the time-out action $\rt$ more or less as a visible action, in \cite{RG24} we propose a variant that elides time-outs, by satisfying laws like $a.\rt.b.0 = a.\rt.\rt.b.0$, just like branching bisimilarity elides $\tau$-transitions. We obtained the present paper from  \cite{RG24} by systemically suppressing this eliding feature, thereby obtaining simpler definitions and proofs.

A topic for future work is to combine this work with the ideas behind \emph{justness} \cite{GH19}, a weaker form of fairness that allows the formulation and derivation of useful liveness properties. In a setting with time-outs, justness would demand that once a parallel component reaches a state in which a time-out transition is enabled, it cannot stay in that state forever after.

\bibliography{biblio}

\input{appendix_CONCUR}

\input{appendix}

\end{document}

%% file: appendix_CONCUR.tex
\newpage
\appendix

\section{Examples} \label{app:examples}

\paragraph*{Scope of the First Clause of Definition \ref{def:intuitive}}

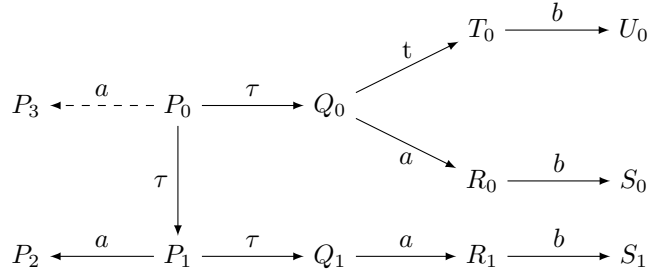
\begin{figure}[ht]
    \centering
    \begin{tikzpicture}
        \node (a) at (0,0) {$P_0$};
        \node (b) at (0,-2) {$P_1$};
        \node (c) at (-2,-2) {$P_2$};
        \node (c1) at (-2,0) {$P_3$};
        \node (d) at (2,0) {$Q_0$};
        \node (e) at (4,1) {$T_0$};
        \node (e1) at (6,1) {$U_0$};
        \node (f) at (4,-1) {$R_0$};
        \node (g) at (6,-1) {$S_0$};
        \node (d1) at (2,-2) {$Q_1$};
        \node (f1) at (4,-2) {$R_1$};
        \node (g1) at (6,-2) {$S_1$};

        \draw[->, >=latex] (a) -- node[midway,left]{$\tau$} (b);
        \draw[->, >=latex, dashed] (a) -- node[midway,above]{$a$} (c1);
        \draw[->, >=latex] (b) -- node[midway,above]{$a$} (c);
        \draw[->, >=latex] (a) -- node[midway,above]{$\tau$} (d);
        \draw[->, >=latex] (d) -- node[midway,above]{$\rt$} (e);
        \draw[->, >=latex] (e) -- node[midway,above]{$b$} (e1);
        \draw[->, >=latex] (d) -- node[midway,below]{$a$} (f);
        \draw[->, >=latex] (f) -- node[midway,above]{$b$} (g);
        \draw[->, >=latex] (b) -- node[midway,above]{$\tau$} (d1);
        \draw[->, >=latex] (d1) -- node[midway,above] {$a$} (f1);
        \draw[->, >=latex] (f1) -- node[midway,above]{$b$} (g1);
    \end{tikzpicture}
    \caption{Counter-Example to a Naive Clause 1.a.}
    \label{fig:counter-example 1a}
\end{figure}

\noindent
In Figure \ref{fig:counter-example 1a}, the process $a.0 + \tau.(\rt.b.0 + a.b.0) + \tau.(\tau.a.b.0 + a.0)$ is represented as an LTS\@. Let $A := \{a,b\}$. Removing the dashed $a$-transition generates the process $\tau.(\rt.b.0 + a.b.0) + \tau.(\tau.a.b.0 + a.0)$.

First, we are going to show that these two processes are not \tb reactive bisimilar. Let's try to build a \tb reactive bisimulation between them. The only way to match the dashed $a$-transition of $a.0 + \tau.(\rt.0 + a.b.0) + \tau.(\tau.a.b.0 + a.0)$ is by the $a$-transition between $P_1$ and $P_2$, because all other $a$-transitions are followed by a $b$-transition. This requires to elide the $\tau$-transition between $P_0$ and $P_1$, who must be \tb reactive bisimilar. Since $P_0 \bisimtbrc P_1$, when considering the $\tau$-transition between $P_0$ and $Q_0$, $Q_0$ has to be \tb reactive bisimilar to $P_1$ or $Q_1$. Now, the $a$-transition between $Q_0$ and $R_0$ has to be matched by the $a$-transition between $Q_1$ and $R_1$ because of the following $b$-transition. This implies $Q_0 \bisimtbrc Q_1$, thus, $Q_0 \bisimtbrc[\emptyset] Q_1$. One has $\deadend{Q_0}{\emptyset}$ and $Q_0 \step{\rt} T_0$, i.e., when the environment temporary allows no visible actions, $Q_0$ can time-out into a state in which $b$ is possible. This behaviour cannot be matched by $Q_1$---a contradiction.

Now, consider the alternative to Definition \ref{def:intuitive} in which the first clause has been changed to 
\begin{enumerate}
    \item \begin{enumerate}
        \item if $P \steptau P'$ then there is a path $Q \pathtau Q_1 \step{\opt{\tau}} Q_2$ with $\R(P,Q_1)$ and $\R(P',Q_2)$.
    \end{enumerate}
\end{enumerate}
In other words, the scope of the first clause is restricted to $\tau$-transitions. This modification enables building a bisimulation between the two processes. Indeed, the dashed $a$-transition is only considered when the environment allows $a$. Thus, it is sufficient to get $P_0 \bisimtbrc[A] P_1$ and $P_0 \bisimtbrc[\{a\}] P_1$ and not $P_0 \bisimtbrc P_1$ anymore. Therefore, it is sufficient to match $Q_0$ and $Q_1$ in environments allowing $a$. As a result, the outgoing time-out transition of $Q_0$ is never considered when matching $Q_0$ with $Q_1$, solving our previous issue. Once this observation is made, building the bisimulation is trivial.

Finally, place both processes in the context $\_\!\_ \parallel_{\{a\}} (\tau.0 + a.0)$. It behaves like a one-way switch enabling to block all $a$-transitions forever as soon as the $\tau$-transition is performed. Let's try to build a \tb reactive bisimulation between the two processes. Following the same reasoning as before, it is necessary to get $P_0 \parallel_{\{a\}} (\tau.0 + a.0) \bisimtbrc[A] P_1 \parallel_{\{a\}} (\tau.0 + a.0)$ because of the dashed $a$-transition, and then $Q_0 \parallel_{\{a\}} (\tau.0 + a.0) \bisimtbrc[A] Q_1 \parallel_{\{a\}} (\tau.0 + a.0)$ because of the $a$-transition between $Q_0$ and $R_0$. Note that $Q_0 \parallel_{\{a\}} (\tau.0 + a.0) \steptau Q_0 \parallel_{\{a\}} 0\linebreak[2] \step{\rt} T_0 \parallel_{\{a\}} 0 \step{b} U_0 \parallel_{\{a\}} 0$ and $\deadend{Q_0 \parallel_{\{a\}} 0}{A}$. As before, $Q_0 \parallel_{\{a\}} (\tau.0 + a.0)$ can time-out into a state in which $b$ is executable, whereas this behaviour is impossible in $Q_1 \parallel_{\{a\}} (\tau.0 + a.0)$. As a result, restricting the scope of the first clause of Definition \ref{def:intuitive} to $\tau$-transitions prevents $\bisimtbrc$ from being a congruence for parallel composition.

\paragraph*{Necessity of the Stability Respecting Clause}

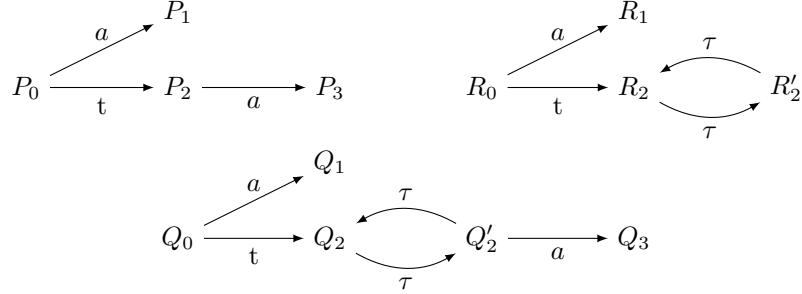
\begin{figure}[ht]
    \centering
    \begin{tikzpicture}
        \node (a) at (0,0) {$P_0$};
        \node (b) at (2,0) {$P_2$};
        \node (c) at (2,1) {$P_1$};
        \node (d) at (4,0) {$P_3$};
        \draw[->, >=latex] (a) -- node[midway,below]{$\rt$} (b);
        \draw[->, >=latex] (a) -- node[midway,above]{$a$} (c);
        \draw[->, >=latex] (b) -- node[midway,below]{$a$} (d);
        \node (a) at (2,-2) {$Q_0$};
        \node (b) at (4,-2) {$Q_2$};
        \node (b') at (6,-2) {$Q'_2$};
        \node (c) at (4,-1) {$Q_1$};
        \node (d) at (8,-2) {$Q_3$};
        \draw[->, >=latex] (a) -- node[midway,below]{$\rt$} (b);
        \draw[->, >=latex] (a) -- node[midway,above]{$a$} (c);
        \draw[->, >=latex] (b) to[bend right] node[midway,below]{$\tau$} (b');
        \draw[->, >=latex] (b') to[bend right] node[midway,above]{$\tau$} (b);
        \draw[->, >=latex] (b') -- node[midway,below]{$a$} (d);
        \node (a) at (6,0) {$R_0$};
        \node (b) at (8,0) {$R_2$};
        \node (b') at (10,0) {$R'_2$};
        \node (c) at (8,1) {$R_1$};
        \draw[->, >=latex] (a) -- node[midway,below]{$\rt$} (b);
        \draw[->, >=latex] (a) -- node[midway,above]{$a$} (c);
        \draw[->, >=latex] (b) to[bend right] node[midway,below]{$\tau$} (b');
        \draw[->, >=latex] (b') to[bend right] node[midway,above]{$\tau$} (b);
    \end{tikzpicture}
    \caption{Counter-Example to the Absence of a Stability Respecting Clause}
    \label{fig:example stability}
\end{figure}

\noindent
In Figure \ref{fig:example stability}, three processes are represented as LTSs. Take $A := \{a\}$. According to Definition~\ref{def:intuitive}, $\neg(P_0 \bisimtbrc Q_0)$ and $Q_0 \bisimtbrc R_0$. 

Let's try to build a \tb reactive bisimulation between the top-left and bottom processes. Matching the time-out between $Q_0$ and $Q_2$ implies that $Q_2 \bisimtbrc[\emptyset] P_0$ or $Q_2 \bisimtbrc[\emptyset] P_2$. However, $P_0 \nsteptau$ and $P_2 \nsteptau$, thus, there should be a path $Q_2 \pathtau Q_2' \nsteptau$, but this is not the case.

The symmetric closure of
\begin{align*}
    \R := \{(Q_0,R_0), (Q_1,R_1), (Q_2,\emptyset,R_2), (Q_2',\emptyset,R_2')\} \cup \{(Q_0,X,R_0),(Q_1,X,R_1) \mid X \subseteq A\}
\end{align*}
is a \tb reactive bisimulation. The $a$-transition between $Q_2'$ and $Q_3$ does not have to be matched since $Q_2'$ is considered only when the environment disallows $a$.

Now, suppose that the stability respecting condition is removed from Definition \ref{def:intuitive}. As a result, a \tb reactive bisimulation can be built between the top-left and bottom processes. The symmetric closure of
\begin{align*}
    \R' := &  \{(P_0,Q_0),(P_1,Q_1),(P_2,Q_2),(P_2,Q_2'),(P_3,Q_3)\} \\
    & \cup \{(P_0,X,Q_0),(P_1,X,Q_1),(P_2,X,Q_2),(P_2,X,Q_2'),(P_3,X,Q_3) \mid X \subseteq A\}
\end{align*}
would be a \tb reactive bisimulation. Moreover, $\R$ would still be a \tb reactive bisimulation, since Definition \ref{def:intuitive} has merely been weakened. Therefore, according to the modified Definition \ref{def:intuitive}, $P_0 \bisimtbrc Q_0$ and $Q_0 \bisimtbrc R_0$. However, when trying to construct a \tb reactive bisimulation between $P_0$ and $R_0$, because of the time-out transition, $R_2$ has to be matched to $P_0$ or $P_2$ and no $a$-transition is reachable from $R_2$; therefore, $\neg(P_0 \bisimtbrc R_0)$. As a result, removing the stability respecting clause from Definition~\ref{def:intuitive} prevents $\bisimtbrc$ from being an equivalence relation.

\section{Generalised concrete branching reactive bisimulation} \label{app:gbrb}

The second clause of Definition~\ref{def:intuitive} is quite tedious to check; thus, an equivalent definition of the bisimilarity would be useful. Actually, it is possible to define the exact same notion in a more general way at the cost of some clear motivations.

\begin{definition}\rm \label{def:generalised}
    A \emph{generalised \tb reactive bisimulation} is a symmetric relation $\R \subseteq (\closed\times\closed)\cup(\closed\times\mathcal{P}(A)\times\closed)$ such that, for all $P,Q \in \closed$ and $X \subseteq A$,
    \begin{enumerate}
        \item if $\R(P,Q)$
        \begin{enumerate}
            \item if $P \step{\alpha} P'$ with $\alpha \in A_\tau$ then there is a path $Q \pathtau Q_1 \step{\opt{\alpha}} Q_2$ with $\R(P,Q_1)$ and $\R(P',Q_2)$,
            \item if $\deadend{P}{X}$ and $P \step{\rt} P'$ then there is a path $Q \pathtau Q_1 \step{\rt} Q_2$ with $\R(P',X,Q_2)$,
            \item if $P \nsteptau$ then there exists a path $Q \pathtau Q_0 \nsteptau$;
        \end{enumerate}
        \item if $\R(P,X,Q)$
        \begin{enumerate}
            \item if $P \steptau P'$ then there is a path $Q \pathtau Q_1 \step{\opt{\tau}} Q_2$ with $\R(P,X,Q_1)$ and $\R(P',X,Q_2)$,
            \item if $P \step{a} P'$ with $a \in X \vee \deadend{P}{X}$ then there is a path $Q \pathtau Q_1 \step{a} Q_2$ with $\R(P,X,Q_1)$ and $\R(P',Q_2)$,
            \item if $\deadend{P}{(X\cup Y)}$ and $P \step{\rt} P'$ then there is a path $Q \pathtau Q_1 \step{\rt} Q_2$ with $\R(P',Y,Q_2)$,
            \item if $P \nsteptau$ then there is a path $Q \pathtau Q_0 \nsteptau$.
        \end{enumerate}
    \end{enumerate}
\end{definition}

\noindent
The strong point of the generalised definitions is the restriction on the use of triplets, making use of them only after performing a time-out. A generalised version of rooted \tb reactive bisimulation can be defined in a similar fashion.

\begin{definition}\rm \label{def:generalised rooted}
    A \emph{generalised rooted \tb reactive bisimulation} is a symmetric relation $\R \subseteq (\closed\times\closed)\cup(\closed\times\mathcal{P}(A)\times\closed)$ such that, for all $P,Q \in \closed$ and $X \subseteq A$,
    \begin{enumerate}
        \item if $\R(P,Q)$
        \begin{enumerate}
            \item if $P \step{\alpha} P'$ with $\alpha \in A_\tau$ then there is a transition $Q \step{\alpha} Q'$ such that $P' \bisimtbrc Q'$,
            \item if $\deadend{P}{X}$ and $P \step{\rt} P'$ then there is a transition $Q \step{\rt} Q'$ with $P' \bisimtbrc[X] Q'$,
        \end{enumerate}
        \item if $\R(P,X,Q)$
        \begin{enumerate}
            \item if $P \steptau P'$ then there is a transition $Q \steptau Q'$ such that $P' \bisimtbrc[X] Q'$,
            \item if $P \step{a} P'$ with $a \in X \vee \deadend{P}{X}$ then there is a transition $Q \step{a} Q'$ such that $P' \bisimtbrc Q'$,
            \item if $\deadend{P}{(X\cup Y)}$ and $P \step{\rt} P'$ then there is a transition $Q \step{\rt} Q'$ such that $P' \bisimtbrc[Y] Q'$.
        \end{enumerate}
    \end{enumerate}
\end{definition}

\noindent
Note that if a system has no time-out, then a generalised [rooted] \tb reactive bisimulation is a stability respecting [rooted] branching bisimulation, thus proving that [rooted] \tb reactive bisimilarity is indeed an extension of stability respecting [rooted] branching bisimilarity to reactive systems with time-outs.

\begin{proposition} \label{prop:generalised}
    Let $P, Q \in \closed$ and $X \subseteq A$,
    \begin{itemize}
        \item $P \bisimtbrc Q$ (resp.\ $P \bisimtbrc[X] Q$) iff there exists a generalised \tb reactive bisimulation $\R$ with $\R(P,Q)$ (resp.\ $\R(P,X,Q)$),
        \item $P \bisimrtbrc Q$ (resp.\ $P \bisimrtbrc[X] Q$) iff there exists a rooted generalised \tb reactive bisimulation $\R$ with $\R(P,Q)$ (resp.\ $\R(P,X,Q)$).
    \end{itemize}
\end{proposition}

\begin{proof}
    Let $\R$ be a \tb reactive bisimulation. Let's check that it is a generalised \tb reactive bisimulation. Let $P,Q \in \closed$ and $X \subseteq A$.
    \begin{enumerate}
        \item If $\R(P,Q)$
        \begin{enumerate}
            \item this condition is shared by both definitions
            \item if $\deadend{P}{X}$ and $P \step{\rt} P'$ then, since $\R(P,Q)$, $\R(P,X,Q)$. Since $\deadend{P}{X}$ and $P \step{\rt} P'$, there exists a path $Q \pathtau Q_1 \step{\rt} Q_2$ with $\R(P',X,Q_2)$.
            \item if $P \nsteptau$ then, since $\R(P,Q)$, $\R(P,\emptyset,Q)$, so there exists a path $Q \pathtau Q_0 \nsteptau$.
        \end{enumerate}
        \item If $\R(P,X,Q)$
        \begin{enumerate}
            \item this condition is shared by both definitions
            \item if $P \step{a} P'$ with $a \in X \vee \deadend{P}{X}$ then if $a \in X$ then there exists a path $Q \pathtau Q_1 \step{a} Q_2$ such that $\R(P,X,Q_1)$ and $\R(P',Q_2)$. Otherwise, $\deadend{P}{X}$ and so $P \nsteptau$, thus there exists a path $Q \pathtau Q_1 \nsteptau$. Since $\R(P,X,Q)$, $P \nsteptau$ and $Q \pathtau Q_1$, $\R(P,X,Q_1)$. As $\deadend{P}{X}$ and $Q_1 \nsteptau$, $\R(P,Q_1)$. Because $P \step{a} P_1$ and $Q_1 \nsteptau$, there exists a transition $Q_1 \step{a} Q_2$ such that $\R(P',Q_2)$. As a result, there exists a path $Q \pathtau Q_1 \step{a} Q_2$ such that $\R(P,X,Q_1)$ and $\R(P',Q_2)$.
            \item if $\deadend{P}{(X\cup Y)}$ and $P \step{\rt} P'$ then, since $P \nsteptau$, there exists a path $Q \pathtau Q_1 \nsteptau$. Furthermore, using Lemma \ref{lem:obvious}.1, $\R(P,X,Q_1)$. Moreover, $\deadend{P}{X}$ and $Q_1 \nsteptau$, thus, $\R(P,Q_1)$ and so $\R(P,Y,Q_1)$. Since $\deadend{P}{Y}$ and $P \step{\rt} P'$, there exists a path $Q_1 \pathtau Q_1' \step{\rt} Q_2$ with $\R(P',Y,Q_2)$. Since $Q_1 \nsteptau$, $Q_1 = Q_1'$. As a result, there exists a path $Q \pathtau Q_1 \step{\rt} Q_2$ with $\R(P',Y,Q_2)$.
            \item this condition is shared by both definitions.
        \end{enumerate}
    \end{enumerate}
    Let $\R$ be a generalised \tb reactive bisimulation and define 
    \begin{align*}
        \R' := \R & \cup \{(P,X,Q) \mid \R(P,Q) \wedge X \subseteq A\} \cup \{(P,Y,Q),(P,Q) \mid \exists X \subseteq A,\; \R(P,X,Q) \\
        & \wedge (\init{P}\cup\init{Q})\cap(X\cup\{\tau\}) = \emptyset \wedge Y \subseteq A\}
    \end{align*}
    $\R'$ is symmetric by definition. Let's check that $\R'$ is a \tb reactive bisimulation. Let $P,Q \in \closed$ and $X \subseteq A$.
    \begin{enumerate}
        \item If $\R'(P,Q)$ then $\R(P,Q)$ or there exists a set $Y \subseteq A$ such that $\R(P,Y,Q)$ and $(\init{P}\cup\init{Q})\cap(Y\cup\{\tau\}) = \emptyset$.
        \begin{enumerate}
            \item If $P \step{\alpha} P'$ then
            \begin{itemize}
                \item if $\R(P,Q)$ then there exists a path $Q \pathtau Q_1 \step{\opt{\alpha}} Q_2$ such that $\R(P,Q_1)$ and $\R(P,Q_2)$ and, since $\R \subseteq \R'$, $\R'(P,Q_1)$ and $\R'(P',Q_2)$
                \item if there exists $Y \subseteq A$ such that $\R(P,Y,Q)$ and $(\init{P}\cup\init{Q})\cap(Y\cup\{\tau\}) = \emptyset$ then, since $\R(P,Y,Q)$ and $\deadend{P}{Y}$, $\alpha \ne \tau$, so there exists a path $Q \pathtau Q_1 \step{\alpha} Q_2$ such that $\R(P,Y,Q_1)$ and $\R(P,Q_2)$. Since $\deadend{Q}{Y}$ and $\R\subseteq \R'$, $Q = Q_1$ so there exists a path $Q \step{\alpha} Q_2$ such that $\R'(P,Q)$ and $\R'(P',Q_2)$.
            \end{itemize}
            \item For all $Z \subseteq A$,
            \begin{itemize}
                \item if $\R(P,Q)$ then, by definition of $\R'$, $\R'(P,Z,Q)$
                \item if there exists $Y \subseteq A$ such that $\R(P,Y,Q)$ and $(\init{P}\cup\init{Q})\cap(Y\cup\{\tau\}) = \emptyset$ then, by definition of $\R'$, $\R'(P,Z,Q)$.
            \end{itemize}
        \end{enumerate}
        \item If $\R'(P,X,Q)$ then $\R(P,X,Q)$, or $\R(P,Q)$, or there exists $Y \subseteq A$ such that $\R(P,Y,Q)$ and $(\init{P}\cup\init{Q})\cap(Y\cup\{\tau\}) = \emptyset$.
        \begin{enumerate}
            \item If $P \steptau P'$ then
            \begin{itemize}
                \item if $\R(P,X,Q)$ then there exists a path $Q \pathtau Q_1 \step{\opt{\tau}} Q_2$ such that $\R(P,X,Q_1)$ and $\R(P,X,Q_2)$ and, since $\R \subseteq \R'$, $\R'(P,X,Q_1)$ and $\R'(P',X,Q_2)$
                \item if $\R(P,Q)$ then there exists a path $Q \pathtau Q_1 \step{\opt{\tau}} Q_2$ such that $\R(P,Q_1)$ and $\R(P,Q_2)$ and, by definition of $\R'$, $\R'(P,X,Q_1)$ and $\R'(P',X,Q_2)$
                \item if there exists $Y \subseteq A$ such that $\R(P,Y,Q)$ and $(\init{P}\cup\init{Q})\cap(Y\cup\{\tau\}) = \emptyset$ then $P \nsteptau$, so this case is impossible.
            \end{itemize}
            \item If $P \step{a} P'$ with $a \in X$ then
            \begin{itemize}
                \item if $\R(P,X,Q)$ then there exists a path $Q \pathtau Q_1 \step{a} Q_2$ such that $\R(P,X,Q_1)$ and $\R(P,Q_2)$ and, since $\R \subseteq \R'$, $\R'(P,X,Q_1)$ and $\R'(P',Q_2)$
                \item if $\R(P,Q)$ then there exists a path $Q \pathtau Q_1 \step{a} Q_2$ such that $\R(P,Q_1)$ and $\R(P,Q_2)$ and, by definition of $\R'$, $\R'(P,X,Q_1)$ and $\R'(P',Q_2)$
                \item if there exists $Y \subseteq A$ such that $\R(P,Y,Q)$ and $(\init{P}\cup\init{Q})\cap(Y\cup\{\tau\}) = \emptyset$ then, since $\deadend{P}{Y}$, there exists a path $Q \pathtau Q_1 \step{a} Q_2$ such that $\R(P,Y,Q_1)$ and $\R(P',Q_2)$. Since $\deadend{Q}{Y}$, $Q = Q_1$ so there exists a path $Q \step{a} Q_2$ such that $\R'(P,X,Q)$ and $\R'(P',Q_2)$.
            \end{itemize}
            \item If $\deadend{P}{X}$ then
            \begin{itemize}
                \item if $\R(P,X,Q)$ then, since $P \nsteptau$, there exists a path $Q \pathtau Q_0 \nsteptau$. By Clause 2.a of Definition~\ref{def:generalised}, $\R(P,X,Q_0)$. Since $\R(P,X,Q_0)$, $\deadend{P}{X}$ and $Q_0 \nsteptau$, $\deadend{Q}{X}$, therefore, by definition, $\R'(P,Q_0)$.
                \item if $\R(P,Q)$ then, since $\R \subseteq \R'$, $\R'(P,Q)$.
                \item if there exists $Y \subseteq A$ such that $\R(P,Y,Q)$ and $(\init{P}\cup\init{Q})\cap(Y\cup\{\tau\}) = \emptyset$ then, by definition of $\R'$, $\R'(P,Q)$.
            \end{itemize}
            \item If $\deadend{P}{X}$ and $P \step{\rt} P'$ then
            \begin{itemize}
                \item if $\R(P,X,Q)$ then, there exists a path $Q \pathtau Q_1 \step{\rt} Q_2$ with $\R(P',X,Q_2)$. Hence also $\R'(P',X,Q_2)$.
                \item if $\R(P,Q)$ then there exists a path $Q \pathtau Q_1 \step{\rt} Q_2$ with $\R(P',X,Q_2)$. Hence also $\R'(P',X,Q_2)$.
                \item if there exists $Y \subseteq A$ such that $\R(P,Y,Q)$ and $(\init{P}\cup\init{Q})\cap(Y\cup\{\tau\}) = \emptyset$ then $\deadend{P}{(Y\cup X)}$ so there exists a path $Q \pathtau Q_1 \step{\rt} Q_2$ with $\R(P',X,Q_2)$. Hence also $\R'(P',X,Q_2)$.
            \end{itemize}
            \item If $P \nsteptau$ then 
            \begin{itemize}
                \item if $\R(P,X,Q)$ then there exists a path $Q \pathtau Q_0 \nsteptau$.
                \item if $\R(P,Q)$ then there exists a path $Q \pathtau Q_0 \nsteptau$.
                \item if there exists $Y \subseteq A$ such that $\R(P,Y,Q)$ and $(\init{P}\cup\init{Q})\cap(Y\cup\{\tau\}) = \emptyset$ then $Q \nsteptau$.
            \end{itemize}
        \end{enumerate}
    \end{enumerate}
    Let $\R$ be a rooted \tb reactive bisimulation. Let's check that it is a generalised rooted \tb reactive bisimulation. Let $P,Q \in \closed$ and $X \subseteq A$.
    \begin{enumerate}
        \item If $\R(P,Q)$
        \begin{enumerate}
            \item this condition is shared by both definitions
            \item if $\deadend{P}{X}$ and $P \step{\rt} P'$ then, since $\R(P,Q)$, $\R(P,X,Q)$. Since $\deadend{P}{X}$ and $P \step{\rt} P'$, there exists a transition $Q \step{\rt} Q'$ such that $P' \bisimtbrc[X] Q'$.
        \end{enumerate}
        \item If $\R(P,X,Q)$
        \begin{enumerate}
            \item this condition is shared by both definitions
            \item if $a\in X$, this condition is shared by both definitions; otherwise, apply Clauses 2.c and 1.a of Definition~\ref{def:rooted intuitive}
            \item if $\deadend{P}{(X\cup Y)}$ and $P \step{\rt} P'$ then, since $\deadend{P}{X}$, $\R(P,Q)$ and so $\R(P,Y,Q)$. Since $\deadend{P}{Y}$ and $P \step{\rt} P'$, there exists a transition $Q \step{\rt} Q'$ such that $P' \bisimtbrc[Y] Q'$.
        \end{enumerate}
    \end{enumerate}
    Let $\R$ be a generalised rooted \tb reactive bisimulation and define 
    \begin{align*}
        \R' := \R & \cup \{(P,X,Q) \mid \R(P,Q) \wedge X \subseteq A\} \cup \{(P,Y,Q),(P,Q) \mid \exists X \subseteq A, \R(P,X,Q) \\
        & \wedge (\init{P}\cup\init{Q})\cap(X\cup\{\tau\}) = \emptyset \wedge Y \subseteq A\}
    \end{align*}
    $\R'$ is symmetric by definition. Let's check that $\R'$ is a rooted \tb reactive bisimulation. Let $P,Q \in \closed$ and $X \subseteq A$.
    \begin{enumerate}
        \item If $\R'(P,Q)$ then $\R(P,Q)$ or there exists $Y \subseteq A$ such that $\R(P,Y,Q)$ and $(\init{P}\cup\init{Q})\cap(Y\cup\{\tau\}) = \emptyset$.
        \begin{enumerate}
            \item If $P \step{\alpha} P'$ then
            \begin{itemize}
                \item if $\R(P,Q)$ then there exists a transition $Q \step{\alpha} Q'$ such that $P' \bisimtbrc Q'$.
                \item if there exists $Y \subseteq A$ such that $\R(P,Y,Q)$ and $(\init{P}\cup\init{Q})\cap(Y\cup\{\tau\}) = \emptyset$ then, since $\R(P,Y,Q)$ and $\deadend{P}{Y}$, $\alpha \ne \tau$ so there exists a transition $Q \step{\alpha} Q'$ such that $P' \bisimtbrc Q'$.
            \end{itemize}
            \item For all $Z \subseteq A$,
            \begin{itemize}
                \item if $\R(P,Q)$ then, by definition of $\R'$, $\R'(P,Z,Q)$
                \item if there exists $Y \subseteq A$ such that $\R(P,Y,Q)$ and $(\init{P}\cup\init{Q})\cap(Y\cup\{\tau\}) = \emptyset$ then, by definition of $\R'$, $\R(P,Z,Q)$.
            \end{itemize}
        \end{enumerate}
        \item If $\R'(P,X,Q)$ then $\R(P,X,Q)$, or $\R(P,Q)$, or there exists $Y \subseteq A$ such that $\R(P,Y,Q)$ and $(\init{P}\cup\init{Q})\cap(Y\cup\{\tau\}) = \emptyset$.
        \begin{enumerate}
            \item If $P \steptau P'$ then
            \begin{itemize}
                \item if $\R(P,X,Q)$ then there exists a transition $Q \step{\tau} Q'$ such that $P' \bisimtbrc[X] Q'$,
                \item if $\R(P,Q)$ then there exists a step $Q \step{\tau} Q'$ such that $P' \bisimtbrc Q'$ and so $P' \bisimtbrc[X] Q'$
                \item if there exists $Y \subseteq A$ such that $\R(P,Y,Q)$ and $(\init{P}\cup\init{Q})\cap(Y\cup\{\tau\}) = \emptyset$ then $P \nsteptau$, so this case is impossible.
            \end{itemize}
            \item If $P \step{a} P'$ with $a \in X$ then
            \begin{itemize}
                \item if $\R(P,X,Q)$ then there exists a transition $Q \step{a} Q'$ such that $P' \bisimtbrc Q'$
                \item if $\R(P,Q)$ then there exists a transition $Q \step{a} Q'$ such that $P' \bisimtbrc Q'$
                \item if there exists $Y \subseteq A$ such that $\R(P,Y,Q)$ and $(\init{P}\cup\init{Q})\cap(Y\cup\{\tau\}) = \emptyset$ then, since $\deadend{P}{Y}$, there exists a transition $Q \step{a} Q'$ such that $P' \bisimtbrc Q'$.
            \end{itemize}
            \item If $\deadend{P}{X}$ then 
            \begin{itemize}
                \item if $\R(P,X,Q)$ then, since $\deadend{P}{X}$, $\deadend{Q}{X}$, therefore, by definition, $\R'(P,Q)$,
                \item if $\R(P,Q)$ then, by definition of $\R'$, $\R'(P,Q)$,
                \item if there exists $Y \subseteq A$ such that $\R(P,Y,Q)$ and $(\init{P}\cup\init{Q})\cap(Y\cup\{\tau\}) = \emptyset$ then, by definition of $\R'$, $\R'(P,Q)$,
            \end{itemize}
            \item If $\deadend{P}{X}$ and $P \step{\rt} P'$ then
            \begin{itemize}
                \item if $\R(P,X,Q)$ then there exists a transition $Q \step{\rt} Q'$ such that $P' \bisimtbrc[X] Q'$.
                \item if $\R(P,Q)$ then there exists a transition $Q  \step{\rt} Q'$ such that $P' \bisimtbrc[X] Q'$.
                \item if there exists $Y \subseteq A$ such that $\R(P,Y,Q)$ and $(\init{P}\cup\init{Q})\cap(Y\cup\{\tau\}) = \emptyset$ then $\deadend{P}{(Y\cup X)}$ so there exists a step $Q \step{\rt} Q'$ such that $P' \bisimtbrc[X] Q'$.
            \popQED
            \end{itemize}
        \end{enumerate}
    \end{enumerate}
\end{proof}

\section{Pohlmann Encoding} \label{app:Pohlmann}

Reactive bisimulations are sometimes complicated to check because of the large number of potential sets of allowed actions. In \cite{Pohlmann}, Pohlmann introduces an encoding which reduces strong reactive bisimilarity to strong bisimilarity. To this end he introduces unary operators $\vartheta$ and $\vartheta_X$ for $X\subseteq A$
that model placing their argument process in an environment that is triggered to change, or allows
exactly the actions in $X$, respectively. The actions $\rt_\varepsilon\notin A$ and $\varepsilon_X\notin A$ for $X \subseteq A$ are generated by the new operators, but may not be used by processes substituted for their arguments $P$.
They model a time-out action taken by the environment, and the stabilisation of an environment into
one that allows exactly the set of actions $X$, respectively. After a slight modification of the encoding, a similar result can be obtained for \tb reactive bisimilarity and its rooted version.

\begin{table}[ht]
    \[
    \begin{array}{l l l}
        \vartheta(P) \step{\alpha} \vartheta(P') & \Leftrightarrow & P \step{\alpha} P' \wedge \alpha \in A_\tau \\
        \vartheta(P) \step{\varepsilon_X} \vartheta_X(P) & & \\
        \vartheta_X(P) \steptau \vartheta_X(P') & \Leftrightarrow & P \steptau P' \\
        \vartheta_X(P) \step{a} \vartheta(P') & \Leftrightarrow & P \step{a} P' \wedge \alpha \in X \\
        \vartheta_X(P) \step{\rt_\varepsilon} \vartheta(P) & \Leftrightarrow & \deadend{P}{X} \\
        \vartheta_X(P) \step{\rt} \vartheta_X(P') & \Leftrightarrow & \deadend{P}{X} \wedge P \step{\rt} P'
    \end{array}
    \]
    \caption{Operational semantics of $\vartheta$ and $(\vartheta_X)_{X \subseteq A}$}
    \label{tab:Pohlmann operator}
\vspace{-1.5ex}
\end{table}

In \cite{Pohlmann}, the first rule only applies to $\tau$-transitions; this echoes the previous remark about applying the first clause of Definition \ref{def:intuitive} only to invisible actions. Note that the encoding rules mirror the clauses of Definition~\ref{def:intuitive}. The encoding transforms $\bisimtbrc$ into $\bisimb$, and $\bisimrtbrc$ in $\bisimrb$.

\begin{proposition} \label{prop:reduction}
  Let $P, Q \in \closed$.\\[1ex]
    \begin{minipage}{2.5in}
    \begin{itemize}
        \item $P \bisimtbrc Q \Leftrightarrow \vartheta(P) \bisimb \vartheta(Q)$
        \item $P \bisimrtbrc Q \Leftrightarrow \vartheta(P) \bisimrb \vartheta(Q)$
    \end{itemize}
    \end{minipage}\hfill
    \begin{minipage}{2.5in}
    \begin{itemize}
        \item $P \bisimtbrc[X] Q \Leftrightarrow \vartheta_X(P) \bisimb \vartheta_X(Q)$
        \item $P \bisimrtbrc[X] Q \Leftrightarrow \vartheta_X(P) \bisimrb \vartheta_X(Q)$
    \end{itemize}
    \end{minipage}
\end{proposition}

\begin{proof}
    It suffices to prove that: if $\R$ is a \tb reactive bisimulation then $\R' := \{(\vartheta(P),\vartheta(Q)) \mid \R(P,Q)\} \cup \{(\vartheta_X(P),\vartheta_X(Q)) \mid \R(P,X,Q)\}$ is a stability respecting branching bisimulation; and if $\R$ is a stability respecting branching bisimulation then $\R' := \{(P,Q), (P,X,Q) \mid \R(\vartheta(P),\vartheta(Q)) \wedge X\subseteq A\} \cup \{(P,X,Q) \mid \R(\vartheta_X(P),\vartheta_X(Q))\}$ is a \tb reactive bisimulation. The rooted case is very similar.
    
    Let $\R$ be a \tb reactive bisimulation and define 
    \begin{align*}
        \R' := \{(\vartheta(P),\vartheta(Q)) \mid \R(P,Q)\} \cup \{(\vartheta_X(P),\vartheta_X(Q)) \mid \R(P,X,Q)\}
    \end{align*}
    We are going to check that $\R'$ is a stability respecting branching bisimulation. Let $P,Q \in \closed$ such that $\R'(P,Q)$.
    \begin{itemize}
        \item If $P = \vartheta(P^\dag)$ and $Q = \vartheta(Q^\dag)$ then, by definition of $\R'$, $\R(P^\dag,Q^\dag)$.
        \begin{enumerate}
            \item If $P \step{\alpha} P'$ with $\alpha \in A_\tau \cup\{\rt_\epsilon,\epsilon_X \mid X \subseteq A\}$ then
            \begin{itemize}
                \item if $\alpha \in A_\tau$ then, by the semantics of $\vartheta$, $P' = \vartheta(P^\ddag)$ and $P^\dag \step{\alpha} P^\ddag$. Since $\R(P^\dag,Q^\dag)$, there exists a path $Q^\dag \pathtau Q^\star \step{\opt{\alpha}} Q^\ddag$ such that $\R(P^\dag,Q^\star)$ and $\R(P^\ddag,Q^\ddag)$. By the semantics, there exists a path $Q \pathtau \vartheta(Q^\star) \step{\opt{\alpha}} \vartheta(Q^\ddag)$ such that, by definition of $\R'$, $\R'(P, \vartheta(Q^\star))$ and $\R'(P', \vartheta(Q^\ddag))$.
                \item if $\alpha = \rt_\epsilon$ then this case is not possible according to the semantics of $\vartheta$.
                \item if $\alpha = \epsilon_X$ with $X \subseteq A$ then, by the semantics of $\vartheta$, $P' = \vartheta_X(P^\dag)$. Since $\R(P^\dag,Q^\dag)$, $\R(P^\dag,X,Q^\dag)$. By the semantics, $Q \step{\epsilon_X} \vartheta_X(Q^\dag)$ such that, by the definition of $\R'$, $\R'(P',\vartheta_X(Q^\dag))$.
            \end{itemize}
            \item If $P \step{\rt} P'$ then, by the semantics, this is impossible.
            \item If $P \nsteptau$ then, by the semantics of $\vartheta$, $P^\dag \nsteptau$. Since $\R(P^\dag,Q^\dag)$, $\R(P^\dag,\emptyset,Q^\dag)$, so there exists a path $Q^\dag \pathtau Q^\star \nsteptau$. By the semantics, there exists a path $Q \pathtau \vartheta(Q^\star) \nsteptau$.
        \end{enumerate}
        \item If there exists $X \subseteq A$ such that $P = \vartheta_X(P^\dag)$ and $Q = \vartheta_X(Q^\dag)$ then, by definition of $\R'$, $\R(P^\dag,X,Q^\dag)$.
        \begin{enumerate}
            \item If $P \step{\alpha} P'$ with $\alpha \in A_\tau \cup \{\rt_\epsilon,\epsilon_X \mid X\subseteq A\}$ then
            \begin{itemize}
                \item if $P \steptau P'$ then, by the semantics, $P' = \vartheta_X(P^\ddag)$ and $P^\dag \steptau P^\ddag$. Since $\R(P^\dag,X,Q^\dag)$, there exists a path $Q^\dag \pathtau Q^\star \step{\opt{\tau}} Q^\ddag$ such that $\R(P^\dag,X,Q^\star)$ and $\R(P^\ddag,X,Q^\ddag)$. By the semantics, there exists a path $Q \pathtau \vartheta_X(Q^\star) \step{\opt{\tau}} \vartheta_X(Q^\ddag)$ such that, by the definition of $\R'$, $\R'(P,\vartheta_X(Q^\star))$ and $\R'(P',\vartheta_X(Q^\ddag))$.
                \item if $P \step{a} P'$ with $a \in A$ then, by the semantics, $a \in X$, $P' = \vartheta(P^\ddag)$ and $P^\dag \step{a} P^\ddag$. Since $\R(P^\dag,X,Q^\dag)$, there exists a path $Q^\dag \pathtau Q^\star \step{a} Q^\ddag$ such that $\R(P^\dag,X,Q^\star)$ and $\R(P^\ddag,Q^\ddag)$. By the semantics, there exists a path $Q \pathtau \vartheta_X(Q^\star) \step{a} \vartheta(Q^\ddag)$ such that, by the definition of $\R'$, $\R'(P,\vartheta_X(Q^\star))$ and $\R'(P',\vartheta(Q^\ddag))$.
                \item if $P \step{\rt_\epsilon} P'$ then, by the semantics, $P' = \vartheta(P^\dag)$ and $\deadend{P^\dag}{X}$. Since $\R(P^\dag,X,Q^\dag)$ and $P^\dag \nsteptau$, there exists a path $Q^\dag \pathtau Q^\star \nsteptau$. Moreover, $\R(P^\dag,X,Q^\star)$. Since $\deadend{P^\dag}{X}$, $\R(P^\dag,X,Q^\star)$ and $Q^\star \nsteptau$, $\deadend{Q^\star}{X}$ and $\R(P^\dag,Q^\star)$. By the semantics, there exists a path $Q \pathtau \vartheta_X(Q^\star) \step{\rt_\epsilon} \vartheta(Q^\star)$ such that, by the definition of $\R'$, $\R'(P,\vartheta_X(Q^\star))$ and $\R'(P',\vartheta(Q^\star))$.
                \item if $\alpha = \epsilon_X$ with $X \subseteq A$ then this case is impossible according to the semantics of $\vartheta_X$.
            \end{itemize}
            \item if $P \step{\rt} P'$ then, by the semantics, $P' = \vartheta_X(P^\ddag)$, $\deadend{P^\dag}{X}$ and $P^\dag \step{\rt} P^\ddag$. Since $\R(P^\dag,X,Q^\dag)$ and $\deadend{P^\dag}{X}$, according to Lemma \ref{lem:obvious}.4, there exists a path $Q^\dag \pathtau Q^\dag_1 \nsteptau$ with $\init{Q_1^\dag} = \init{P^\dag}$ and $\R(P^\dag,X,Q^\dag_1)$. Since $\deadend{P^\dag}{X}$ and $P^\dag \step{\rt} P^\ddag$ and $Q^\dag_1 \nsteptau$, there exists a transition $Q^\dag_1 \step{\rt} Q^\dag_2$ with $\R(P^\ddag,X,Q^\dag_2)$. According to the semantics, $Q \pathtau \vartheta_X(Q^\dag_1) \step{\rt} \vartheta_X(Q_2^\dag)$ and, by definition of $\R'$, $\R'(P',\vartheta_X(Q_2^\dag))$.
            \item if $P \nsteptau$ then, by the semantics of $\vartheta_X$, $P^\dag \nsteptau$. Since $\R(P^\dag,X,Q^\dag)$, there exists a path $Q^\dag \pathtau Q_0 \nsteptau$. By the semantics, there exists a path $Q \pathtau \vartheta_X(Q_0) \nsteptau$.
        \end{enumerate}
    \end{itemize}
    Let $\R$ be a stability respecting branching bisimulation and define
    \begin{align*}
        \R' := & \{(P,Q), (P,X,Q) \mid \R(\vartheta(P),\vartheta(Q)) \wedge X \subseteq A\} \cup \{(P,X,Q) \mid \R(\vartheta_X(P),\vartheta_X(Q))\}
    \end{align*}
    We are going to show that $\R'$ is a \tb reactive bisimulation. Let $P, Q \in \closed$ and $X \subseteq A$.
    \begin{enumerate}
        \item If $\R'(P,Q)$ then $\R(\vartheta(P),\vartheta(Q))$.
        \begin{enumerate}
            \item If $P \step{\alpha} P'$ with $\alpha \in A_\tau$ then, by the semantics, $\vartheta(P) \step{\alpha} \vartheta(P')$. Since $\R(\vartheta(P),\vartheta(Q))$, there exists a path $\vartheta(Q) \pathtau Q^\star \step{\opt{\alpha}} Q^\ddag$ such that $\R(\vartheta(P),Q^\star)$ and $\R(\vartheta(P'),Q^\ddag)$. By the semantics, $Q^\star = \vartheta(Q_1)$, $Q^\ddag = \vartheta(Q_2)$ and $Q \pathtau Q_1 \step{\opt{\alpha}} Q_2$ such that, by definition of $\R'$, $\R'(P,Q_1)$ and $\R'(P',Q_2)$.
            \item For all $Y \subseteq A$, by definition of $\R'$, $\R'(P,Y,Q)$.
        \end{enumerate}
        \item If $\R'(P,X,Q)$ then $\R(\vartheta(P),\vartheta(Q))$ or $\R(\vartheta_X(P),\vartheta_X(Q))$. If $\R(\vartheta(P),\vartheta(Q))$ then $\vartheta(P) \!\step{\epsilon_X} \vartheta_X(P)$, thus there exists a path $\vartheta(Q) \pathtau Q^\star \step{\epsilon_X} Q^\ddag$ such that $\R(\vartheta(P),Q^\star)$ and $\R(\vartheta_X(P),Q^\ddag)$. By the semantics, $Q^\star = \vartheta(Q_0)$, $Q^\ddag = \vartheta_X(Q_0)$ and $Q \pathtau Q_0$. Therefore, there exists a path $Q \pathtau Q_0$ such that $\R(\vartheta_X(P),\vartheta_X(Q_0))$.
        \begin{enumerate}
            \item If $P \step{\tau} P'$ then, by the semantics, $\vartheta_X(P) \step{\tau} \vartheta_X(P')$. Since $\R(\vartheta_X(P),\vartheta_X(Q_0))$, there exists a path $\vartheta_X(Q_0) \pathtau Q^\star \step{\opt{\tau}} Q^\ddag$ such that $\R(\vartheta_X(P),Q^\star)$ and $\R(\vartheta_X(P'),Q^\ddag)$. By the semantics, $Q^\star = \vartheta_X(Q_1)$, $Q^\ddag = \vartheta_X(Q_2)$ and $Q \pathtau Q_1 \step{\opt{\tau}} Q_2$ such that, by definition of $\R'$, $\R'(P,X,Q_1)$ and $\R'(P',X,Q_2)$.
            \item If $P \step{a} P'$ with $a \in X$ then, by the semantics, $\vartheta_X(P) \step{a} \vartheta(P')$. As $\R(\vartheta_X(P),\vartheta_X(Q_0))$, there exists a path $\vartheta_X(Q_0) \pathtau Q^\star \step{a} Q^\ddag$ such that $\R(\vartheta_X(P),Q^\star)$ and $\R(\vartheta(P'),Q^\ddag)$. By the semantics, $Q^\star = \vartheta_X(Q_1)$, $Q^\ddag = \vartheta(Q_2)$ and $Q \pathtau Q_1 \step{a} Q_2$ such that, by definition of $\R'$, $\R'(P,X,Q_1)$ and $\R'(P',Q_2)$.
            \item If $\deadend{P}{X}$ then, by the semantics, $\vartheta_X(P) \step{\rt_\epsilon} \vartheta(P)$. As $\R(\vartheta_X(P),\vartheta_X(Q_0))$, there exists a path $\vartheta_X(Q_0) \pathtau Q^\star \step{\rt_\epsilon} Q^\ddag$ such that $\R(\vartheta_X(P),Q^\star)$ and $\R(\vartheta(P),Q^\ddag)$. By the semantics, $Q^\star = \vartheta_X(Q_0')$, $Q^\ddag = \vartheta(Q_0')$ and $Q \pathtau Q_0'$ such that, by definition of $\R'$, $\R'(P,Q_0')$.
            \item If $\deadend{P}{X}$ and $P \step{\rt} P'$ then, by the semantics, $\vartheta_X(P) \step{\rt} \vartheta_X(P')$. Since $\R(\vartheta_X(P),\vartheta_X(Q_0))$, there exists a path $\vartheta_X(Q_0) \pathtau Q^\dag_1 \step{\rt} Q^\dag_2$ with $\R(\vartheta_X(P),Q^\dag_1)$ and $\R(\vartheta_X(P'),Q^\dag_2)$. By the semantics, there exists a path $Q_0 \pathtau Q_1 \step{\rt} Q_2$ such that $Q^\dag_1 = \vartheta_X(Q_1)$, $\deadend{Q_1}{X}$ and $Q^\dag_2 = \vartheta_X(Q_2)$. Thus, by definition of $\R'$, $\R'(P',X,Q_{2})$.
            \item If $P \nsteptau$ then, by the semantics, $\vartheta_X(P) \nsteptau$. Since $\R(\vartheta_X(P),\vartheta_X(Q_0))$, there exists a path $\vartheta_X(Q_0) \pathtau Q^\star \nsteptau$. By the semantics, $Q^\star = \vartheta_X(Q_1)$ and $Q \pathtau Q_1 \nsteptau$.
        \end{enumerate}
    \end{enumerate}
    Let $\R$ be a rooted \tb reactive bisimulation and define 
    \begin{align*}
        \R' := \{(\vartheta(P),\vartheta(Q)) \mid \R(P,Q)\} \cup \{(\vartheta_X(P),\vartheta_X(Q)) \mid \R(P,X,Q)\}
    \end{align*}
    We are going to check that $\R'$ is a rooted stability respecting branching bisimulation. Let $P,Q \in \closed$ such that $\R'(P,Q)$.
    \begin{itemize}
        \item If $P = \vartheta(P^\dag)$ and $Q = \vartheta(Q^\dag)$ then, by definition of $\R'$, $\R(P^\dag,Q^\dag)$.
        \begin{enumerate}
            \item Let $P \step{\alpha} P'$ with $\alpha \in Act \cup\{\rt_\epsilon,\epsilon_X \mid X \mathop\subseteq A\}$.
            \begin{itemize}
                \item If $\alpha \in A_\tau$ then, by the semantics of $\vartheta$, $P' = \vartheta(P^\ddag)$ and $P^\dag \step{\alpha} P^\ddag$. Since $\R(P^\dag,Q^\dag)$, there exists a transition $Q^\dag \step{\alpha} Q^\ddag$ such that $P^\ddag \bisimtbrc Q^\ddag$. By the semantics, there exists a transition $Q \step{\alpha} \vartheta(Q^\ddag)$ such that, by the first part of this proof, $\vartheta(P') \bisimb \vartheta(Q^\ddag)$.
                \item The case $\alpha = \rt_\epsilon$ is not possible according to the semantics of $\vartheta$.
                \item If $\alpha = \epsilon_X$ with $X \subseteq A$ then, by the semantics of $\vartheta$, $P' = \vartheta_X(P^\dag)$. Since $\R(P^\dag,Q^\dag)$, $P^\dag \bisimrtbrc Q^\dag$ so $P^\dag \bisimtbrc[X] Q^\dag$. By the semantics, $Q \step{\epsilon_X} \vartheta_X(Q^\dag)$ such that, by the first part of this proof, $P' \bisimb \vartheta_X(Q^\dag)$.
                \item The case $\alpha=\rt$, by the semantics, is not possible.
            \end{itemize}
        \end{enumerate}
        \item If there exists $X \subseteq A$ such that $P = \vartheta_X(P^\dag)$ and $Q = \vartheta_X(Q^\dag)$ then, by definition of $\R'$, $\R(P^\dag,X,Q^\dag)$.
        \begin{enumerate}
            \item Let $P \step{\alpha} P'$ with $\alpha \in Act \cup\{\rt_\epsilon,\epsilon_X \mid X \mathop\subseteq A\}$.
            \begin{itemize}
                \item If $P \steptau P'$ then, by the semantics, $P' = \vartheta_X(P^\ddag)$ and $P^\dag \steptau P^\ddag$. Since $\R(P^\dag,X,Q^\dag)$, there exists a transition $Q^\dag \steptau Q^\ddag$ such that $P^\ddag \bisimtbrc[X] Q^\ddag$. By the semantics, there exists a transition $Q \steptau \vartheta_X(Q^\ddag)$ such that, by the first part, $P' \bisimb \vartheta_X(Q^\ddag)$.
                \item If $P \step{a} P'$ with $a \in A$ then, by the semantics, $a \in X$, $P' = \vartheta(P^\ddag)$ and $P^\dag \step{a} P^\ddag$. Since $\R(P^\dag,X,Q^\dag)$, there exists a transition $Q^\dag \step{a} Q^\ddag$ such that $P^\ddag \bisimtbrc Q^\ddag$. By the semantics, there exists a transition $Q \step{a} \vartheta(Q^\ddag)$ such that, by the first part, $P' \bisimb \vartheta(Q^\ddag)$.
                \item If $P \step{\rt_\epsilon} P'$ then, by the semantics, $P' = \vartheta(P^\dag)$ and $\deadend{P^\dag}{X}$. Since $\R(P^\dag,X,Q^\dag)$ and $\deadend{P^\dag}{X}$, $\deadend{Q^\dag}{X}$ and $\R(P^\dag,Q^\dag)$. By the semantics, there exists a path $Q \step{\rt_\epsilon} \vartheta(Q^\dag)$ such that, by the definition of $\R'$, $\R'(P',\vartheta(Q^\dagger))$. Considering the previous case, this implies that $P' \bisimrb \vartheta(Q^\dag)$ and so $P' \bisimb \vartheta(Q^\dag)$.
                \item The case $\alpha = \epsilon_X$ with $X \subseteq A$ is impossible according to the semantics of $\vartheta_X$.
                \item If $P \step{\rt} P'$ then, by the semantics, $P' = \vartheta_X(P^\ddag)$, $\deadend{P^\dag}{X}$ and $P^\dag \step{\rt} P^\ddag$. Since $\R(P^\dag,X,Q^\dag)$, there exists a path $Q^\dag \step{\rt} Q^\ddag$ such that $P^\ddag \bisimtbrc[X] Q^\ddag$. Moreover, $\deadend{Q^\dag}{X}$. By the semantics, there exists a path $Q \step{\rt} \vartheta_X(Q^\ddag)$ such that, by the first part, $P' \bisimb \vartheta_X(Q^\ddag)$.
            \end{itemize}
        \end{enumerate}
    \end{itemize}
    Let $\R$ be a rooted stability respecting branching bisimulation and define
\[
        \R' := {\bisimrtbrc}\, \cup \{(P,Q), (P,X,Q) \mid \R(\vartheta(P),\vartheta(Q)) \wedge X \subseteq A\} \cup \{(P,X,Q) \mid \R(\vartheta_X(P),\vartheta_X(Q))\}
\]
    We are going to show that $\R'$ is a rooted \tb reactive bisimulation. Let $P, Q \in \closed$ and $X \subseteq A$.
    \begin{enumerate}
        \item If $\R'(P,Q)$ then $\R(\vartheta(P),\vartheta(Q))$.
        \begin{enumerate}
            \item If $P \step{\alpha} P'$ with $\alpha \in A_\tau$ then, by the semantics, $\vartheta(P) \step{\alpha} \vartheta(P')$. Since $\R(\vartheta(P),\vartheta(Q))$, there exists a transition $\vartheta(Q) \step{\alpha} Q^\ddag$ such that $\vartheta(P') \bisimb Q^\ddag$. By the semantics, $Q^\ddag = \vartheta(Q')$ and $Q \step{\alpha} Q'$ such that, by the second part, $P' \bisimtbrc Q'$.
            \item For all $Y \subseteq A$, by definition of $\R'$, $\R'(P,Y,Q)$.
        \end{enumerate}
        \item If $\R'(P,X,Q)$ then $\R(\vartheta(P),\vartheta(Q))$ or $\R(\vartheta_X(P),\vartheta_X(Q))$.
        \begin{enumerate}
            \item If $P \steptau P'$ then, by the semantics, $\vartheta(P) \steptau \vartheta(P')$ and $\vartheta_X(P) \steptau \vartheta_X(P')$.
            \begin{itemize}
                \item If $\R(\vartheta(P),\vartheta(Q))$, there exists a path $\vartheta(Q) \steptau Q^\ddag$ such that $\vartheta(P') \bisimb Q^\ddag$. By the semantics, $Q^\ddag = \vartheta(Q')$ and $Q \steptau Q'$ so that, by the second part, $P' \bisimtbrc Q'$ and thus $P' \bisimtbrc[X] Q'$.
                \item If $\R(\vartheta_X(P),\vartheta_X(Q))$, there exists a path $\vartheta_X(Q) \steptau Q^\ddag$ such that $\vartheta_X(P') \bisimb Q^\ddag$. By the semantics, $Q^\ddag = \vartheta_X(Q')$ and $Q \steptau Q'$ so that, by the second part, $P' \bisimtbrc[X] Q'$.
            \end{itemize}
            \item If $P \step{a} P'$ with $a \in X$ then, by the semantics, $\vartheta(P) \step{a} \vartheta(P')$, $\vartheta_X(P) \step{a} \vartheta(P')$. 
            \begin{itemize}
                \item If $\R(\vartheta(P),\vartheta(Q))$, there exists a step $\vartheta(Q) \step{a} Q^\ddag$ such that $\vartheta(P') \bisimb Q^\ddag$. By the semantics, $Q^\ddag = \vartheta(Q')$ and $Q \step{a} Q'$ such that, by the second part, $P' \bisimtbrc Q'$.
                \item If $\R(\vartheta_X(P),\vartheta_X(Q))$, there exists a path $\vartheta_X(Q) \step{a} Q^\ddag$ such that $\vartheta(P') \bisimb Q^\ddag$. By the semantics, $Q^\ddag = \vartheta(Q')$ and $Q \step{a} Q'$ such that, by the second part, $P' \bisimtbrc Q'$.
            \end{itemize}
            \item If $\deadend{P}{X}$ then
            \begin{itemize}
                \item if $\R(\vartheta(P),\vartheta(Q))$ then, by definition of $\R'$, $\R'(P,Q)$.
                \item if $\R(\vartheta_X(P),\vartheta_X(Q))$ then, by the semantics, $\vartheta_X(P) \step{\rt_\epsilon} \vartheta(P)$. As $\R(\vartheta_X(P),\vartheta_X(Q))$, $\vartheta_X(Q) \step{t_\epsilon} Q'$ with $\vartheta(P) \bisimb Q'$ and, by the semantics, $Q' = \vartheta(Q)$, thus, by the second part, $P \bisimtbrc Q$. Since $\deadend{P}{X}$ and $\deadend{Q}{X}$, $P \bisimrtbrc Q$ by Lemma~\ref{lem:obvious}.3, and so $\R'(P,Q)$. 
            \end{itemize}
            \item If $\deadend{P}{X}$ and $P \step{\rt} P'$ then, by the semantics, $\vartheta_X(P) \step{\rt} \vartheta_X(P')$.
            \begin{itemize}
                \item If $\R(\vartheta(P),\vartheta(Q))$ then, since $\vartheta(P) \step{\epsilon_X} \vartheta_X(P)$, there exists a transition $\vartheta(Q) \step{\epsilon_X} Q^\ddag$ with $\vartheta_X(P) \bisimb Q^\dag$. By the semantics, $Q^\dag = \vartheta_X(Q)$. Since $P\!\nsteptau$, also $\vartheta(P)\!\nsteptau$, so $\vartheta(Q)\!\nsteptau$ and $Q\!\nsteptau$. Moreover, since $\vartheta_X(P) \step{\rt} \vartheta_X(P')$ and $Q\nsteptau$, there exists a transition $\vartheta_X(Q) \step{\rt} Q^\ddag$ such that $\vartheta_X(P') \bisimb Q^\ddag$. By the semantics, $Q \step{\rt} Q'$ and $Q^\ddag = \vartheta_X(Q')$, thus, by the second part, $P' \bisimtbrc[X] Q'$.
                \item if $\R(\vartheta_X(P),\vartheta_X(Q))$ then there exists a transition $\vartheta_X(Q) \step{\rt} Q^\ddag$ with $\vartheta_X(P') \bisimb Q^\ddag$. By the semantics, $Q^\ddag = \vartheta_X(Q')$ and $Q \step{\rt} Q'$. Moreover, by the second part, $P' \bisimtbrc[X] Q'$.
            \popQED
            \end{itemize}
        \end{enumerate}
    \end{enumerate}
\end{proof}

%% file: appendix.tex
\section{Proofs of Stuttering Property and Transitivity} \label{app:intro}

\begin{proof}[Proof of Lemma \ref{lem:stuttering}]
    Let $\R$ be a \tb reactive bisimulation. Let's define
    \begin{align*}
        \R' := & \{(P^\dag,Q),(Q,P^\dag) \mid \exists P,P^\ddag \in \closed, P \pathtau P^\dag \pathtau P^\ddag \wedge \R(P,Q) \wedge \R(P^\ddag,Q)\} \cup {}\\
        & \{(P^\dag\!\!,X,Q),(Q,X,P^\dag) \mid \exists P,P^\ddag \mathbin\in \closed, P \mathbin{\pathtau} P^\dag \mathbin{\pathtau} P^\ddag \wedge \R(P,X,Q) \wedge \R(P^\ddag\!\!,X,Q)\}
    \end{align*}
    $\R'$ is symmetric by definition and we are going to prove that $\R'$ is a \tb reactive bisimulation. Note that $\R \subseteq \R'$ (by taking $P^\ddag=P^\dag$). Let $P,Q \in \closed$ and $X \subseteq A$.
    \begin{enumerate}
        \item Let $\R'(P,Q)$.
        \begin{enumerate}
            \item Suppose $P \step{\alpha} P'$ with $\alpha \in A_\tau$.
            \begin{itemize}
                \item Let there exist $P^\dag, P^\ddag \in \closed$ such that $P^\dag \pathtau P \pathtau P^\ddag$, $\R(P^\dag,Q)$ and $\R(P^\ddag,Q)$. Since $P^\dag \pathtau P$ and $\R(P^\dag,Q)$, there exists a path $Q \pathtau Q_0$ such that $\R(P,Q_0)$. Since $P \step{\alpha} P'$, there exists a path $Q_0 \pathtau Q_1 \step{\opt{\alpha}} Q_2$ such that $\R(P,Q_1)$ and $\R(P',Q_2)$. Thus, there exists a path $Q \pathtau Q_1 \step{\opt{\alpha}} Q_2$ such that, since $\R \subseteq \R'$, $\R'(P,Q_1)$ and $\R'(P',Q_2)$.
                \item Let there exist $Q^\dag, Q^\ddag \in \closed$ such that $Q^\dag \pathtau Q \pathtau Q^\ddag$, $\R(P,Q^\dag)$ and $\R(P,Q^\ddag)$. Since $\R(P,Q^\ddag)$, there exists a path $Q^\ddag \pathtau Q_1 \step{\opt{\alpha}} Q_2$ such that $\R(P,Q_1)$ and $\R(P',Q_2)$. Since $Q \pathtau Q^\ddag$, there exists a path $Q \pathtau Q_1 \step{\opt{\alpha}} Q_2$ such that, since $\R \subseteq \R'$, $\R'(P,Q_1)$ and $\R'(P',Q_2)$.
            \end{itemize}
            \item For all $Y \subseteq A$, $\R'(P,Y,Q)$ by definition of $\R'$.
        \end{enumerate}
        \item Let $\R'(P,X,Q)$.
        \begin{enumerate}
            \item Suppose $P \steptau P'$.
            \begin{itemize}
                \item Let there exist $P^\dag, P^\ddag \mathbin\in \closed$ such that $P^\dag \mathbin{\pathtau} P \mathbin{\pathtau} P^\ddag$, $\R(P^\dag\!,X,Q)$ and $\R(P^\ddag\!,X,Q)$. Since $P^\dag \pathtau P$ and $\R(P^\dag,X,Q)$, there exists a path $Q \pathtau Q_0$ such that $\R(P,X,Q_0)$. Since $P \steptau P'$, there exists a path $Q_0 \pathtau Q_1 \step{\opt{\tau}} Q_2$ such that $\R(P,X,Q_1)$ and $\R(P',X,Q_2)$. Thus, there exists a path $Q \pathtau Q_1 \step{\opt{\tau}} Q_2$ such that, since $\R \subseteq \R'$, $\R'(P,X,Q_1)$ and $\R'(P',X,Q_2)$.
                \item Let there exist $Q^\dag, Q^\ddag \mathbin\in \closed$ such that $Q^\dag \mathbin{\pathtau} Q \mathbin{\pathtau} Q^\ddag$, $\R(P,X,Q^\dag)$ and $\R(P,X,Q^\ddag)$. Since $\R(P,X,Q^\ddag)$, there exists a path $Q^\ddag \pathtau Q_1 \step{\opt{\tau}} Q_2$ such that $\R(P,X,Q_1)$ and $\R(P',X,Q_2)$. Since $Q \pathtau Q^\ddag$, there exists a path $Q \pathtau Q_1 \step{\opt{\tau}} Q_2$ such that, since $\R \subseteq \R'$, $\R'(P,X,Q_1)$ and $\R'(P',X,Q_2)$.
            \end{itemize}
            \item Suppose $P \step{a} P'$ with $a \in X$.
            \begin{itemize}
                \item Let there exist $P^\dag, P^\ddag \mathbin\in \closed$ such that $P^\dag \mathbin{\pathtau} P \mathbin{\pathtau} P^\ddag$, $\R(P^\dag\!,X,Q)$ and $\R(P^\ddag\!,X,Q)$. Since $P^\dag \pathtau P$ and $\R(P^\dag,X,Q)$, there exists a path $Q \pathtau Q_0$ such that $\R(P,X,Q_0)$. Since $P \step{a} P'$, there exists a path $Q_0 \pathtau Q_1 \step{a} Q_2$ such that $\R(P,X,Q_1)$ and $\R(P',Q_2)$. Thus, there exists a path $Q \pathtau Q_1 \step{a} Q_2$ such that, since $\R \subseteq \R'$, $\R'(P,X,Q_1)$ and $\R'(P',Q_2)$.
                \item Let there exist $Q^\dag, Q^\ddag \mathbin\in \closed$ such that $Q^\dag \mathbin{\pathtau} Q \mathbin{\pathtau} Q^\ddag$, $\R(P,X,Q^\dag)$ and $\R(P,X,Q^\ddag)$. Since $\R(P,X,Q^\ddag)$, there exists a path $Q^\ddag \pathtau Q_1 \step{a} Q_2$ such that $\R(P,X,Q_1)$ and $\R(P',Q_2)$. Since $Q \pathtau Q^\ddag$, there exists a path $Q \pathtau Q_1 \step{a} Q_2$ such that, since $\R \subseteq \R'$, $\R'(P,X,Q_1)$ and $\R'(P',Q_2)$.
            \end{itemize}
            \item Suppose $\deadend{P}{X}$.
            \begin{itemize}
                \item Let there exist $P^\dag, P^\ddag \mathbin\in \closed$ such that $P^\dag \mathbin{\pathtau} P \mathbin{\pathtau} P^\ddag$, $\R(P^\dag\!,X,Q)$ and $\R(P^\ddag\!,X,Q)$. Since $P^\dag \pathtau P$ and $\R(P^\dag,X,Q)$, there exists a path $Q \pathtau Q_0$ such that $\R(P,X,Q_0)$. Since $\deadend{P}{X}$, there exists a path $Q_0 \pathtau Q_0'$ such that $\R(P,Q_0')$. Thus, there exists a path $Q \pathtau Q_0'$ such that, since $\R \subseteq \R'$, $\R'(P,Q'_0)$.
                \item Let there exist $Q^\dag, Q^\ddag \mathbin\in \closed$ such that $Q^\dag \mathbin{\pathtau} Q \mathbin{\pathtau} Q^\ddag$, $\R(P,X,Q^\dag)$ and $\R(P,X,Q^\ddag)$. Since $\R(P,X,Q^\ddag)$, there exists a path $Q^\ddag \pathtau Q_0$ such that $\R(P,Q_0)$. Since $Q \pathtau Q^\ddag$, there exists a path $Q \pathtau Q_0$ such that, since $\R \subseteq \R'$, $\R'(P,Q_0)$.
            \end{itemize}
            \item Suppose $\deadend{P}{X}$ and $P \step{\rt} P'$.
            \begin{itemize}
                \item Let there exist $P^\dag, P^\ddag \mathbin\in \closed$ such that $P^\dag \mathbin{\pathtau} P \mathbin{\pathtau} P^\ddag$, $\R(P^\dag\!,X,Q)$ and $\R(P^\ddag\!,X,Q)$. Since $P^\dag \pathtau P$ and $\R(P^\dag,X,Q)$, there exists a path $Q \pathtau Q_0$ such that $\R(P,X,Q_0)$. Since $\deadend{P}{X}$ and $P \step{\rt} P'$, there exists a path $Q_0 \pathtau Q_1 \step{\rt} Q_2$ with $\R(P',X,Q_2)$. Thus, there exists a path $Q \pathtau Q_1 \step{\rt} Q_2$ such that, since $\R\subseteq\R'$, $\R'(P',X,Q_2)$.
                \item Let there exist $Q^\dag, Q^\ddag \mathbin\in \closed$ such that $Q^\dag \mathbin{\pathtau} Q \mathbin{\pathtau} Q^\ddag$, $\R(P,X,Q^\dag)$ and $\R(P,X,Q^\ddag)$. Since $\R(P,X,Q^\ddag)$, there exists a path $Q^\ddag \pathtau Q_1 \step{\rt} Q_2$ with $\R(P',X,Q_2)$. Since $Q \pathtau Q^\ddag$, there exists a path $Q \pathtau Q_1 \step{\rt} Q_2$ such that, since $\R\subseteq\R'$, $\R'(P',X,Q_2)$.
            \end{itemize}
            \item Suppose $P \nsteptau$.
            \begin{itemize}
                \item Let there exist $P^\dag, P^\ddag \mathbin\in \closed$ such that $P^\dag \mathbin{\pathtau} P \mathbin{\pathtau} P^\ddag$, $\R(P^\dag\!,X,Q)$ and $\R(P^\ddag\!,X,Q)$. Since $P^\dag \pathtau P$ and $\R(P^\dag,X,Q)$, there exists a path $Q \pathtau Q_0$. Since $P \nsteptau$, there exists a path $Q_0 \pathtau Q_1 \nsteptau$. Thus, there exists a path $Q \pathtau Q_1 \nsteptau$.
                \item Let there exist $Q^\dag, Q^\ddag \mathbin\in \closed$ such that $Q^\dag \mathbin{\pathtau} Q \mathbin{\pathtau} Q^\ddag$, $\R(P,X,Q^\dag)$ and $\R(P,X,Q^\ddag)$. Since $\R(P,X,Q^\ddag)$, there exists a path $Q^\ddag \pathtau Q_1 \nsteptau$. Since $Q \pathtau Q^\ddag$, there exists a path $Q \pathtau Q_1 \nsteptau$.
           \popQED
            \end{itemize}
        \end{enumerate}
    \end{enumerate}
\end{proof}

\begin{proof}[Proof of Proposition \ref{prop:equivalence}]
    Let $\R_1$ and $\R_2$ be two \tb reactive bisimulations and define
    \begin{align*}
        \R := (\R_1 \circ \R_2)\cup(\R_2\circ \R_1)
    \end{align*}
    $\R$ is clearly symmetric by definition. Let's check that $\R$ is a \tb reactive bisimulation. Let $P,Q \in \closed$ and $X \subseteq A$.
    \begin{enumerate}
        \item If $\R(P,Q)$ then there exists $R \in \closed$ such that $\R_1(P,R)$ and $\R_2(R,Q)$, or $\R_2(P,R)$ and $\R_1(R,Q)$. The two possibilities are similar; thus, suppose without loss of generality that $\R_1(P,R)$ and $\R_2(R,Q)$.
        \begin{enumerate}
            \item If $P \step{\alpha} P'$ with $\alpha \in A_\tau$ then, since $\R_1(P,R)$, there exists a path $R \pathtau R_1 \step{\opt{\alpha}} R_2$ such that $\R_1(P,R_1)$ and $\R_1(P',R_2)$. Since $\R_2(R,Q)$ and $R \pathtau R_1$, there exists a path $Q \pathtau Q_0$ such that $\R_2(R_1,Q_0)$. Since $R_1 \step{\opt{\alpha}} R_2$, there exists a path $Q_0 \pathtau Q_1 \step{\opt{\alpha}} Q_2$ such that $\R_2(R_1,Q_1)$ and $\R_2(R_2,Q_2)$. By definition of $\R$, there exists a path $Q \pathtau Q_1 \step{\opt{\alpha}} Q_2$ such that $\R(P,Q_1)$ and $\R(P',Q_2)$.
            \item For all $Y \subseteq A$, since $\R_1(P,R)$ and $\R_2(R,Q)$, $\R_1(P,Y,R)$ and $\R_2(R,Y,Q)$, thus, $\R(P,Y,Q)$.
        \end{enumerate}
        \item If $\R(P,X,Q)$ then there exists $R \in \closed$ such that $\R_1(P,X,R)$ and $\R_2(R,X,Q)$, or $\R_2(P,X,R)$ and $\R_1(R,X,Q)$. The two possibilities are similar; thus, suppose without loss of generality that $\R_1(P,X,R)$ and $\R_2(R,X,Q)$.
        \begin{enumerate}
            \item If $P \steptau P'$ then, since $\R_1(P,X,R)$, there exists a path $R \pathtau R_1 \step{\opt{\tau}} R_2$ such that $\R_1(P,X,R_1)$ and $\R_1(P',X,R_2)$. Since $\R_2(R,X,Q)$ and $R \pathtau R_1$, there exists a path $Q \pathtau Q_0$ such that $\R_2(R_1,X,Q_0)$. Since $R_1 \step{\opt{\tau}} R_2$, there exists a path $Q_0 \pathtau Q_1 \step{\opt{\tau}} Q_2$ such that $\R_2(R_1,X,Q_1)$ and $\R_2(R_2,X,Q_2)$. By definition of $\R$, there exists a path $Q \pathtau Q_1 \step{\opt{\tau}} Q_2$ such that $\R(P,X,Q_1)$ and $\R(P',X,Q_2)$.
            \item If $P \step{a} P'$ with $a \in X$ then, since $\R_1(P,X,R)$, there exists a path $R \pathtau R_1 \step{a} R_2$ such that $\R_1(P,X,R_1)$ and $\R_1(P',R_2)$. Since $\R_2(R,X,Q)$ and $R \pathtau R_1$, there exists a path $Q \pathtau Q_0$ such that $\R_2(R_1,X,Q_0)$. Since $R_1 \step{a} R_2$, there exists a path $Q_0 \pathtau Q_1 \step{a} Q_2$ such that $\R_2(R_1,X,Q_1)$ and $\R_2(R_2,Q_2)$. By definition of $\R$, there exists a path $Q \pathtau Q_1 \step{a} Q_2$ such that $\R(P,X,Q_1)$ and $\R(P',Q_2)$.
            \item If $\deadend{P}{X}$ then, since $P \nsteptau$, there exists a path $R \pathtau R_0 \nsteptau$. Moreover, using Clause 2.a, $\R_1(P,X,R_0)$. Moreover, there exists a path $R_0 \pathtau R'_0$ such that $\R_1(P,R'_0)$, but, since $R_0 \nsteptau$, $R_0 = R'_0$. By Clause 1.a, $\init{P} =\init{R_0}$, so $\deadend{R_0}{X}$. Since $\R_2(R,X,Q)$ and $R \pathtau R_0$, there exists a path $Q \pathtau Q_0$ such that $\R_2(R_0,X,Q_0)$. Moreover, since $\deadend{R_0}{X}$, there exists a path $Q_0 \pathtau Q_0'$ such that $\R_2(R_0,Q'_0)$. Thus, there exists a path $Q \pathtau Q_0'$ such that, by definition of $\R$, $\R(P,Q'_0)$.
            \item If $\deadend{P}{X}$ and $P \step{\rt} P'$ then, since $\R_1(P,X,R)$, according to Lemma \ref{lem:obvious}.4, there exists a path $R \pathtau R_1 \nsteptau$ with $\deadend{R_1}{X}$ and $\R_1(P,X,R_1)$. Moreover, there exists a transition $R_1 \mathbin{\step{\rt}} R_2$ such that $\R_1(P'\!,X,R_2)$. Since $\R_2(R,X,Q)$ and $R \pathtau R_1$, there exists a path $Q \pathtau Q_0$ such that $\R_2(R_1,X,Q_0)$. Since $\deadend{R_1}{X}$ and $R_1 \step{\rt} R_2$, there exists a path $Q_0 \pathtau Q_1 \step{\rt} Q_2$ with $\R_2(R_2,X,Q_2)$. As a result, there exists a path $Q \pathtau Q_1 \step{\rt} Q_2$ such that $\R(P',X,Q_2)$.
            \item If $P \nsteptau$ then, since $\R_1(P,X,R)$, there exists a path $R \pathtau R_0 \nsteptau$. Since $\R_2(R,X,Q)$ and $R \pathtau R_0$, there exists a path $Q \pathtau Q_0$ such that $\R_2(R_0,X,Q_0)$. Since $R_0 \nsteptau$, there exists a path $Q_0 \pathtau Q_0' \nsteptau$. Hence there exists a path $Q \pathtau Q_0' \nsteptau$.
           \popQED
        \end{enumerate}
    \end{enumerate}
\end{proof}

\section{Proof of Modal Characterisation} \label{app:modal}

\begin{proof}[Proof of Theorem \ref{thm:modal characterisation}]
    $(\Rightarrow)$ We are going to prove by structural induction on $\logic_b$ and $\logic_b^r$ that, for all $P,Q \in \closed$, $X \subseteq A$, $\varphi \in \logic_b$ and $\psi \in \logic_b^r$,
    \begin{itemize}
        \item if $P \bisimtbrc Q$ and $P \models \varphi$ then $Q \models \varphi$
        \item if $P \bisimtbrc[X] Q$ and $P \models_X \varphi$ then $Q \models_X \varphi$
        \item if $P \bisimrtbrc Q$ and $P \models \psi$ then $Q \models \psi$
        \item if $P \bisimrtbrc[X] Q$ and $P \models_X \psi$ then $Q \models_X \psi$
    \end{itemize}
    Note that, in the four cases, we dispose of the contraposition. Let $P, Q \in \closed$, $X \subseteq A$, $\varphi \in \logic_b$ and $\psi \in \logic_b^r$.
    \begin{itemize}
        \item If $P \bisimtbrc Q$ and $P \models \varphi$ then
        \begin{itemize}
            \item if $\varphi = \top$ then $Q \models \top$.
            \item if $\varphi = \bigwedge_{i\in I}\varphi_i$ with $(\varphi_i)_{i \in I} \in (\logic_b)^I$ then, for all $i \in I$, $P \models \varphi_i$. Thus, by induction, for all $i \in I$, $Q \models \varphi_i$. Therefore, $Q \models \bigwedge_{i\in I}\varphi_i$.
            \item if $\varphi = \neg\varphi'$ then $P \not\models \varphi'$. Thus, by induction, $Q \not\models \varphi'$. Therefore, $Q \models \neg\varphi'$.
            \item if $\varphi = \langle\epsilon\rangle (\varphi_1\langle\hat{\alpha}\rangle\varphi_2)$ then there exists a path $P \pathtau P_1 \step{\opt{\alpha}} P_2$ such that $P_1 \models \varphi_1$ and $P_2 \models \varphi_2$. Since $P \bisimtbrc Q$, there exists a path $Q \pathtau Q_1 \step{\opt{\alpha}} Q_2$ such that $P_1 \bisimtbrc Q_1$ and $P_2 \bisimtbrc Q_2$. By induction, $Q_1 \models \varphi_1$ and $Q_2 \models \varphi_2$. Therefore, $Q \models \varphi$.
            \item if $\varphi = \langle\epsilon\rangle\langle X\rangle\varphi$ then there is a path $P \pathtau P_1 \step{\rt} P_2$ with $\deadend{P_1}{X}$ and $P_2 \models_X \varphi_2$. Since $P \bisimtbrc Q$, there exists a path $Q \pathtau Q_1$ such that $P_1 \bisimtbrc Q_1$ and $\deadend{Q_1}{X}$ (cf.\ Lemma~\ref{lem:obvious}.4). Thus, there exists a path $Q \pathtau Q_1 \step{\rt} Q_2$ such that $P_2 \bisimtbrc[X] Q_2$. By induction, $Q_2 \models_X \varphi$ so $Q \models \langle\epsilon\rangle\langle X\rangle\varphi$.
            \item if $\varphi = \langle\epsilon\rangle \neg\langle\tau\rangle\top$ then there exists a path $P \pathtau P_0 \nsteptau$. Since $P \bisimtbrc Q$, there exists a path $Q \pathtau Q_0 \nsteptau$. Therefore, $Q \models \varphi$.
        \end{itemize}
        \item If $P \bisimtbrc[X] Q$ and $P \models_X \varphi$ then
        \begin{itemize}
            \item if $\varphi = \top$ then $Q \models_X \top$.
            \item if $\varphi = \bigwedge_{i\in I}\varphi_i$ with $(\varphi_i)_{i \in I} \in (\logic_b)^I$ then, for all $i \in I$, $P \models_X \varphi_i$. Thus, by induction, for all $i \in I$, $Q \models_X \varphi_i$. Therefore, $Q \models_X \bigwedge_{i\in I}\varphi_i$.
            \item if $\varphi = \neg\varphi'$ then $P \not\models_X \varphi'$. Thus, by induction, $Q \not\models_X \varphi'$. Therefore, $Q \models_X \neg\varphi'$.
            \item if $\varphi = \langle\epsilon\rangle (\varphi_1\langle\hat{\alpha}\rangle\varphi_2)$ then
            \begin{itemize}
                \item if $\alpha = \tau$ then there exists a path $P \pathtau P_1 \step{\opt{\tau}} P_2$ such that $P_1 \models_X \varphi_1$ and $P_2 \models_X \varphi_2$. Since $P \bisimtbrc[X] Q$, there exists a path $Q \pathtau Q_1 \step{\opt{\tau}} Q_2$ such that $P_1 \bisimtbrc[X] Q_1$ and $P_2 \bisimtbrc[X] Q_2$. By induction, $Q_1 \models_X \varphi_1$ and $Q_2 \models_X \varphi_2$. Therefore, $Q \models_X \varphi$.
                \item if $\alpha \in A$ then $a \in X$ or $\deadend{P}{X}$ and there exists a path $P \pathtau P_1 \step{a} P_2$ such that $P_1 \models_X \varphi_1$ and $P_2 \models \varphi_2$. Since $P \bisimtbrc[X] Q$, there exists a path $Q \pathtau Q_1 \step{a} Q_2$ such that $P_1 \bisimtbrc[X] Q_1$ and $P_2 \bisimtbrc Q_2$. Moreover, with Lemma~\ref{lem:obvious}.4 we can get that $\deadend{P}{X} \Leftrightarrow \deadend{Q_1}{X}$. By induction, $Q_1 \models_X \varphi_1$ and $Q_2 \models \varphi_2$. Therefore, $Q \models_X \varphi$.
            \end{itemize}
            \item if $\varphi = \langle\epsilon\rangle\langle Y\rangle\varphi$ then there is a path $P \pathtau P_1 \step{\rt} P_2$ with $\deadend{P_1}{X\cup Y}$ and $P_2 \models_Y \varphi$. Since $P \bisimtbrc[X] Q$, there exists a path $Q \pathtau Q_0$ such that $P_1 \bisimtbrc[X] Q_0$, and by Lemma~\ref{lem:obvious}.4 there exists a path $Q_0 \pathtau Q_1$ such that $\deadend{Q_1}{X\cup Y}$ and $P_1 \bisimtbrc Q_1$, and hence $P_1 \bisimtbrc[Y] Q_1$. Thus, there exists a transition $Q_1 \step{\rt} Q_2$ such that $P_2 \bisimtbrc[Y] Q_2$. By induction, $Q_2 \models_Y \varphi$. Therefore, $Q \models_X \langle\epsilon\rangle\langle Y\rangle\varphi$.
            \item if $\varphi = \langle\epsilon\rangle \neg\langle\tau\rangle\top$ then there exists a path $P \pathtau P_0 \nsteptau$. Since $P \bisimtbrc Q$, there exists a path $Q \pathtau Q_0 \nsteptau$. Therefore, $Q \models \varphi$.
        \end{itemize}
        \item If $P \bisimrtbrc Q$ and $P \models \psi$ then
        \begin{itemize}
            \item if $\psi = \top$ then $Q \models \top$.
            \item if $\psi = \bigwedge_{i\in I}\psi_i$ with $(\psi_i)_{i \in I} \in (\logic_b^r)^I$ then, for all $i \in I$, $P \models \psi_i$. Thus, by induction, for all $i \in I$, $Q \models \psi_i$. Therefore, $Q \models \bigwedge_{i\in I}\psi_i$.
            \item if $\psi = \neg\psi'$ then $P \not\models \psi'$. Thus, by induction, $Q \not\models \psi'$. Therefore, $Q \models \neg\psi'$.
            \item if $\psi = \langle\alpha\rangle\varphi$ then there is a transition $P \step{\alpha} P'$ such that $P' \models \varphi$. Since $P \bisimrtbrc Q$, there exists a path $Q \step{\alpha} Q'$ such that $P' \bisimtbrc Q'$. By induction, $Q' \models \varphi$. Therefore, $Q \models \psi$.
            \item if $\psi = \langle X\rangle\varphi$ then $\deadend{P}{X}$ and there exists a transition $P \step{\rt} P'$ such that $P' \models_X \varphi$. Since $P \bisimrtbrc Q$, $\deadend{Q}{X}$ and there exists a path $Q \step{\rt} Q'$ such that $P' \bisimtbrc[X] Q'$. By induction, $Q' \models_X \varphi$. Therefore, $Q \models \psi$.
        \end{itemize}
        \item If $P \bisimrtbrc[X] Q$ and $P \models_X \psi$ then
        \begin{itemize}
            \item if $\psi = \top$ then $Q \models_X \top$.
            \item if $\psi = \bigwedge_{i\in I}\psi_i$ with $(\psi_i)_{i \in I} \in (\logic_b^r)^I$ then, for all $i \in I$, $P \models_X \psi_i$. Thus, by induction, for all $i \in I$, $Q \models_X \psi_i$. Therefore, $Q \models_X \bigwedge_{i\in I}\psi_i$.
            \item if $\psi = \neg\psi'$ then $P \not\models_X \psi'$. Thus, by induction, $Q \not\models_X \psi'$. Therefore, $Q \models_X \neg\psi'$.
            \item if $\psi = \langle\alpha\rangle\varphi$
            \begin{itemize}
                \item if $\alpha = \tau$ then there exists a transition $P \step{\tau} P'$ such that $P' \models_X \varphi$. Since $P \bisimrtbrc[X] Q$, there exists a transition $Q \step{\tau} Q'$ such that $P' \bisimtbrc[X] Q'$. By induction, $Q' \models_X \varphi$. Therefore, $Q \models_X \psi$.
                \item if $\alpha \in A$ then $a \in X$ or $\deadend{P}{X}$ and there exists a transition $P \step{a} P'$ such that $P' \models \varphi$. Since $P \bisimrtbrc[X] Q$, $\deadend{P}{X} \Leftrightarrow \deadend{Q}{X}$ and there exists a transition $Q \step{a} Q'$ such that $P' \bisimtbrc Q'$. By induction, $Q' \models \varphi$. Therefore, $Q \models_X \psi$.
            \end{itemize}
            \item if $\psi = \langle Y\rangle\varphi$ then $\deadend{P}{(X\cup Y)}$ and there exists a transition $P \step{\rt} P'$ such that $P_1 \models_Y \varphi$. Since $P \bisimrtbrc[X] Q$, $\deadend{Q}{(X\cup Y)}$ and there exists a transition $Q \step{\rt} Q'$ such that $P' \bisimtbrc[Y] Q'$. By induction, $Q' \models_Y \varphi$. Therefore, $Q \models_X \psi$.
        \end{itemize}
    \end{itemize}

    $(\Leftarrow)$ Let $\equiv \; := \{(P,Q) \mid \forall \varphi \in \logic^c_b, P \models \varphi \Leftrightarrow Q \models \varphi\} \cup \{(P,X,Q) \mid \forall \varphi \in \logic^c_b, P \models_X \varphi \Leftrightarrow Q \models_X \varphi\}$, and $\equiv^r \; := \{(P,Q) \mid \forall \psi \in \logic_b^{cr}, P \models \psi \Leftrightarrow Q \models \psi\} \cup \{(P,X,Q) \mid \forall \psi \in \logic_b^{cr}, P \models_X \psi \Leftrightarrow Q \models_X \psi\}$. $(P,X,Q) \in {\equiv}$ will be denoted $P \equiv_X Q$ for clarity. Note that ${\equiv^r} \subseteq {\equiv}$. We are going to check that $\equiv$ is a generalised \tb bisimulation and $\equiv^r$ a generalised rooted \tb reactive bisimulation. Let $P,Q \in \closed$ and $X \subseteq A$.
    \begin{enumerate}
        \item If $P \equiv Q$
        \begin{enumerate}
            \item if $P \step{\alpha} P'$ then define $\mathcal{Q}^\dag := \{Q^\dag \mid Q \pathtau Q^\dag \wedge P \not\equiv Q^\dag\}$ and $\mathcal{Q}^\ddag := \{Q^\ddag \mid Q \pathtau Q^\dag \step{\opt{\alpha}} Q^\ddag \wedge P' \not\equiv Q^\ddag\}$. Since $\logic_b$ is closed under negation and conjunction, there exist two formulas $\varphi^\dag, \varphi^\ddag \in \logic_b$ such that $P \models \varphi^\dag$, $P' \models \varphi^\ddag$, for all $Q^\dag \in \mathcal{Q}^\dag$, $Q^\dag \not\models \varphi^\dag$ and, for all $Q^\ddag \in \mathcal{Q}^\ddag$, $Q^\ddag \not\models \varphi^\ddag$. Note that $P \models \langle\epsilon\rangle(\varphi^\dag\langle\hat{\alpha}\rangle\varphi^\ddag)$. Thus, $Q \models \langle\epsilon\rangle(\varphi^\dag\langle\hat{\alpha}\rangle\varphi^\ddag)$. Therefore, there exists a path $Q \pathtau Q_1 \step{\opt{\alpha}} Q_2$ such that $Q_1 \models \varphi^\dag$ and $Q_2 \models \varphi^\ddag$. By definition of $\mathcal{Q}^\dag$ and $\mathcal{Q}^\ddag$, $P \equiv Q_1$ and $P' \equiv Q_2$.
            \item if $\deadend{P}{X}$ and $P \step{\rt} P'$ then define $\mathcal{Q} := \{Q_2 \mid Q \pathtau Q_1 \step{\rt} Q_2 \wedge P' \not\equiv_X Q_2\}$. Since $\logic_b^c$ is closed under negation and conjunction, there exists a formula $\varphi \in \logic^c_b$ such that $P' \models_X \varphi\!$ and, for all $Q \mathbin\in \mathcal{Q}$, $Q^\dag \mathbin{\not\models_X} \varphi$. Note that $P \models \langle\epsilon\rangle\langle X\rangle\varphi$. Thus, $Q \models \langle\epsilon\rangle\langle X\rangle\varphi$. Therefore, there exists a path $Q \pathtau Q_1 \step{\rt} Q_2$ with $\deadend{Q_1}{X}$ and $Q_2 \models_X \varphi$. By definition of $\mathcal{Q}$, $P' \equiv_X Q_2$.
            \item if $P \nsteptau$ then $P \models \langle\epsilon\rangle \neg\langle\tau\rangle\top$. Thus $Q \models \langle\epsilon\rangle \neg\langle\tau\rangle\top$. Therefore, $Q \pathtau Q_1 \nsteptau$.
       \end{enumerate}
        \item If $P \equiv_X Q$ then
        \begin{enumerate}
            \item if $P \step{\tau} P'$ then define $\mathcal{Q}^\dag := \{Q^\dag \mid Q \pathtau Q^\dag \wedge P \not\equiv_X Q^\dag\}$ and $\mathcal{Q}^\ddag := \{Q^\ddag \mid Q \pathtau Q^\ddag\linebreak[2] \wedge P' \not\equiv_X Q^\ddag\}$. Since $\logic_b$ is closed under negation and conjunction, there exist two formulas $\varphi^\dag, \varphi^\ddag \in \logic_b$ such that $P \models_X \varphi^\dag$, $P' \models_X \varphi^\ddag$, for all $Q^\dag \in \mathcal{Q}^\dag$, $Q^\dag \not\models_X \varphi^\dag$ and, for all $Q^\ddag \in \mathcal{Q}^\ddag$, $Q^\ddag \not\models_X \varphi^\ddag$. Note that $P \models_X \langle\epsilon\rangle(\varphi^\dag\langle\hat{\tau}\rangle\varphi^\ddag)$. Thus, $Q \models_X \langle\epsilon\rangle(\varphi^\dag\langle\hat{\tau}\rangle\varphi^\ddag)$. Therefore, there exists a path $Q \pathtau Q_1 \step{\opt{\tau}} Q_2$ such that $Q_1 \models_X \varphi^\dag$ and $Q_2 \models_X \varphi^\ddag$. By definition of $\mathcal{Q}^\dag$ and $\mathcal{Q}^\ddag$, $P \equiv_X Q_1$ and $P' \equiv_X Q_2$.
            \item if $P \step{a} P'$ with $a \in X$ or $\deadend{P}{X}$ then define $\mathcal{Q}^\dag := \{Q^\dag \mid Q \pathtau Q^\dag\linebreak[2] \wedge P \not\equiv_X Q^\dag\}$ and $\mathcal{Q}^\ddag := \{Q^\ddag \mid Q \pathtau Q^\dag \step{a} Q^\ddag \wedge P' \not\equiv Q^\ddag\}$. Since $\logic_b$ is closed under negation and conjunction, there exist two formulas $\varphi^\dag, \varphi^\ddag \in \logic_b$ such that $P \models_X \varphi^\dag$, $P' \models \varphi^\ddag$, for all $Q^\dag \in \mathcal{Q}^\dag$, $Q^\dag \not\models_X \varphi^\dag$ and, for all $Q^\ddag \in \mathcal{Q}^\ddag$, $Q^\ddag \not\models \varphi^\ddag$. Note that $P \models_X \langle\epsilon\rangle(\varphi^\dag\langle\hat{\alpha}\rangle\varphi^\ddag)$. Thus, $Q \models_X \langle\epsilon\rangle(\varphi^\dag\langle\hat{\alpha}\rangle\varphi^\ddag)$. Therefore, there exists a path $Q \pathtau Q_1 \step{\opt{\alpha}} Q_2$ such that $a \in X \vee \deadend{Q_1}{X}$, $Q_1 \models_X \varphi^\dag$ and $Q_2 \models \varphi^\ddag$. By definition of $\mathcal{Q}^\dag$ and $\mathcal{Q}^\ddag$, $P \equiv_X Q_1$ and $P' \equiv Q_2$.
            \item if $\deadend{P}{X\cup Y}$ and $P \step{\rt} P'$ then define $\mathcal{Q} := \{Q_2 \mid Q \pathtau Q_1 \step{\rt} Q_2\linebreak[2] \wedge P' \not\equiv_Y Q_2\}$. Since $\logic_b^c$ is closed under negation and conjunction, there exists a formula $\varphi \in \logic^c_b$ such that $P' \models_Y \varphi\!$ and, for all $Q \mathbin\in \mathcal{Q}$, $Q^\dag \mathbin{\not\models_Y} \varphi$. Note that $P \models_X \langle\epsilon\rangle\langle Y\rangle\varphi$. Thus, $Q \models_X \langle\epsilon\rangle\langle Y\rangle\varphi$. Therefore, there exists a path $Q \pathtau Q_1 \step{\rt} Q_2$ with $\deadend{Q_1}{X\cup Y}$ and $Q_2 \models_Y \varphi$. By definition of $\mathcal{Q}$, $P' \equiv_Y Q_2$.
            \item if $P \nsteptau$ then $P \models_X \langle\epsilon\rangle \neg\langle\tau\rangle\top$. Thus $Q \models_X \langle\epsilon\rangle \neg\langle\tau\rangle\top$. Therefore, $Q \pathtau Q_1 \nsteptau$.
        \end{enumerate}
    \end{enumerate}
    \begin{enumerate}
        \item If $P \equiv^r Q$ then
        \begin{enumerate}
            \item if $P \step{\alpha} P'$ with $\alpha \in A_\tau$ then define $\mathcal{Q}^\ddag := \{Q^\ddag \mid Q \step{\alpha} Q^\ddag \wedge P' \not\equiv Q^\ddag\}$. Since $\logic_b^r$ is closed under negation and conjunction, there exist a formula $\varphi^\ddag \in \logic_b^r$ such that $P' \models \varphi^\ddag$ and, for all $Q^\ddag \in \mathcal{Q}^\ddag$, $Q^\ddag \not\models \varphi^\ddag$. Note that $P \models \langle\alpha\rangle\varphi^\ddag$. Thus, $Q \models \langle\alpha\rangle\varphi^\ddag$. Therefore, there\linebreak[3] exists a transition $Q \step{\alpha} Q'$ such that $Q' \models \varphi^\ddag$. By definition of $\mathcal{Q}^\ddag$, $P' \equiv Q'$.
            \item if $\deadend{P}{X}$ and $P \step{\rt} P'$ then define $\mathcal{Q}^\ddag := \{Q^\ddag \mid Q \step{\rt} Q^\ddag \wedge P' \not\equiv_X Q^\ddag\}$. Since $\logic_b^r$ is closed under negation and conjunction, there exist a formula $\varphi^\ddag \in \logic_b^r$ such that $P' \models_X \varphi^\ddag$ and, for all $Q^\ddag \in \mathcal{Q}^\ddag$, $Q^\ddag \not\models_X \varphi^\ddag$. Note that $P \models \langle X\rangle\varphi^\ddag$. Thus, $Q \models \langle X\rangle\varphi^\ddag$. Therefore, there exists a transition $Q \step{\rt} Q'$ such that $Q' \models_X \varphi^\ddag$. By definition of $\mathcal{Q}^\ddag$, $P' \equiv_X Q'$.
        \end{enumerate}\vspace{2ex}
        \item If $P \equiv^r_X Q$ then\vspace{-3ex}
        \begin{enumerate}
            \item if $P \step{\tau} P'$ then define $\mathcal{Q}^\ddag := \{Q^\ddag \mid Q \step{\tau} Q^\ddag \wedge P' \not\equiv_X Q^\ddag\}$. Since $\logic_b^r$ is closed under negation and conjunction, there exist a formula $\varphi^\ddag \in \logic_b^r$ such that $P' \models_X \varphi^\ddag$ and, for all $Q^\ddag \in \mathcal{Q}^\ddag$, $Q^\ddag \not\models_X \varphi^\ddag$. Note that $P \models_X \langle\tau\rangle\varphi^\ddag$. Thus, $Q \models_X \langle\tau\rangle\varphi^\ddag$. Therefore, there exists a transition $Q \step{\tau} Q'$ such that $Q' \models_X \varphi^\ddag$. By definition of $\mathcal{Q}^\ddag$, $P' \equiv_X Q'$.
            \item if $P \step{a} P'$ with $a \in X \vee \deadend{P}{X}$ then define $\mathcal{Q}^\ddag := \{Q^\ddag \mid Q \step{a} Q^\ddag \wedge P' \not\equiv Q^\ddag\}$. Since $\logic_b^r$ is closed under negation and conjunction, there exist a formula $\varphi^\ddag \in \logic_b^r$ such that $P' \models \varphi^\ddag$ and, for all $Q^\ddag \in \mathcal{Q}^\ddag$, $Q^\ddag \not\models \varphi^\ddag$. Note that $P \models_X \langle a\rangle\varphi^\ddag$. Thus, $Q \models_X \langle a\rangle\varphi^\ddag$. Therefore, there exists a path $Q \step{\alpha} Q'$ such that $Q' \models \varphi^\ddag$. By definition of $\mathcal{Q}^\ddag$, $P' \equiv Q'$.
            \item if $\deadend{P}{(X\cup Y)}$ and $P \step{\rt} P'$ then define $\mathcal{Q}^\ddag := \{Q^\ddag \mid Q \step{\rt} Q^\ddag \wedge P' \not\equiv_Y Q^\ddag\}$. Since $\logic_b^r$ is closed under negation and conjunction, there exist a formula $\varphi^\ddag \in \logic_b^r$ such that $P' \models_Y \varphi^\ddag$ and, for all $Q^\ddag \in \mathcal{Q}^\ddag$, $Q^\ddag \not\models_Y \varphi^\ddag$. Note that $P \models_X \langle Y\rangle\varphi^\ddag$. Thus, $Q \models_X \langle Y\rangle\varphi^\ddag$. Therefore, there exists a path $Q \step{\rt} Q'$ such that $Q' \models_Y \varphi^\ddag$. By definition of $\mathcal{Q}^\ddag$, $P' \equiv_Y Q'$.
        \popQED
        \end{enumerate}
    \end{enumerate}
\end{proof}

\section{Correctness of Time-out Bisimulation} \label{app:time-out}

\begin{proof}[Proof of Proposition \ref{prop:time-out bisim}]
    Let $\R$ be a \tb reactive bisimulation, let's define
    \begin{align*}
        \tbisim := \{(P,Q) \mid \R(P,Q)\} \cup \{(\theta_X(P),\theta_X(Q)) \mid \R(P,X,Q)\}
    \end{align*}
    We are going to show that $\tbisim$ is a \tb time-out bisimulation. Let $P,Q \in \closed$ such that $P \tbisim Q$. By definition of $\tbisim$, $\R(P,Q)$ or $P = \theta_X(P^\dag)$, $Q = \theta_X(Q^\dag)$ and $\R(^\dag,X,Q^\dag)$.
    \begin{enumerate}
        \item If $P \step{\alpha} P'$ with $\alpha \in A_\tau$ then 
        \begin{itemize}
            \item if $\R(P,Q)$ then there exists a path $Q \pathtau Q_1 \step{\opt{\alpha}} Q_2$ such that $\R(P,Q_1)$ and $\R(P',Q_2)$. Thus, by definition of $\tbisim$, $P \tbisim Q_1$ and $P \tbisim Q_2$.
            \item if $P = \theta_X(P^\dag)$, $Q = \theta_X(Q^\dag)$ and $\R(P^\dag,X,Q^\dag)$ then
            \begin{itemize}
                \item if $\alpha = \tau$ then, by the semantics, $P' = \theta_X(P^\ddag)$ and $P^\dag \steptau P^\ddag$. Since $\R(P^\dag,X,Q^\dag)$, there exists a path $Q^\dag \pathtau Q^\dag_1 \step{\opt{\tau}} Q^\dag_2$ such that $\R(P^\dag,X,Q^\dag_1)$ and $\R(P^\ddag,X,Q^\dag_2)$. By the semantics, there exists a path $Q \pathtau \theta_X(Q^\dag_1) \step{\opt{\tau}} \theta_X(Q^\dag_2)$ such that, by the definition of $\tbisim$, $P \tbisim \theta_X(Q^\dag_1)$ and $P' \tbisim \theta_X(Q^\dag_2)$.
                \item if $\alpha = a \in A$ then, by the semantics, $P^\dag \step{a} P'$ and $a \in X \vee \deadend{P^\dag}{X}$.
                \begin{itemize}
                   \item if $a \in X$ then, since $\R(P^\dag,X,Q^\dag)$, there exists a path $Q^\dag \pathtau Q^\dag_1 \step{a} Q_2$ such that $\R(P^\dag,X,Q^\dag_1)$ and $\R(P',Q_2)$.  By the semantics, there exists a path $Q \pathtau \theta_X(Q^\dag_1) \step{a} Q_2$ such that, by the definition of $\tbisim$, $P \tbisim \theta_X(Q^\dag_1)$ and $P' \tbisim Q_2$.
                   \item if $\deadend{P^\dag}{X}$, then there is a path $Q^\dag \pathtau Q^\dag_0\nsteptau$ with $\R(P^\dag,Q^\dag_0)$. Now there exists a path $Q^\dag_0 \pathtau Q^\dag_1 \step{a} Q_2$ such that $\R(P^\dag,Q^\dag_{1})$ and $\R(P',Q_2)$. Moreover, we find that $\deadend{P^\dag}{X} \Leftrightarrow \deadend{Q^\dag_1}{X}$. By the semantics, there exists a path $Q \pathtau \theta_X(Q^\dag_1) \step{a} Q_2$ such that, by the definition of $\tbisim$, $P \tbisim \theta_X(Q_1)$ and $P' \tbisim Q_2$.
                \end{itemize}
            \end{itemize}
        \end{itemize}
        \item If $\deadend{P}{X}$ and $P \step{\rt} P'$ then
        \begin{itemize}
            \item if $\R(P,Q)$ then $\R(P,X,Q)$, so there exists a path $Q \pathtau Q_1 \step{\rt} Q_2$ with $\R(P',X,Q_2)$. By definition of $\tbisim$, $\theta_X(P') \tbisim \theta_X(Q_2)$.
            \item if $P = \theta_Y(P^\dag)$, $Q = \theta_Y(Q^\dag)$ and $\R(P^\dag,Y,Q^\dag)$ then, by the semantics of $\theta_Y$, $\deadend{P^\dag}{Y}$ and $P^\dag \step{\rt} P'$. Therefore, using Lemma~\ref{lem:obvious}.4, there exists a path $Q^\dag \pathtau Q_1 \step{\rt} Q_2$ with $\deadend{Q_1}{X \cup Y}$ and $\R(P',X,Q_2)$. By the semantics,\linebreak there exists a path $Q \pathtau \theta_Y(Q_1) \step{\rt} Q_2$ with, by the definition of $\tbisim$, $\theta_X(P') \tbisim \theta_X(Q_2)$.
        \end{itemize}
        \item If $P \nsteptau$ then
        \begin{itemize}
            \item if $\R(P,Q)$ then $\R(P,\emptyset,Q)$, so there exists a path $Q \pathtau Q_0 \nsteptau$.
            \item if $P = \theta_X(P^\dag)$, $Q = \theta_X(Q^\dag)$ and $\R(P^\dag,X,Q^\dag)$ then, by the semantics, $P^\dag \nsteptau$. Since $\R(P^\dag,X,Q^\dag)$, there exists a path $Q^\dag \pathtau Q_0 \nsteptau$. By the semantics, there exists a path $Q \pathtau \theta_X(Q_0) \nsteptau$.
        \end{itemize}
    \end{enumerate}
    Let $\tbisim$ be a \tb time-out bisimulation, let's define
    \begin{align*}
        \R = \{(P,Q) \mid P \tbisim Q\} \cup \{(P,X,Q) \mid \theta_X(P) \tbisim \theta_X(Q)\}
    \end{align*}
    We are going to show that $\R$ is a generalised \tb reactive bisimulation. Let $P,Q \in \closed$ and $X \subseteq A$.
    \begin{enumerate}
        \item If $\R(P,Q)$ then $P \tbisim Q$.
        \begin{enumerate}
            \item If $P \step{\alpha} P'$ then there exists a path $Q \pathtau Q_1 \step{\opt{\alpha}} Q_2$ such that $P \tbisim Q_1$ and $P' \tbisim Q_2$, thus, by definition of $\R$, $\R(P,Q_1)$ and $\R(P',Q_2)$.
            \item If $\deadend{P}{X}$ and $P \step{\rt} P'$ then there exists a path $Q \pathtau Q_1 \step{\rt} Q_2$ with $\theta_X(P') \tbisim \theta_X(Q_2)$. Thus, by definition of $\R$, $\R(P',X,Q_2)$.
            \item If $P \nsteptau$ then there exists a path $Q \pathtau Q_0 \nsteptau$ such that $P \tbisim Q_0$, thus, by definition of $\R$, $\R(P,Q_0)$.
        \end{enumerate}
        \item If $\R(P,X,Q)$ then $\theta_X(P) \tbisim \theta_X(Q)$.
        \begin{enumerate}
            \item If $P \steptau P'$ then, by the semantics, $\theta_X(P) \steptau \theta_X(P')$. Therefore, there exists a path $\theta_X(Q) \pathtau Q^\dag \step{\opt{\tau}} Q^\ddag$ such that $\theta_X(P) \tbisim Q^\dag$ and $\theta_X(P') \tbisim Q^\ddag$. By the semantics, $Q^\dag = \theta_X(Q_1)$, $Q^\ddag = \theta_X(Q_2)$ and $Q \pathtau Q_1 \step{\opt{\tau}} Q_2$. Moreover, by definition of $\R$, $\R(P,X,Q_1)$ and $\R(P',X,Q_2)$.
            \item If $P \step{a} P'$ with $a \in X \vee \deadend{P}{X}$ then, by the semantics, $\theta_X(P) \step{a} P'$. Therefore, there exists a path $\theta_X(Q) \pathtau Q^\dag \step{a} Q_2$ such that $\theta_X(P) \tbisim Q^\dag$ and $P' \tbisim Q_2$. By the semantics, $Q^\dag = \theta_X(Q_1)$ and $Q \pathtau Q_1 \step{a} Q_2$. Moreover, by definition of $\R$, $\R(P,X,Q_1)$ and $\R(P',Q_2)$.
            \item If $\deadend{P}{(X\cup Y)}$ and $P \step{\rt} P'$ then, by the semantics, $\deadend{\theta_X(P)}{Y}$ and $\theta_X(P) \step{\rt} P'$. Therefore, $\theta_X(Q) \pathtau Q^\dag_1 \step{\rt} Q_2$ with $\theta_Y(P') \tbisim \theta_Y(Q_2)$. By the semantics, $Q^\dag_1 = \theta_X(Q_1)$ and we have $Q \pathtau Q_1 \step{\rt} Q_2$. Moreover, by definition of $\R$, and $\R(P',Y,Q_2)$.
            \item If $P \nsteptau$ then, by the semantics, $\theta_X(P) \nsteptau$. Therefore, there exists a path $\theta_X(Q) \pathtau Q^\dag \nsteptau$ such that $\theta_X(P) \tbisim Q^\dag$. By the semantics, $Q^\dag = \theta_X(Q_0)$ and $Q \pathtau Q_0 \nsteptau$. Moreover, by definition of $\R$, $\R(P,X,Q_0)$.
        \end{enumerate}
    \end{enumerate}
    This ends the proof of Proposition~\ref{prop:time-out bisim}.1, and thereby its corollary~\ref{prop:time-out bisim}.2.\\
    Let $\R$ be a generalised rooted \tb reactive bisimulation, let's define
    \begin{align*}
        \tbisim := \{(P,Q) \mid \R(P,Q)\} \cup \{(\theta_X(P),\theta_X(Q)) \mid \R(P,X,Q)\}
    \end{align*}
    We are going to show that $\tbisim$ is a rooted \tb time-out bisimulation. Let $P,Q \in \closed$ such that $P \tbisim Q$, by definition of $\tbisim$, $\R(P,Q)$ or $P = \theta_X(P^\dag)$, $Q = \theta_X(Q^\dag)$ and $\R(P,X,Q)$.
    \begin{enumerate}
        \item If $P \step{\alpha} P'$ with $\alpha \in A_\tau$ then
        \begin{itemize}
            \item if $\R(P,Q)$ then there exists a transition $Q \step{\alpha} Q'$ such that $P'  \bisimtbrc Q'$.
            \item if $P = \theta_X(P^\dag)$, $Q = \theta_X(Q^\dag)$ and $\R(P^\dag,X,Q^\dag)$ then 
            \begin{itemize}
                \item if $\alpha = \tau$ then, by the semantics, $P' = \theta_X(P^\ddag)$ and $P^\dag \steptau P^\ddag$. Since $\R(P^\dag,X,Q^\dag)$, there exists a transition $Q^\dag \steptau Q^\ddag$ such that $P^\ddag  \bisimtbrc[X] Q^\ddag$. By the semantics, there exists a transition $Q \steptau \theta_X(Q^\ddag)$. Moreover, by Proposition~\ref{prop:time-out bisim}.\ref{corr}, $P'  \bisimtbrc \theta_X(Q^\ddag)$.
                \item if $\alpha = a \in A$ then, by the semantics, $P^\dag \step{a} P'$ and $a \in X \vee \deadend{P^\dag}{X}$. Since $\R(P^\dag,X,Q^\dag)$, there exists a transition $Q^\dag \step{a} Q'$ such that $P'  \bisimtbrc Q'$. Moreover, $\deadend{P^\dag}{X} \Leftrightarrow \deadend{Q^\dag}{X}$. By the semantics, there exists a transition $Q \step{a} Q'$ such that $P'  \bisimtbrc Q'$.
            \end{itemize}
        \end{itemize}
        \item If $\deadend{P}{X}$ and $P \step{\rt} P'$ then
        \begin{itemize}
            \item if $\R(P,Q)$ then there exists a transition $Q \step{\rt} Q'$ such that $P' \bisimtbrc[X] Q'$. Thus, $\theta_X(P')  \bisimtbrc \theta_X(Q')$ by Proposition~\ref{prop:time-out bisim}.\ref{corr}.
            \item if $P = \theta_Y(P^\dag)$, $Q = \theta_Y(Q^\dag)$ and $\R(P^\dag,Y,Q^\dag)$ then, by the semantics, $P^\dag \step{\rt} P'$ and $\deadend{P^\dag}{Y}$. Since $\R(P^\dag,Y,Q^\dag)$, $\deadend{P}{(X\cup Y)}$ and $P^\dag \step{\rt} P'$, there exists a transition $Q^\dag \step{\rt} Q'$ such that $P' \bisimtbrc[X] Q'$. Thus, $\theta_X(P')  \bisimtbrc \theta_X(Q')$. Moreover, $\deadend{Q^\dag}{X}$. By the semantics, there exists a transition $Q \step{\rt} Q'$ such that $\theta_X(P')  \bisimtbrc \theta_X(Q')$.
        \end{itemize}
    \end{enumerate}
    Let $\tbisim$ be a rooted \tb time-out bisimulation, let's define
    \begin{align*}
        B := \{(P,Q) \mid P \tbisim Q\} \cup \{(P,X,Q) \mid \theta_X(P) \tbisim \theta_X(Q)\}
    \end{align*}
    We are going to show that $\R$ is a generalised rooted \tb reactive bisimulation. Let $P,Q \in \closed$ and $X \subseteq A$.
    \begin{enumerate}
        \item If $\R(P,Q)$ then $P \tbisim Q$.
        \begin{enumerate}
            \item If $P \step{\alpha} P'$ with $\alpha \in A_\tau$ then there exists a transition $Q \step{\alpha} Q'$ such that $P'  \bisimtbrc Q'$.
            \item If $\deadend{P}{X}$ and $P \step{\rt} P'$ then there exists a transition $Q \step{\rt} Q'$ such that $\theta_X(P')  \bisimtbrc \theta_X(Q')$. Thus, $P' \bisimtbrc[X] Q'$, by Proposition~\ref{prop:time-out bisim}.\ref{corr}.
        \end{enumerate}
        \item If $\R(P,X,Q)$ then $\theta_X(P) \tbisim \theta_X(Q)$.
        \begin{enumerate}
            \item If $P \steptau P'$ then, by the semantics, $\theta_X(P) \steptau \theta_X(P')$. Since $\theta_X(P) \tbisim \theta_X(Q)$, there exists a transition $\theta_X(Q) \steptau Q^\ddag$ such that $\theta_X(P')  \bisimtbrc Q^\ddag$. By the semantics, $Q^\ddag = \theta_X(Q')$ and there exists a transition $Q \steptau Q'$. By Proposition~\ref{prop:time-out bisim}.\ref{corr}, $P' \bisimtbrc[X] Q'$.
            \item If $P \step{a} P'$ with $a \in X \vee \deadend{P}{X}$ then, by the semantics, $\theta_X(P) \step{a} P'$. Since $\theta_X(P) \tbisim \theta_X(Q)$, there exists a transition $\theta_X(Q) \step{a} Q'$ such that $P'  \bisimtbrc Q'$. By the semantics, there exists a transition $Q \step{a} Q'$ such that $P'  \bisimtbrc Q'$.
            \item If $\deadend{P}{(X\cup Y)}$ and $P \step{\rt} P'$ then, by the semantics, $\theta_X(P) \step{\rt} P'$ and $\deadend{P}{Y}$. Since $\theta_X(P) \tbisim \theta_X(Q)$, there exists a transition $\theta_X(Q) \step{\rt} Q'$ such that $\theta_Y(P')  \bisimtbrc \theta_Y(Q')$. By the semantics, there exists a transition $Q \step{\rt} Q'$. By Proposition~\ref{prop:time-out bisim}.\ref{corr}, $P' \bisimtbrc[Y] Q'$.
        \popQED
        \end{enumerate}
    \end{enumerate}
\end{proof}

\section{Congruence Proofs for \texorpdfstring{$\bisimtbrc$ and $\bisimb$}{Branching Reactive Bisimilarity}}  \label{app:stability}

To prove congruence properties, the notion of bisimulation \emph{up to}, introduced by Milner in \cite{Mi90ccs}, is going to be helpful. Let $\bisim\,$ denote the classical notion of strong bisimilarity~\cite{Mi90ccs}:\linebreak[3] A \emph{(strong) bisimulation} is a symmetric relation ${\R} \subseteq \closed\times\closed$ such that, for all $P,Q \in \closed$ with $P \mathrel\R Q$, if $P \step{\alpha} P'$ with $\alpha \in Act$ then there is a transition $Q \step{\alpha} Q'$ such that $P' \mathrel\R Q'$; write $P \bisim Q$ if $P \mathrel\R Q$ for some strong bisimulation $\R$.
\begin{definition}\rm \label{def:up to}
    A \emph{\tb time-out bisimulation up to $\bisim$\,} is a symmetric relation ${\tbisim} \subseteq \closed\times\closed$ such that, for all $P,Q \in \closed$ with $P \tbisim Q$,
    \begin{enumerate}
        \item if $P \step{\alpha} P'$ with $\alpha \in A_\tau$ then there exists a path $Q \pathtau Q_1 \step{\opt{\alpha}} Q_2$ such that $P \upto[\bisim] Q_1$ and $P' \upto[\bisim] Q_2$
        \item if $\deadend{P}{X}$ and $P \step{\rt} P'$ then there exists a path $Q \pathtau Q_1 \step{\rt} Q_2$ with $\theta_X(P') \upto[\bisim] \theta_X(Q_2)$
        \item if $P \nsteptau$ then there exists a path $Q \pathtau Q_0 \nsteptau$,
    \end{enumerate}
    where $\upto[\bisim]$ stands for the relational composition $\bisim \circ \tbisim \circ \bisim$\,.
\end{definition}

\begin{proposition} \label{prop:up to}
    Let $P,Q \in \closed$.  Then $P \bisimtbrc Q$ iff there exists a \tb time-out bisimulation $\mathcal{B}$ up to $\bisim$ such that $P \tbisim Q$.
\end{proposition}

\begin{proof}
    First of all, a \tb time-out bisimulation is a \tb time-out bisimulation up to $\bisim$ by reflexivity of $\bisim$\,. Conversely, let $\tbisim$ be a \tb bisimulation up to $\bisim$. We are going to show that $\upto[\bisim]$ is a \tb time-out bisimulation. By the reflexivity of $\bisimtbrc$ this will suffice. Let $P,Q \in \closed$ such that $P \upto[\bisim] Q$. Then there exists $P^\dag, Q^\dag \in \closed$ such that $P \bisim P^\dag \tbisim Q^\dag \bisim Q$.
    \begin{enumerate}
        \item If $P \step{\alpha} P'$ with $\alpha \in A_\tau$ then, since $P \bisim P^\dag$, there exists a transition $P^\dag \step{\alpha} P^\ddag$ such that $P' \bisim P^\ddag$. Since $P^\dag \tbisim Q^\dag$, there exists a path $Q^\dag \pathtau Q^\star \step{\opt{\alpha}} Q^\ddag$ such that $P^\dag \upto[\bisim] Q^\star$ and $P^\ddag \upto[\bisim] Q^\ddag$. Since $Q^\dag \bisim Q$, there exists a path $Q \pathtau Q_1 \step{\opt{\alpha}} Q_2$ such that $Q^\star \bisim Q_1$ and $Q^\ddag \bisim Q_2$. Since $\bisim$ is transitive, $P \upto[\bisim] Q_1$ and $P' \upto[\bisim] Q_2$.
        \item If $\deadend{P}{X}$ and $P \step{\rt} P'$ then, since $P \bisim P^\dag$, $\deadend{P^\dag}{X}$ and there exists a transition $P^\dag \step{\rt} P^\ddag$ such that $P' \bisim P^\ddag$. Since $P^\dag \tbisim Q^\dag$, there exists a path $Q^\dag \pathtau Q^\dag_1 \step{\rt} Q^\dag_2$ with $\theta_X(P^\ddag) \upto[\bisim] \theta_X(Q^\dag_2)$. Since $Q^\dag \bisim Q$, there exists a path $Q \pathtau Q_1 \step{\rt} Q_2$ such that $Q^\dag_2 \bisim Q_2$. Since $\bisim$ is transitive and a congruence for $\theta_X$ \cite{strongreactivebisimilarity}, $\theta_X(P') \upto[\bisim] \theta_X(Q_2)$.
        \item If $P \nsteptau$ then, since $P \bisim P^\dag$, $P^\dag \nsteptau$. Since $P^\dag \tbisim Q^\dag$, there exists a path $Q^\dag \pathtau Q^\star \nsteptau$. Since $Q^\dag \bisim Q$, there exists a path $Q \pathtau Q_0 \nsteptau$ such that $Q^\star \bisim Q_0$.
    \popQED
    \end{enumerate}
\end{proof}

\noindent
The following lemma was proven in \cite[Appendix B]{strongreactivebisimilarity}. It will be useful in the proof of Proposition \ref{prop:stability}.

\begin{lemma}\label{lem:strong identities}
    Let $P,Q \in \closed$, $X,S,I \subseteq A$, $\rename \subseteq A\times A$.
    \begin{itemize}
        \item If $P \nsteptau$ and $\init{P}\cap X \subseteq S$ then $\theta_X(P \parallel_S Q) \bisim \theta_X(P \parallel_S \theta_{X \setminus (S \setminus \init{P})}(Q))$.
        \item $\theta_X(\tau_I(P)) \bisim \theta_X(\tau_I(\theta_{X \cup I}(P)))$.
        \item $\theta_X(\rename(P)) \bisim \theta_X(\rename(\theta_{\rename^{-1}(X)}(P)))$.
    \end{itemize}
\end{lemma}

\begin{proof}[Proof of Proposition \ref{prop:stability}]
    Let $\tbisim$ be the smallest relation satisfying, for all $P,Q \in \closed$,
    \begin{itemize}
        \item if $P \bisimtbrc Q$ then $P \tbisim Q$
        \item if $P \tbisim Q$ and $\alpha \in Act$ then $\alpha.P \tbisim \alpha.Q$
        \item if $P_1 \tbisim Q_1$, $P_2 \tbisim Q_2$ and $S \subseteq A$ then $P_1 \parallel_S P_2 \tbisim Q_1 \parallel_S Q_2$
        \item if $P \tbisim Q$ and $I \subseteq A$ then $\tau_I(P) \tbisim \tau_I(Q)$
        \item if $P \tbisim Q$ and $\rename \subseteq A\times A$ then $\rename(P) \tbisim \rename(Q)$
        \item if $P \tbisim Q$ and $L \subseteq U \subseteq A$ then $\theta_L^U(P) \tbisim \theta_L^U(Q)$.
    \end{itemize}
    We are going to show that $\tbisim$ is a \tb time-out bisimulation up to $\bisim$\,. This implies that ${\B}={\bisimtbrc}$\,, using Proposition~\ref{prop:up to}, and as $\B$ is a congruence for the operators of  Proposition \ref{prop:stability}, so is $\bisimtbrc$\,. Before we do so, we show, by induction on the construction of $\B$, that
    \begin{equation}\label{stability}
    \mbox{if $P \B Q$ and $P\nsteptau$ then $Q\pathtau Q'$ for some $Q'$ with $P \B Q'$ and $\init{Q'}=\init{P}$.}
    \end{equation}
    Let $P \B Q$ and $P\nsteptau$.
    \begin{itemize}
      \item If $P \bisimtbrc Q$ then, by Clause 3 of Definition~\ref{def:time-out bisim}, $Q\pathtau Q'$ for some $Q'$ with $Q'\nsteptau$. By (the symmetric counterpart of) Clause 1, one obtains $P \B Q'$. Clause 1 gives $\init{Q'}=\init{P}$.
      \item If $P = \alpha.P^\dag$ and $Q = \alpha.Q^\dag$ with $\alpha \in Act$ then note that $\alpha \ne \tau$ and take $Q':=Q$. One has $\init{Q'}=\init{P}$.
      \item If $P = P_1 \parallel_S P_1$ and $Q = Q_1 \parallel_S Q_2$ with $S \subseteq A$ and $P_i \tbisim Q_i$ for $i=1,2$, then, for $i=1,2$, $P_i\nsteptau$, so by induction $Q_i\pathtau Q'_i$ for some $Q'_i$ with $P_i \B Q'_i$ and $\init{Q'_i}=\init{P_i}$.
      Now $Q \pathtau Q'_1\|_S Q'_2$, $P \B  Q'_1\|_S Q'_2$ and $\init{Q'_1\|_S Q'_2}=\init{P}$.
      \item If $P = \tau_I(P_1)$ and $Q = \tau_I(Q_1)$ with $I \subseteq A$ and $P_1 \tbisim Q_1$, then $P_1\nsteptau$, so by induction $Q_1\pathtau Q'_1$ for some $Q'_1$ with $P_1 \B Q'_1$ and $\init{Q'_1}=\init{P_1}$. Now $Q \pathtau \tau_I(Q'_1)$, $P \B \tau_I(Q'_1)$ and $\init{\tau_I(Q'_1)}=\init{P}$.
      \item If $P = \rename(P_1)$ and $Q = \rename(Q_1)$ with $\rename \subseteq A\times A$ and $P_1 \tbisim Q_1$, then $P_1\nsteptau$, so by induction $Q_1\pathtau Q'_1$ for some $Q'_1$ with $P_1 \B Q'_1$ and $\init{Q'_1}=\init{P_1}$. Now $Q \pathtau \rename(Q'_1)$, $P \B \rename(Q'_1)$ and $\init{\rename(Q'_1)}=\init{P}$.
      \item  If $P = \theta_L^U(P_1)$ and $Q = \theta_L^U(Q_1)$ with  $L \subseteq U \subseteq A$ and $P_1 \tbisim Q_1$, then $P_1\nsteptau$, so by induction $Q_1\pathtau Q'_1$ for some $Q'_1$ with $P_1 \B Q'_1$ and $\init{Q'_1}=\init{P_1}$. Now $Q \pathtau \theta_L^U(Q'_1)$, $P \B \theta_L^U(Q'_1)$ and $\init{\theta_L^U(Q'_1)}=\init{P}$.
    \end{itemize}
    We now check that $\tbisim$ is a \tb time-out bisimulation up to $\bisim$\,. Note that $\tbisim$ is symmetric because $\bisimtbrc$ is. $\bisim$ was proven to be a congruence for $\ccsp$ in \cite{strongreactivebisimilarity}. Let $P,Q \in \closed$ such that $P \tbisim Q$.
    \begin{enumerate}
        \item If $P \step{\alpha} P'$ with $\alpha \in A_\tau$ then we have to find a path $Q \pathtau Q_1 \step{\opt{\alpha}} Q_2$ such that $P \upto[\bisim] Q_1$ and $P' \upto[\bisim] Q_2$. Remember that ${\tbisim} \subseteq {\upto[\bisim]}$. We are going to proceed by structural induction on $P$ and by case distinction on the derivation of $P \tbisim Q$.
        \begin{itemize}
            \item If $P \bisimtbrc Q$ then, by definition of $\bisimtbrc$\,, there exists a path $Q \pathtau Q_1 \step{\opt{\alpha}} Q_2$ such that $P \bisimtbrc Q_1$ and $P' \bisimtbrc Q_2$, thus, by definition of $\tbisim$, $P \tbisim Q_1$ and $P' \tbisim Q_2$.
            \item If $P = \beta.P^\dag$ and $Q = \beta.Q^\dag$ with $\beta \in Act$ and $P^\dag \tbisim Q^\dag$ then, by the semantics, $P' = P^\dag$, $\beta = \alpha$, and thus there exists a path $Q \step{\alpha} Q^\dag$ such that $P \tbisim Q$ and $P' \tbisim Q^\dag$.
            \item If $P = P^\dag \parallel_S P^\ddag$ and $Q = Q^\dag \parallel_S Q^\ddag$ with $S \subseteq A$, $P^\dag \tbisim Q^\dag$ and $P^\ddag \tbisim Q^\ddag$ then
            \begin{itemize}
                \item if $\alpha \in S$ then, by the semantics, $P' = P'^\dag \parallel_S P'^\ddag$, $P^\dag \step{\alpha} P'^\dag$ and $P^\ddag \step{\alpha} P'^\ddag$. Note that $\alpha \ne \tau$ because $\alpha \in A$. Since $P^\dag \tbisim Q^\dag$ and $P^\ddag \tbisim Q^\ddag$, by induction, there exist two paths $Q^\dag \pathtau Q_1^\dag \step{\alpha} Q_2^\dag$ and $Q^\ddag \pathtau Q_1^\ddag \step{\alpha} Q_2^\ddag$ such that $P^\dag \upto[\bisim] Q^\dag_1$,\linebreak[4] $P'^\dag \upto[\bisim] Q^\dag_2$, $P^\ddag \upto[\bisim] Q_1^\ddag$ and $P'^\ddag \upto[\bisim] Q^\ddag_2$. By the semantics, $Q \pathtau Q^\dag_1 \parallel_S Q^\ddag_1\linebreak[3] \step{\alpha} Q^\dag_2 \parallel_S Q^\ddag_2$. Moreover, by definition of $\tbisim$ and the congruence property of $\bisim$, $P \upto[\bisim] Q^\dag_1 \parallel_S Q^\ddag_1$ and $P' \upto[\bisim] Q^\dag_2 \parallel_S Q^\ddag_2$.
                \item if $\alpha \not\in S$ then, by the semantics, two cases are possible. Suppose that $P' = P'^\dag \parallel_S P^\ddag$ and $P^\dag \step{\alpha} P'^\dag$; the other case is symmetrical. Since $P^\dag \tbisim Q^\dag$, by induction, there exists a path $Q^\dag \pathtau Q^\dag_1 \step{\opt{\alpha}} Q^\dag_2$ such that $P^\dag \upto[\bisim] Q^\dag_1$ and $P'^\dag \upto[\bisim] Q^\dag_2$. By the semantics, there exists a path $Q \pathtau Q^\dag_1 \parallel_S Q^\ddag \step{\opt{\alpha}} Q^\dag_2 \parallel_S Q^\ddag$. Moreover, by definition of $\tbisim$ and the congruence property of $\bisim$, $P \upto[\bisim] Q^\dag_1 \parallel_S Q^\ddag$ and $P' \upto[\bisim] Q^\dag_2 \parallel_S Q^\ddag$.
            \end{itemize}
            \item If $P = \tau_I(P^\dag)$ and $Q = \tau_I(Q^\dag)$ with $I \subseteq A$ and $P^\dag \tbisim Q^\dag$ then, by the semantics, $P' = \tau_I(P'^\dag)$, $P^\dag \step{\beta} P'^\dag$ and $(\beta \in I \wedge \alpha = \tau) \vee \beta = \alpha$. Since $P^\dag \tbisim Q^\dag$, by induction, there exists a path $Q^\dag \pathtau Q^\dag_1 \step{\opt{\beta}} Q^\dag_2$ such that $P^\dag \upto[\bisim] Q^\dag_1$ and $P'^\dag \upto[\bisim] Q^\dag_2$. By the semantics, $Q \pathtau \tau_I(Q_1^\dag) \step{\opt{\alpha}} \tau_I(Q^\dag_2)$ such that, by definition of $\tbisim$ and the congruence property of $\bisim$, $P \upto[\bisim] \tau_I(Q_1^\dag)$ and $P' \upto[\bisim] \tau_I(Q_2^\dag)$.
            \item If $P = \rename(P^\dag)$ and $Q = \rename(Q^\dag)$ with $\rename \subseteq A\times A$ and $P^\dag \tbisim Q^\dag$ then, by the semantics, $P' = \rename(P'^\dag)$, $P^\dag \step{\beta} P'^\dag$ and $(\beta,\alpha) \in \rename \vee \alpha = \beta = \tau$. Since $P^\dag \tbisim Q^\dag$, by induction, there exists a path $Q^\dag \pathtau Q^\dag_1 \step{\opt{\beta}} Q^\dag_2$ such that $P^\dag \upto[\bisim] Q^\dag_1$ and $P'^\dag \upto[\bisim] Q^\dag_2$. By the semantics, $Q \pathtau \rename(Q_1^\dag) \step{\opt{\alpha}} \rename(Q^\dag_2)$ such that, by definition of $\tbisim$ and the congruence property of $\bisim$, $P \upto[\bisim] \rename(Q_1^\dag)$ and $P' \upto[\bisim] \rename(Q_2^\dag)$.
            \item If $P = \theta_L^U(P^\dag)$ and $Q = \theta_L^U(Q^\dag)$ with $L \subseteq U \subseteq A$ and $P^\dag \tbisim Q^\dag$ then
            \begin{itemize}
                \item if $\alpha = \tau$ then, by the semantics, $P' = \theta_X(P'^\dag)$ and $P^\dag \steptau P'^\dag$. Since $P^\dag \tbisim Q^\dag$, by induction, there exists a path $Q^\dag \pathtau Q^\dag_1 \step{\opt{\tau}} Q^\dag_2$ such that $P^\dag \upto[\bisim] Q^\dag_1$ and $P'^\dag \upto[\bisim] Q^\dag_2$. By the semantics, there exists a path $Q \pathtau \theta_L^U(Q^\dag_1) \step{\opt{\tau}} \theta_L^U(Q^\dag_2)$ such that, by definition of $\tbisim$ and the congruence property of $\bisim$, $P \upto[\bisim] \theta_L^U(Q^\dag_1)$ and $P' \upto[\bisim] \theta_L^U(Q^\dag_2)$.
                \item if $\alpha = a \in A$ then, by the semantics, $a \in U \vee \deadend{P^\dag}{L}$ and $P^\dag \step{a} P'$. Since $P^\dag \tbisim Q^\dag$, by induction there exists a path $Q^\dag \pathtau Q_1^\dag \step{a} Q_2^\dag$ such that $P^\dag \upto[\bisim] Q^\dag_1$ and $P' \upto[\bisim] Q^\dag_2$. Moreover, in case $a \notin U$ we have $P^\dag\nsteptau$ so (\ref{stability}) ensures that $Q^\dag\pathtau Q'$ for some $Q'$ with $P^\dag \B Q'$ and $\init{Q'}=\init{P}$. This implies that we may choose $Q_1^\dag$ such that $Q' \pathtau Q_1^\dag$, and thus $Q'=Q_1^\dag$. This gives us $\deadend{P^\dag}{L} \Leftrightarrow \deadend{Q^\dag_1}{L}$. By the semantics, there exists a path $Q \pathtau \theta_L^U(Q^\dag_1) \step{a} Q^\dag_2$ such that, by definition of $\tbisim$ and the congruence property of $\bisim$, $P \upto[\bisim] \theta_L^U(Q^\dag_1)$ and $P' \upto[\bisim] Q^\dag_2$.
            \end{itemize}
        \end{itemize}
        \item If $\deadend{P}{X}$ and $P \step{\rt} P'$ then we have to find a path $Q \pathtau
        Q_1 \step{\rt} Q_2$
        with
        $\theta_X(P') \upto[\bisim] \theta_X(Q_2)$. Remember that ${\tbisim} \subseteq
        {\upto[\bisim]}$. We are going to proceed by structural induction on $P$ and by case distinction on the derivation of $P \tbisim Q$.
        \begin{itemize}
            \item If $P \bisimtbrc Q$ then, by Definition~\ref{def:time-out bisim}, there exists a path $Q \pathtau Q_1 \step{\rt} Q_2$ with $\theta_X(P') \tbisim \theta_X(Q_2)$.
            \item If $P = \beta.P^\dag$ and $Q = \beta.Q^\dag$ with $\beta \in Act$ and
            $P^\dag \tbisim Q^\dag$ then, by the semantics, $P' = P^\dag$ and $\beta = \rt$. Thus, by the semantics, there exists a path $Q \step{\rt} Q^\dag$ such that, by definition of $\tbisim$, $\theta_X(P') \tbisim \theta_X(Q^\dag)$.
            \item If $P = P^\dag \parallel_S P^\ddag$ and $Q = Q^\dag \parallel_S Q^\ddag$ with $S \subseteq A$, $P^\dag \tbisim Q^\dag$ and $Q^\dag \tbisim Q^\ddag$ then, since $\rt \not\in S$, by the semantics, two cases are possible. Suppose that $P' = P^\dag \parallel_S P'^\ddag$ and $P^\ddag \step{\rt} P'^\ddag$; the other case is symmetrical. Since $\deadend{P}{X}$, $P^\dag \nsteptau$ and $\init{P^\dag} \cap X \subseteq S$. Moreover, $\deadend{P^\ddag}{(X\setminus S) \cup (X \cap S \cap \init{P^\dag})}$. Note that $(X\setminus S) \cup (X \cap S \cap \init{P^\dag}) = X \setminus (S \setminus \init{P^\dag})$. Since $P^\ddag \tbisim Q^\ddag$, $P^\ddag \step{\rt} P'^\ddag$ and $\deadend{P^\ddag}{X \setminus (S\setminus\init{P^\dag})}$, by induction, there exists a path $Q^\ddag \pathtau Q^\ddag_1 \step{\rt} Q^\ddag_2$ with $\theta_{X \setminus (S\setminus\init{P^\dag})}(P'^\ddag) \upto[\bisim] \theta_{X \setminus (S\setminus\init{P^\dag})}(Q^\ddag_{2r})$. Moreover, by (\ref{stability}), since $P^\dag \nsteptau$, there exists a path $Q^\dag \pathtau Q^\dag_0 \nsteptau$ such that $P^\dag \tbisim Q^\dag_0$ and $\init{Q^\dag_0} = \init{P^\dag}$. By the semantics, there exists a path $Q \pathtau Q^\dag_0 \parallel_S Q^\ddag_1 \step{\rt} Q^\dag_0 \parallel_S Q^\ddag_2$. By Lemma~\ref{lem:strong identities}, the definition of $\tbisim$ and the congruence property of $\bisim$, $\theta_X(P') \bisim \mbox{}$ $$\theta_X(P^\dag \parallel_S \theta_{X \setminus (S\setminus\init{P^\dag})}(P'^\ddag)) \upto[\bisim] \theta_X(Q^\dag_0 \parallel_S \theta_{X \setminus (S\setminus\init{Q_0^\dag})}(Q^\ddag_2)) \bisim \theta_X(Q^\dag_0 \parallel_S Q^\ddag_2).$$
            \item If $P = \tau_I(P^\dag)$ and $Q = \tau_I(Q^\dag)$ with $I \subseteq A$ and $P^\dag \tbisim Q^\dag$ then, by the semantics, $P' = \tau_I(P'^\dag)$, $P^\dag \step{\rt} P'^\dag$ and $\deadend{P^\dag}{(X\cup I)}$. Since $P^\dag \tbisim Q^\dag$, there exists a path $Q^\dag \pathtau Q^\dag_1 \step{\rt} Q^\dag_2$ with $\theta_{X\cup I}(P'^\dag) \upto[\bisim] \theta_{X\cup I}(Q^\dag_2)$. Lemma~\ref{lem:strong identities}, the definition of $\tbisim$ and the congruence property of $\bisim$, $$\theta_X(P') \bisim \theta_X(\tau_I(\theta_{X\cup I}(P'^\dag))) \upto[\bisim] \theta_X(\tau_I(\theta_{X\cup I}(Q^\dag_2))) \bisim \theta_X(\tau_I(Q_2^\dag)).$$
            \item If $P = \rename(P^\dag)$ and $Q = \rename(Q^\dag)$ with $\rename \subseteq A\times A$ and $P^\dag \tbisim Q^\dag$ then, by the $\deadend{P^\dag}{\rename^{-1}(X)}$. Since $P^\dag \tbisim Q^\dag$, by induction, there exists a path $Q^\dag \pathtau Q^\dag_1 \step{\rt} Q^\dag_2$ with $\theta_{\rename^{-1}(X)}(P'^\dag) \upto[\bisim] \theta_{\rename^{-1}(X)}(Q^\dag_2)$. By the semantics, $Q \pathtau \rename(Q_1^\dag) \step{\rt} \rename(Q^\dag_2)$. By Lemma~\ref{lem:strong identities}, the definition of $\tbisim$ and the congruence property of $\bisim$, $$\theta_X(P') \bisim \theta_X(\rename(\theta_{\rename^{-1}(X)}(P'^\dag))) \upto[\bisim] \theta_X(\rename(\theta_{\rename^{-1}(X)}(Q^\dag_{2r}))) \bisim \theta_X(\rename(Q_{2r}^\dag)).$$
            \item If $P = \theta_L^U(P^\dag)$ and $Q = \theta_L^U(Q^\dag)$ with $L \subseteq U \subseteq A$ and $P^\dag \tbisim Q^\dag$ then, by the semantics, $\deadend{P^\dag}{L\cup X}$ and $P^\dag \step{\rt} P'$. Since $P^\dag \tbisim Q^\dag$ and $P^\ddag\nsteptau$, (\ref{stability}) ensures that $Q^\dag \pathtau Q^\dag_0$ for some $Q^\dag_0$ with $P^\dag \B Q^\dag_0$ and $\init{Q^\dag_0}=\init{P^\dag}$. Since $P^\dag \B Q^\dag_0$, by induction, there exists a path $Q^\dag_0 \pathtau Q_1^\dag \step{\rt} Q_2$ with $\theta_X(P') \upto[\bisim] \theta_X(Q_2)$. As $Q^\dag_0\nsteptau$ we have $Q^\dag_0=Q^\dag_1$. As $\init{Q^\dag_1}=\init{P^\dag}$ one has $\deadend{Q^\dag_1}{L}$. By the semantics, there exists a path $Q \pathtau \theta_L^U(Q^\dag_1) \step{\rt} Q_2$.
        \end{itemize}
        \item The last condition of Definition~\ref{def:up to} is implied by (\ref{stability}).
    \popQED
    \end{enumerate}
\end{proof}

\section{Full Congruence Proofs for \texorpdfstring{$\bisimrtbrc$ and $\bisimrb$}{Rooted Branching Reactive Bisimilarity}} \label{app:congruence}

\begin{definition}\rm \label{def:rooted time-out bisim up to}
    Here, a \emph{rooted \tb time-out bisimulation up to $\bisimtbrc$} is a symmetric relation ${\tbisim} \subseteq \closed\times\closed$ such that, for all $P,Q \in \closed$ with $P \tbisim Q$,
    \begin{enumerate}
        \item if $P \step{\alpha} P'$ with $\alpha \in A_\tau$ then there is a transition $Q \step{\alpha} Q'$ such that $P' \bisim\,\B\,\bisimtbrc Q'$
        \item if $\deadend{P}{X}$ and $P \step{\rt} P'$ then there is a transition $Q \step{\rt} Q'$ such that $\theta_X(P') \bisim\,\B\,\bisimtbrc \theta_X(Q')$.
    \end{enumerate}
\end{definition}

\begin{proposition}
    Let $P,Q \in \closed$. Then $P \bisimrtbrc Q$ iff there exists a rooted \tb time-out bisimulation $\B$ up to $\bisimtbrc$ such that $P \B Q$.
\end{proposition}

\begin{proof}
    First of all, a rooted \tb time-out bisimulation is a rooted \tb time-out bisimulation up to $\bisimtbrc$ by reflexivity of $\bisim$ and $\bisimtbrc$\,. Conversely, we are going to show that $\bisim\,\B\,\bisimtbrc$ is a \tb time-out bisimulation. This implies that ${\bisim\,\B\,\bisimtbrc} \subseteq {\bisimtbrc}$\,, so that each rooted \tb time-out bisimulation up to $\bisimtbrc$ is in fact a rooted \tb time-out bisimulation. Let $P,Q \in \closed$ such that $P \bisim\,\B\,\bisimtbrc Q$. There exists $P^\dag,Q^\dag \in \closed$ such that $P \bisim P^\dag \B Q^\dag \bisimtbrc Q$.
    \begin{enumerate}
        \item If $P \step{\alpha} P'$ then, since $P \bisim P^\dag$, there is a transition $P^\dag \step{\alpha} P^\ddag$ such that $P' \bisim P^\ddag$. Thus, by Clause 1 of Definition~\ref{def:rooted time-out bisim up to}, there is a transition $Q^\dag \step{\alpha} Q^\ddag$ such that $P^\ddag \bisim\,\B\,\bisimtbrc Q^\ddag$. By Clause~1 of Definition~\ref{def:time-out bisim} there is a path $Q \pathtau Q_1 \step{\opt{\alpha}} Q_2$ with $Q^\dag \bisimtbrc Q_1$ and $Q^\ddag \bisimtbrc Q_2$. By the transitivity of $\bisim$ and $\bisimtbrc$ we obtain $P \bisim\,\B\,\bisimtbrc Q_1$ and $P' \bisim\,\B\,\bisimtbrc Q_2$.
        \item If $\deadend{P}{X}$ and $P \step{\rt} P'$ then, since $P \bisim P^\dag$, $\deadend{P^\dag}{X}$ and $P^\dag \step{\rt} P^\ddag$ for some $P^\ddag$ with $P' \bisim P^\ddag$. Thus, by Clauses 1 and 2 of Definition~\ref{def:rooted time-out bisim up to}\linebreak[4] $\deadend{Q^\dag}{X}$ and there is a transition $Q^\dag \step{\rt} Q^\ddag$ with $\theta_X(P^\ddag) \bisim\,\B\,\bisimtbrc \theta_X(Q^\ddag)$.  By Clause~2 of Definition~\ref{def:time-out bisim} there is a path $Q \pathtau Q_1 \step{\rt} Q_2$ with $\theta_X(Q^\ddag) \bisimtbrc \theta_X(Q_2)$. Since $\bisim$ is a congruence for $\theta_X$ \cite[Theorem 20]{strongreactivebisimilarity}, we have $\theta_X(P')\bisim\theta_X(P^\ddag)$. By the transitivity of $\bisim$ and $\bisimtbrc$ we obtain $\theta_X(P') \bisim\,\B\,\bisimtbrc \theta_X(Q_2)$.
        \item If $P\nsteptau$ then, since $P \bisim P^\dag$, $P^\dag\nsteptau$, so by Clause 1 of Definition~\ref{def:rooted time-out bisim up to} $Q^\dag\nsteptau$, and by Clause~3 of Definition~\ref{def:time-out bisim} there is a path $Q \pathtau Q_0 \nsteptau$.
        \popQED
    \end{enumerate}
\end{proof}

\begin{proof}[Proof of Theorem \ref{thm:congruence}]
    Let ${\tbisim} \subseteq \closed\times\closed$ be the smallest relation such that 
    \begin{itemize}
        \item if $P \bisimrtbrc Q$ then $P \tbisim Q$
        \item if $P \tbisim Q$ and $\alpha \in Act$ then $\alpha.P \tbisim \alpha.Q$
        \item if $P_1 \tbisim Q_1$ and $P_2 \tbisim Q_2$ then $P_1 + P_2 \tbisim Q_1 + Q_2$
        \item if $P_1 \tbisim Q_1$, $P_2 \tbisim Q_2$ and $S \subseteq A$ then $P_1 \parallel_S P_2 \tbisim Q_1 \parallel_S Q_2$
        \item if $P \tbisim Q$ and $I \subseteq A$ then $\tau_I(P) \tbisim \tau_I(Q)$
        \item if $P \tbisim Q$ and $\rename \subseteq A\times A$ then $\rename(P) \tbisim \rename(Q)$
        \item if $P \tbisim Q$ and $L \subseteq U \subseteq A$ then $\theta_L^U(P) \tbisim \theta_L^U(Q)$
        \item if $P \tbisim Q$ and $X \subseteq A$ then $\psi_X(P) \tbisim \psi_X(Q)$
        \item if $\equa$ is a recursive specification with $z \in V_\equa$ and $\rho, \nu \in V\setminus V_\equa \rightarrow \closed$ are substitutions such that $\forall x \in V\setminus V_\equa,\; \rho(x) \tbisim \nu(x)$, then $\langle z|\equa\rangle[\rho] \tbisim \langle z|\equa\rangle[\nu]$.
        \item if $\equa$ and $\equa'$ are recursive specifications and $x \in V_\equa=V_{\equa'}$ with $\langle x|\equa\rangle, \langle x|\equa'\rangle \in \closed$ such that $\forall y \in V_\equa,\; \equa_y \bisimrtbrc \equa'_y$, then $\langle x|\equa\rangle \tbisim \langle x|\equa'\rangle$.
    \end{itemize}
    Note that since $\bisim$, $\tbisim$ and $\bisimtbrc$ are congruences for the operators listed in Proposition~\ref{prop:stability}, so are the composed relations $\tbisim\,\bisimtbrc$ and $\bisim\,\tbisim\,\bisimtbrc$\;.\hfill (\textsterling)

   Let ${=_\mathcal{I}} := \{(P,Q) \mid \init{P} = \init{Q}\}$. A trivial induction on the derivation of $P \tbisim Q$, using the fact that $=_\mathcal{I}$ is a full congruence for $\ccsp$ \cite{strongreactivebisimilarity}, shows that 
    \begin{align*}
        P \tbisim Q \Rightarrow \init{P} = \init{Q} \tag{$@$}
    \end{align*}
    (For the second last case, the assumption that $\rho(x) \B \nu(x)$ for all $x\in V\setminus V_\equa$ implies $\rho =_{\mathcal{I}} \nu$ by induction. Since $=_{\mathcal{I}}$ is a lean congruence, this implies $\langle z|\equa\rangle[\rho] =_{\mathcal{I}}\langle{z|\equa}\rangle[\nu]$.)\\
    A trivial induction on $\expr$ shows that
    \begin{align*}
        \forall E \in \expr, \rho,\nu \in V \rightarrow \closed,\; (\forall x \mathbin\in V, \rho(x) \tbisim \nu(x)) \Rightarrow E[\rho] \tbisim E[\nu] \tag{$\star$}
    \end{align*}
    A useful corollary is
    \begin{equation}
        \begin{aligned}
            \forall E\in\expr, \equa \mbox{ a recursive specification}, \rho, \nu \in V\setminus V_\equa \rightarrow \closed,\\ (\forall x \in V\setminus V_\equa,\; \rho(x) \tbisim \nu(x)) \Rightarrow \langle E|\equa\rangle[\rho] \tbisim \langle E|\equa\rangle[\nu]
        \end{aligned} \tag{$\$$}
    \end{equation}
    Applied in the context of the last condition of $\tbisim$, it implies
    \begin{align}
        \forall E \in \expr, \mbox{the variables of \textit{E} are in }V_\equa \Rightarrow \langle E|\equa\rangle \tbisim \langle E|\equa'\rangle \tag{$\#$}
    \end{align}
    Since ${\bisimrtbrc} \subseteq {\tbisim}$, it suffices to prove that $\tbisim$ is a rooted \tb time-out bisimulation up to $\bisimtbrc$ (so that ${\tbisim}={\bisimrtbrc}$). Note that $\tbisim$ is symmetric, since $\bisimrtbrc$ is. Let $P, Q \in \closed$ such that $P \tbisim Q$.
    \begin{enumerate}
        \item If $P \step{\alpha} P'$ with $\alpha \in A_\tau$ then we need to find a transition $Q \step{\alpha} Q'$ such that $P' \tbisim\,\bisimtbrc Q'$. This is sufficient as ${\tbisim\bisimtbrc} \subseteq {\bisim\tbisim\bisimtbrc}$. We are going to proceed by induction on the proof of $P \step{\alpha} P'$ and by case distinction on the derivation of $P \tbisim Q$.
        \begin{itemize}
            \item If $P \bisimrtbrc Q$ then there exists a transition $Q \step{\alpha} Q'$ such that $P' \bisimtbrc Q'$, and so $P' \tbisim\bisimtbrc Q'$.
            \item If $P = \beta.P^\dag$ and $Q = \beta.Q^\dag$ such that $\beta \in Act$ and $P^\dag \tbisim Q^\dag$ then $\alpha = \beta$ and $P' = P^\dag$. Thus there exists a transition $Q \step{\alpha} Q^\dag$ such that $P^\dag \tbisim Q^\dag$, and so $P^\dag \tbisim \bisimtbrc Q^\dag$.
            \item If $P = P^\dag + P^\ddag$ and $Q = Q^\dag + Q^\ddag$ such that $P^\dag \tbisim Q^\dag$ and $P^\ddag \tbisim Q^\ddag$ then, by the semantics, $P^\dag \step{\alpha} P'$ or $P^\ddag \step{\alpha} P'$. Suppose that $P^\dag \step{\alpha} P'$ (the other case proceeds symmetrically). Since $P^\dag \tbisim Q^\dag$, there exists a transition $Q^\dag \step{\alpha} Q'$ such that $P' \tbisim\,\bisimtbrc Q'$. By the semantics, there exists a transition $Q \step{\alpha} Q'$ such that $P' \tbisim\,\bisimtbrc Q'$.
            \item If $P = P^\dag \parallel_S P^\ddag$ and $Q = Q^\dag \parallel_S Q^\ddag$ such that $S\subseteq A$, $P^\dag \tbisim Q^\dag$ and $P^\ddag \tbisim Q^\ddag$ then
            \begin{itemize}
                \item if $\alpha \not\in S$ then, by the semantics, $P' = P'^\dag \parallel_S P^\ddag$ and $P^\dag \step{\alpha} P'^\dag$ or $P' = P^\dag \parallel_S P'^\ddag$ and $P^\ddag \step{\alpha} P'^\ddag$. Suppose that $P^\dag \step{\alpha} P'^\dag$ (the other case proceeds symmetrically). Since $P^\dag \tbisim Q^\dag$, there exists a transition $Q^\dag \step{\alpha} Q'^\dag$ such that $P'^\dag \tbisim\,\bisimtbrc Q'^\dag$. By the semantics, there exists a transition $Q \step{\alpha} Q'^\dag \parallel_S Q^\ddag$ such that, by (\textsterling), $P' \tbisim\,\bisimtbrc Q'^\dag \parallel_S Q^\ddag$.
                \item if $\alpha \in S$ then, by the semantics, $P' = P'^\dag \parallel_S P'^\ddag$, $P^\dag \step{\alpha} P'^\dag$ and $P^\ddag \step{\alpha} P'^\ddag$. Since $P^\dag \tbisim Q^\dag$ and $P^\ddag \tbisim Q^\ddag$, there exists two transitions $Q^\dag \step{\alpha} Q'^\dag$ and $Q^\ddag \step{\alpha} Q'^\ddag$ such that $P'^\dag \tbisim\,\bisimtbrc Q'^\dag$ and $P'^\ddag \tbisim\,\bisimtbrc Q'^\ddag$. By the semantics, there exists a transition $Q \step{\alpha} Q'^\dag \parallel_S Q'^\ddag$ such that, by (\textsterling), $P' \tbisim\,\bisimtbrc Q'^\dag \parallel_S Q'^\ddag$.
            \end{itemize}
            \item If $P = \tau_I(P^\dag)$ and $Q = \tau_I(Q^\dag)$ with $I \subseteq A$ and $P^\dag \tbisim Q^\dag$ then, by the semantics, $P' = \tau_I(P'^\dag)$, $P^\dag \step{\beta} P'^\dag$ and $(\beta \in I\cup\{\tau\} \wedge \alpha = \tau) \vee \beta = \alpha \not\in I$. Since $P^\dag \tbisim Q^\dag$, there exists a transition $Q^\dag \step{\beta} Q'^\dag$ such that $P'^\dag \tbisim\,\bisimtbrc Q'^\dag$. By the semantics, there exists a transition $Q \step{\alpha} \tau_I(Q'^\dag)$ such that, by (\textsterling), $P'\tbisim\,\bisimtbrc \tau_I(Q'^\dag)$.
            \item If $P = \rename(P^\dag)$ and $Q = \rename(Q^\dag)$ with $\rename \subseteq A\times A$ and $P^\dag \tbisim Q^\dag$ then, by the semantics, $P' = \rename(P'^\dag)$, $P^\dag \step{\beta} P'^\dag$ and $(\beta,\alpha) \in \rename \vee \beta = \alpha = \tau$. Since $P^\dag \tbisim Q^\dag$, there exists a transition $Q^\dag \step{\beta} Q'^\dag$ such that $P'^\dag \tbisim\,\bisimtbrc Q'^\dag$. By the semantics, there exists a transition $Q \step{\alpha} \rename(Q'^\dag)$ such that, by (\textsterling), $P'\tbisim\,\bisimtbrc \rename(Q'^\dag)$.
            \item If $P = \theta_L^U(P^\dag)$ and $Q = \theta_L^U(Q^\dag)$ with $L \subseteq U \subseteq A$ and $P^\dag \tbisim Q^\dag$ then
            \begin{itemize}
                \item if $\alpha = \tau$ then, by the semantics, $P' = \theta_X(P'^\dag)$ and $P^\dag \step{\tau} P'^\dag$. Since $P^\dag \tbisim Q^\dag$, there exists a transition $Q^\dag \step{\tau} Q'^\dag$ such that $P'^\dag \tbisim\,\bisimtbrc Q'^\dag$. By the semantics, there exists a transition $Q \step{\tau} \theta_L^U(Q'^\dag)$ such that, by (\textsterling), $P'\tbisim\,\bisimtbrc \theta_L^U(Q'^\dag)$.
                \item if $\alpha = a \in A$ then, by the semantics, $P^\dag \step{a} P'$ and $a \in U \vee \deadend{P^\dag}{L}$. Since $P^\dag \tbisim Q^\dag$, there exists a transition $Q^\dag \step{a} Q'$ such that $P' \tbisim\,\bisimtbrc Q'$. According to $(@)$, $\deadend{P^\dag}{L} \Leftrightarrow \deadend{Q^\dag}{L} $, thus, by the semantics, there exists a transition $Q \step{a} Q'$ such that $P'\tbisim\,\bisimtbrc Q'$.
            \end{itemize}
            \item If $P = \psi_X(P^\dag)$ and $Q = \psi_X(Q^\dag)$ with $X \subseteq A$ and $P^\dag \tbisim Q^\dag$ then, by the semantics, $P^\dag \step{\alpha} P'$. Since $P^\dag \tbisim Q^\dag$, there exists a transition $Q^\dag \step{\alpha} Q'$ such that $P' \tbisim\,\bisimtbrc Q'$. By the semantics, there exists a transition $Q \step{\alpha} Q'$ such that $P'\tbisim\,\bisimtbrc Q'$.
            \item Let $P = \langle z|\equa\rangle[\rho]$ and $Q = \langle z|\equa\rangle[\nu]$ with $\equa$ a recursive specification, $z \in V_\equa$ and $\rho, \nu \in V\setminus V_\equa \rightarrow\closed$ such that $\forall x \in V\setminus V_\equa,\; \rho(x) \tbisim \nu(x)$. By the semantics, $\langle\equa_z|\equa\rangle[\rho] \step{\alpha} P'$ is provable by a strict sub-proof of $P \step{\alpha} P'$. Moreover, according to $(\$)$, $\langle\equa_z|\equa\rangle[\rho] \tbisim \langle\equa_z|\equa\rangle[\nu]$. By induction, there exists a transition $\langle\equa_z|\equa\rangle[\nu] \step{\alpha} Q'$ such that $P'\tbisim\,\bisimtbrc Q'\!$. By the semantics, there exists a transition $\langle z|\equa\rangle[\nu] \step{\alpha} Q'$ such that $P'\tbisim\,\bisimtbrc Q'$.
            \item Let $P = \langle x|\equa\rangle$ and $Q = \langle x|\equa'\rangle$ with $\equa$ and $\equa'$ two recursive specifications such that $\forall y \in V_\equa = V_{\equa'},\; \equa_y \bisimrtbrc \equa'_y$ and $x \in V_\equa$. By the semantics, $\langle\equa_x|\equa\rangle \step{\alpha} P'$ is provable by a strict sub-proof of $P \step{\alpha} P'$. Moreover, according $(\#)$, $\langle \equa_x|\equa\rangle \tbisim \langle\equa_x|\equa'\rangle$. By induction, there exists a transition $\langle\equa_x|\equa'\rangle \step{\alpha} R'$ such that $P' \tbisim\,\bisimtbrc R'$. Since $\langle \_ |\equa'\rangle \in V_{\equa'} \rightarrow \closed$ and $\equa_x \bisimrtbrc \equa'_x$, $\langle\equa_x|\equa'\rangle \bisimrtbrc \langle\equa_x'|\equa'\rangle$. Therefore, there exists a transition $\langle\equa'_x|\equa'\rangle \step{\alpha} Q'$ such that $R' \bisimtbrc Q'$. By the semantics, there exists a transition $Q \step{\alpha} Q'$ such that, by transitivity of $\bisimtbrc$\,, $P' \tbisim\,\bisimtbrc Q'$.
        \end{itemize}
        \item If $\deadend{P}{X}$ and $P \step{\rt} P'$ then we need to find a transition $Q \step{\rt} Q'$ such that $\theta_X(P') \bisim\,\tbisim\,\bisimtbrc \theta_X(Q')$. We are going to proceed by induction on the proof of $P \step{\rt} P'$ and by case distinction on the derivation of $P \tbisim Q$.
        \begin{itemize}
            \item If $P \bisimrtbrc Q$ then there exists a transition $Q \step{\rt} Q'$ such that $\theta_X(P') \bisimtbrc \theta_X(Q')$ and so $\theta_X(P') \bisim\,\tbisim\,\bisimtbrc \theta_X(Q')$.
            \item If $P = \beta.P^\dag$ and $Q = \beta.Q^\dag$ such that $\beta \in Act$ and $P^\dag \tbisim Q^\dag$ then $\beta = \rt$ and $P' = P^\dag$. Thus there is a transition $Q \step{t} Q^\dag$ such that, by definition of $\tbisim$, $\theta_X(P^\dag) \bisim\,\tbisim\,\bisimtbrc \theta_X(Q^\dag)$.
            \item If $P = P^\dag + P^\ddag$ and $Q = Q^\dag + Q^\ddag$ such that $P^\dag \tbisim Q^\dag$ and $P^\ddag \tbisim Q^\ddag$ then, by the semantics, $\deadend{P^\dag}{X}$, $\deadend{P^\ddag}{X}$ and $P^\dag \step{\rt} P'$ or $P^\ddag \step{\rt} P'$. Suppose that $P^\dag \step{\rt} P'$ (the other case is symmetrical). Since $P^\dag \tbisim Q^\dag$, there exists a transition $Q^\dag \step{\rt} Q'$ such that $\theta_X(P') \bisim\,\tbisim\,\bisimtbrc \theta_X(Q')$. By the semantics, there exists a transition $Q \step{\rt} Q'$ such that $\theta_X(P') \bisim\,\tbisim\, \bisimtbrc \theta_X(Q')$.
            \item If $P = P^\dag \parallel_S P^\ddag$ and $Q = Q^\dag \parallel_S Q^\ddag$ such that $S\subseteq A$, $P^\dag \tbisim Q^\dag$ and $P^\ddag \tbisim Q^\ddag$ then, by the semantics, $P' = P'^\dag \parallel_S P^\ddag$ and $P^\dag \step{\rt} P'^\dag$ or $P' = P^\dag \parallel_S P'^\ddag$ and $P^\ddag \step{\rt} P'^\ddag$. Suppose that $P^\ddag \step{\rt} P'^\ddag$ (the other case is symmetrical). Since $\deadend{P}{X}$, $P^\dag \nsteptau$ and $\init{P^\dag} \cap X \subseteq S$. Moreover, $\deadend{P^\ddag}{X \setminus S \cup (X\cap S\cap\init{P^\dag})}$. Note that $X \setminus S \cup (X\cap S\cap\init{P^\dag}) = X \setminus(S\setminus\init{P^\dag})$. Since $P^\ddag \tbisim Q^\ddag$, there exists a transition $Q^\ddag \step{\rt} Q'^\ddag$ such that $\theta_{X \setminus(S\setminus\init{P^\dag})}(P'^\ddag) \bisim\,\tbisim\,\bisimtbrc \theta_{X \setminus(S\setminus\init{P^\dag})}(Q'^\ddag)$. Since $P^\dag \tbisim Q^\dag$ and $P^\ddag \tbisim Q^\ddag$, $\init{P^\dag} = \init{Q^\dag}$ and $\init{P^\ddag} = \init{Q^\ddag}$. By the semantics, there exists a transition $Q \step{\rt} Q^\dag \parallel_S Q'^\ddag$. By (\textsterling) and Lemma~\ref{lem:strong identities}, \(\theta_X(P') = \theta_X(P^\dag \parallel_S P'^\ddag) \bisim \theta_X(P^\dag \parallel_S\theta_{X \setminus(S\setminus\init{P^\dag})}(P'^\ddag)) \bisim\,\tbisim\,\bisimtbrc \theta_X(Q^\dag \parallel_S \theta_{X \setminus(S\setminus\init{Q^\dag})}(Q'^\ddag)) \bisim \theta_X(Q^\dag \parallel_S Q'^\ddag)\). In the last step we use that $Q^\dag \nsteptau$, since $P^\dag\nsteptau$ and $\init{P}=\init{Q}$, using ($@$). Now apply that ${\bisim}\subseteq{\bisimtbrc}$ and the transitivity of $\bisim$ and $\bisimtbrc$\,.
            \item If $P = \tau_I(P^\dag)$ and $Q = \tau_I(Q^\dag)$ with $I \subseteq A$ and $P^\dag \tbisim Q^\dag$ then, by the semantics, $P' = \tau_I(P'^\dag)$ and $P^\dag \step{\rt} P'^\dag$. Since $\deadend{P}{X}$, $\deadend{P^\dag}{X\cup I}$. Since $P^\dag \tbisim Q^\dag$, there exists a transition $Q^\dag \step{\rt} Q'^\dag$ such that $\theta_{X\cup I}(P'^\dag) \bisim\,\tbisim\,\bisimtbrc \theta_{X\cup I}(Q'^\dag)$. By the semantics, there exists a transition $Q \step{\rt} \tau_I(Q'^\dag)$ such that, by (\textsterling) and Lemma~\ref{lem:strong identities}, $\theta_X(P') \bisim \theta_X(\tau_I(\theta_{X\cup I}(P'^\dag))) \bisim\,\tbisim\,\bisimtbrc \theta_X(\tau_I(\theta_{X\cup I}(Q'^\dag))) \bisim \tau_I(Q'^\dag)$.
            \item If $P = \rename(P^\dag)$ and $Q = \rename(Q^\dag)$ with $\rename \subseteq A\times A$ and $P^\dag \tbisim Q^\dag$ then, by the semantics, $P' = \rename(P'^\dag)$ and $P^\dag \step{\rt} P'^\dag$. Since $\deadend{P}{X}$, $\deadend{P^\dag}{\rename^{-1}(X)}$. Since $P^\dag \tbisim Q^\dag$, there exists a transition $Q^\dag \step{\rt} Q'^\dag$ such that $\theta_{\rename^{-1}(X)}(P'^\dag) \bisim\,\tbisim\,\bisimtbrc \theta_{\rename^{-1}(X)}(Q'^\dag)$. By the semantics, there exists a transition $Q \step{\rt} \rename(Q'^\dag)$ such that, by (\textsterling) and Lemma~\ref{lem:strong identities}, $\theta_X(P') \bisim \theta_X(\rename(\theta_{\rename^{-1}(X)}(P'^\dag))) \bisim\,\tbisim\,\bisimtbrc \theta_X(\rename(\theta_{\rename^{-1}(X)}(Q'^\dag))) \bisim \rename(Q'^\dag)$.
            \item If $P = \theta_L^U(P^\dag)$ and $Q = \theta_L^U(Q^\dag)$ with $L \subseteq U \subseteq A$ and $P^\dag \tbisim Q^\dag$ then, by the semantics, $P^\dag \step{\rt} P'$ and $\deadend{P^\dag}{L}$. Since $P^\dag \tbisim Q^\dag$, there exists a transition $Q^\dag \step{\rt} Q'$ such that $\theta_X(P') \bisim\,\tbisim\,\bisimtbrc \theta_X(Q')$. According to $(@)$, $\deadend{P^\dag}{L} \Leftrightarrow \deadend{Q^\dag}{L} $, thus, by the semantics, there exists a transition $Q \step{\rt} Q'$ such that $\theta_X(P')\bisim\,\tbisim\,\bisimtbrc \theta_X(Q')$.
            \item If $P = \psi_Y(P^\dag)$ and $Q = \psi_Y(Q^\dag)$ with $Y \subseteq A$ and $P^\dag \tbisim Q^\dag$ then, by the semantics, $P' = \theta_Y(P'^\dag)$, $P^\dag \step{\rt} P'^\dag$ and $\deadend{P^\dag}{Y}$. Since $P^\dag \tbisim Q^\dag$, there exists a transition $Q^\dag \step{\rt} Q'^\dag$ such that $\theta_Y(P'^\dag) \bisim\,\tbisim\,\bisimtbrc \theta_Y(Q'^\dag)$. Using ($@$), $\deadend{Q^\dag}{Y}$, so by the semantics, there exists a transition $Q \step{\rt} \theta_Y(Q'^\dag)$. By (\textsterling), $\theta_X(P') \bisim\,\tbisim\,\bisimtbrc \theta_X(\theta_Y(Q'^\dag))$.
            \item Let $P = \langle z|\equa\rangle[\rho]$ and $Q = \langle z|\equa\rangle[\nu]$ with $\equa$ a recursive specification, $z \in V_\equa$ and $\rho, \nu \in V\setminus V_\equa \rightarrow\closed$ such that $\forall x \in V\setminus V_\equa, \rho(x) \tbisim \nu(x)$. By the semantics, $\langle\equa_z|\equa\rangle[\rho] \step{\rt} P'$ is provable by a strict sub-proof of $P \step{\rt} P'$ and $\init{P} = \init{\langle\equa_z|\equa\rangle[\rho]}$. Moreover, according to $(\$)$, $\langle\equa_z|\equa\rangle[\rho] \tbisim \langle\equa_z|\equa\rangle[\nu]$. By induction, there exists a transition $\langle\equa_z|\equa\rangle[\nu] \step{\rt} Q'$ such that $\theta_X(P') \bisim\,\tbisim\,\bisimtbrc \theta_X(Q')$. By the semantics, there exists a transition $\langle z|\equa\rangle[\nu] \step{\rt} Q'$ such that $\theta_X(P') \bisim\,\tbisim\,\bisimtbrc \theta_X(Q')$.
            \item Let $P = \langle x|\equa\rangle$ and $Q = \langle x|\equa'\rangle$ with $\equa$ and $\equa'$ two recursive specifications such that $\forall y \in V_\equa = V_{\equa'}, \equa_y \bisimrtbrc \equa'_y$ and $x \in V_\equa$. By the semantics, $\langle\equa_x|\equa\rangle \step{\rt} P'$ is provable be a strict sub-proof of $P \step{\alpha} P'$ and $\init{P} = \init{\langle\equa_x|\equa\rangle}$. Moreover, according $(\#)$, $\langle \equa_x|\equa\rangle \tbisim \langle\equa_x|\equa'\rangle$. By induction, there exists a transition $\langle\equa_x|\equa'\rangle \step{\rt} R'$ such that $\theta_X(P') \bisim\,\tbisim\,\bisimtbrc \theta_X(R')$. Since $\langle \_ |\equa'\rangle \in V_{\equa'} \rightarrow \closed$ and $\equa_x \bisimrtbrc \equa'_x$, $\langle\equa_x|\equa'\rangle \bisimrtbrc \langle\equa_x'|\equa'\rangle$. Moreover, according to $(@)$, $\deadend{\langle\equa_x|\equa'\rangle}{X}$. Therefore, there exists a transition $\langle\equa'_x|\equa'\rangle \step{\alpha} Q'$ such that $\theta_X(R') \bisimtbrc \theta_X(Q')$. By the semantics, there exists a transition $Q \step{\alpha} Q'$ such that, by transitivity of $\bisimtbrc$, $\theta_X(P') \bisim\,\tbisim\,\bisimtbrc \theta_X(Q')$.
        \end{itemize}
    \end{enumerate}
    As a result, $\mathcal{B}$ is a rooted \tb time-out bisimulation up to $\bisimtbrc$\,, and $(\star)$ gives us that $\bisimrtbrc$ is a lean congruence and the last condition of $\tbisim$ adds that it is a full congruence.
\end{proof}

\section{Proof of RSP}\label{app:RSP}

To prove RSP, another version of $\bisimtbrc$ is needed, this time up to itself.

\begin{definition}\rm \label{def:up to b}
    A \textit{\tb time-out bisimulation up to $\bisimtbrc$} is a symmetric relation ${\tbisim} \subseteq \closed\times\closed$ such that, for all $P,Q \in \closed$ such that $P \tbisim Q$, and for all $X\subseteq A$,
    \begin{enumerate}
        \item if $P \pathtau P' \step{\alpha} P''$ with $\alpha \in A_\tau$ and $P \bisimtbrc P'$ then there exists a path $Q \pathtau Q_1 \step{\opt{\alpha}} Q_2$ such that $P' \upto[\bisimtbrc] Q_1$ and $P'' \upto[\bisimtbrc] Q_2$        \item if $P \pathtau P_1 \step{\rt} P_2$ with $P \bisimtbrc P_1$ and $\deadend{P_1}{X}$ then there exists a path $Q \pathtau Q_1 \step{\rt} Q_2$ with $\theta_X(P_2) \upto[\bisimtbrc] \theta_X(Q_2)$
        \item if $P \pathtau P_0 \nsteptau$ with $P \bisimtbrc P_0$ then there exists a path $Q \pathtau Q_0 \nsteptau$.
    \end{enumerate}
\end{definition}

\begin{proposition} \label{prop:up to b}
    Let $P,Q \in \closed$. Then $P \bisimtbrc Q$ iff there exists a \tb time-out bisimulation $\mathcal{B}$ up to $\bisimtbrc$ such that $P \tbisim Q$.
\end{proposition}

\begin{proof}
    Let $\tbisim$ be a \tb time-out bisimulation up to $\bisimtbrc$\,. We are going to show that $\upto[\bisimtbrc]$ is a \tb time-out bisimulation. Let $P,Q \in \closed$ such that $P \upto[\bisimtbrc] Q$. Then there exists $P^\dag,Q^\dag \in \closed$ such that $P \bisimtbrc P^\dag \tbisim Q^\dag \bisimtbrc Q$.
    \begin{enumerate}
        \item If $P \step{\alpha} P'$ with $\alpha \in A_\tau$ then, since $P \bisimtbrc P^\dag$, there exists a path $P^\dag \pathtau P^\star \step{\opt{\alpha}} P^\ddag$ such that $P \bisimtbrc P^\star$ and $P' \bisimtbrc P^\ddag$. Since $P^\dag \pathtau P^\star \step{\opt{\alpha}} P^\ddag$ and $P^\dag \tbisim Q^\dag$, there exists a path $Q^\dag \pathtau Q^\star \step{\opt{\alpha}} Q^\ddag$ such that $P^\star \upto[\bisimtbrc] Q^\star$ and $P^\ddag \upto[\bisimtbrc] Q^\ddag$. Since $Q^\dag \pathtau Q^\star$ and $Q^\dag \bisimtbrc Q$, there exists a path $Q \pathtau Q_0$ such that $Q^\star \bisimtbrc Q_0$; moreover, since $Q^\star \step{\opt{\alpha}} Q^\ddag$, there exists a path $Q_0 \pathtau Q_1 \step{\opt{\alpha}} Q_2$ such that $Q^\star \bisimtbrc Q_1$ and $Q^\ddag \bisimtbrc Q_2$. As a result, there exists a path $Q \pathtau Q_1 \step{\opt{\alpha}} Q_2$ such that, by transitivity of $\bisimtbrc$\,, $P \upto[\bisimtbrc] Q_1$ and $P' \upto[\bisimtbrc] Q_2$.
        \item If $\deadend{P}{X}$ and $P \step{\rt} P'$ then, since $P \bisimtbrc P^\dag$, there exists a path $P^\dag \pathtau P^\dag_1 \step{\rt} P^\dag_2$ with $P^\dag \bisimtbrc P^\dag_1$, $\deadend{P_1^\dag}{X}$ and $\theta_X(P') \bisimtbrc \theta_X(P^\dag_2)$. Since $P^\dag \tbisim Q^\dag$, there exists a path $Q^\dag \pathtau Q^\dag_1 \step{\rt} Q^\dag_2$ with $\deadend{Q^\dag_1}{X}$ and $\theta_X(P^\dag_2) \upto[\bisimtbrc] \theta_X(Q^\dag_2)$. Since $Q^\dag \bisimtbrc Q$, there exists a path $Q \pathtau Q_1 \step{\rt} Q_2$ with $\theta_X(Q^\dag_2) \bisimtbrc \theta_X(Q_2)$. As a result, there exists a path $Q \pathtau Q_1 \step{\rt} Q_2$ with $\theta_X(P') \upto[\bisimtbrc] \theta_X(Q_2)$.
        \item If $P \nsteptau$ then, since $P \bisimtbrc P^\dag$, there exists a path $P^\dag \pathtau P^\dag_0 \nsteptau$, and $P^\dag \bisimtbrc P^\dag_0$. Since $P^\dag \tbisim Q^\dag$, there exists a path $Q^\dag \pathtau Q^\dag_0 \nsteptau$. Since $Q^\dag \bisimtbrc Q$, there exists a path $Q \pathtau Q_0 \nsteptau$.
\popQED
    \end{enumerate}
\end{proof}

\noindent
The following lemma will be useful to deal with the matching of paths.

\begin{lemma} \label{lem:guarded}
    Let $H \mathbin\in \expr$ be well-guarded and have free variables from $W \mathbin\subseteq V$ only, and let $\rho,\nu \mathbin\in \closed^W\!\!$. 
    \begin{enumerate}
        \item $\init{H[\rho]} = \init{H[\nu]}$. \label{init}
        \item If $H[\rho] \step{\alpha} R$ with $\alpha \in Act$ then there exists $H' \in \expr$ with free variables in $W$ only such that $R = H'[\rho]$ and $H[\nu] \step{\alpha} H'[\nu]$.
        Moreover, in case $\alpha=\tau$, also $H'$ is well-guarded.\label{2}
    \end{enumerate}
\end{lemma}

\begin{proof}
    \newcommand{\spar}[1]{\mathbin{\|^{}_{#1}}}        % parallel composition
    \ref{init}.\ has been proven in \cite{strongreactivebisimilarity}. We obtain \ref{2}.\ by induction on the derivation of $H[\rho] \step\alpha R$, making a case distinction on the shape of $H$.

    Let $H=\alpha.G$, so that $H[\rho] = \alpha.G[\rho]$. Then $R = G[\rho]$ and $H[\nu] \step\alpha G[\nu]$. In case $\alpha=\tau$, also $G$ is well-guarded.

    The case $H=0$ cannot occur. Nor can the case $H=x\in V$, as $H$ is well-guarded.

    Let $H = H_1 \spar{S} H_2$, so that $H[\rho] = H_1[\rho]\spar{S} H_2[\rho]$. Note that $H_1$ and $H_2$ are well-guarded and have free variables in $W$ only. One possibility is that $a\notin S$, $H_1[\rho]\step\alpha R_1$ and $R= R_1 \spar{S} H_2[\rho]$. By induction, $R_1$ has the form $H'_1[\rho]$ for some term $H'_1\in\expr$ with free variables in $W$ only, and in case $\alpha=\tau$, also $H'_1$ is well-guarded. Moreover, $H_1[\nu] \step\alpha H'_1[\nu]$. Thus $R = (H'_1 \spar{S} H_2)[\rho]$, and $H':= H'_1 \spar{S} H_2$ has free variables in $W$ only. In case $\alpha=\tau$, $H$ is well-guarded. Moreover, $H[\nu] =  H_1[\nu]\spar{S} H_2[\nu] \step\alpha  H'_1[\nu]\spar{S} H_2[\nu] = H'[\nu]$.

    The other two cases for $\spar{S}$, and the cases for the operators $+$ and $\rename$, are equally trivial.

    Let $H= \theta_L^U(H^\dagger)$, so that $H[\rho] = \theta_L^U(H^\dagger[\rho])$. Note that $H^\dagger$ is well-guarded and has free variables in $W$ only. The case $\alpha = \tau$ is again trivial, so assume $\alpha\neq\tau$. Then {$H^\dagger[\rho] \step\alpha R$} and either $\alpha\in X$ or $\init{H^\dagger[\rho]} \cap (L\cup\{\tau\}) = \emptyset$. By induction, $R$ has the form $H'[\rho]$ for some term $H'\in\expr$ with free variables in $W$ only. Moreover, $H^\dagger[\nu] \step\alpha H'[\nu]$. Since $\init{H^\dagger[\rho]} = \init{(H^\dagger[\nu]}$ by Lemma~\ref{lem:guarded}.\ref{init}, either $\alpha\in X$ or $\init{H^\dagger[\nu]} \cap (L\cup\{\tau\}) = \emptyset$. Consequently, $H[\nu] = \theta_L^U(H^\dagger[\nu])\step\alpha H'[\nu]$.

    Let $H= \psi_X(H^\dagger)$, so that $H[\rho] = \psi_X(H^\dagger[\rho])$. Note that $H^\dagger$ is well-guarded and has free variables in $W$ only. The case $\alpha \in A\cup\{\tau\}$ is trivial, so assume $\alpha=\rt$. Then {$H^\dagger[\rho] \step\rt R^\dagger$} for some $R^\dagger$ such that $R=\theta_X(R^\dagger)$. Moreover, $H^\dagger[\rho] \cap (X\cup\{\tau\}) = \emptyset$. By induction, $R^\dagger$ has the form $H'[\rho]$ for some term $H'\in\expr$ with free variables in $W$ only. Moreover, $H^\dagger[\nu] \step\rt H'[\nu]$. Thus $R=(\theta_X(H'))[\rho]$ and $\theta_X(H')$ has free variables in $W$ only. Since $\init{H^\dagger[\rho]} = \init{(H^\dagger[\nu]}$ by Lemma~\ref{lem:guarded}.\ref{init}, $H^\dagger[\nu] \cap (X\cup\{\tau\}) = \emptyset$. Consequently, $H[\nu] = \psi_X(H^\dagger[\nu])\step\rt \theta_X(H'[\nu]) = (\theta_X(H'))[\nu]$.

    Finally, let $H = \langle x|\equa \rangle$, so that $H[\rho] = \langle x|\equa[\rho^\dagger]\rangle$, where $\rho^\dagger \in \closed^{W {\setminus} V_\equa}$ is the restriction of $\rho$ to $W {\setminus} V_\equa$. The transition $\langle\equa_x[\rho^\dagger]|\equa[\rho^\dagger]\rangle \step\alpha R$ is derivable through a subderivation of the one for $\langle x|\equa[\rho^\dagger]\rangle \step\alpha R$. Moreover, $\langle\equa_x[\rho^\dagger]|\equa[\rho^\dagger]\rangle = \langle\equa_x|\equa\rangle[\rho]$. So by induction, $R$ has the form $H'[\rho]$ for some term $H'\mathbin\in\expr$ with free variables in $W$ only, and $\langle\equa_x|\equa\rangle[\nu] \step\alpha H'[\nu]$. Moreover, in case $\alpha=\tau$, also $H'$ is well-guarded. Since $\langle\equa_x|\equa\rangle[\nu] = \langle\equa_x[\nu^\dagger]|\equa[\nu^\dagger]\rangle$, it follows that $H[\nu] = \langle x|\equa\rangle[\nu]= \langle x|\equa[\nu^\dagger]\rangle\step\alpha H'[\nu]$. 
\end{proof}

\begin{corollary} \label{cor:guarded path}
    Let $H \in \expr$ be well-guarded and have free variables from $W \subseteq V$ only, and let $\rho,\nu \in \closed^W$. If $H[\rho] \pathtau R \step{\alpha} S$ with $\alpha \in Act$ then there exists $H',H'' \in \expr$ with free variables in $W$ only such that $R = H'[\rho]$, $S = H'[\rho]$ and $H[\nu] \pathtau H'[\nu] \step{\alpha} H'[\nu]$.
\end{corollary}

\begin{proof}[Proof of Proposition \ref{prop:rsp}]
    It suffices to prove the proposition when $\rho, \nu \in \closed^{V_\equa}$ and only variables of $V_\equa$ can occur in the expressions $\equa_x$ for $x \in V_\equa$. Indeed, the general case requires to prove that, for all $\sigma: V \rightarrow \closed$, $\rho[\sigma] \bisimrtbrc \nu[\sigma]$. Let $\hat{\sigma}: V\setminus V_\equa \rightarrow \closed$ be defined as $\forall x \in V \setminus V_\equa, \hat{\sigma}(x) = \sigma(x)$. Since $\rho \bisimrtbrc \equa[\rho]$, $\rho[\sigma] \bisimrtbrc \equa[\rho][\sigma] = \equa[\hat{\sigma}][\rho[\sigma]]$, therefore, proving the proposition with $\rho[\sigma]$, $\nu[\sigma]$ and $\equa[\hat{\sigma}]$ is sufficient.

    It also suffices to prove the proposition for the case that $\equa$ is manifestly well-guarded. Indeed, if $\equa$ is well-guarded, let $\equa'$ be the manifestly well-guarded specification into which $\equa$ can be converted. Since $\bisimrtbrc$ is a lean congruence, a solution to $\equa$ up to $\bisimrtbrc$ is a solution to $\equa'$ up to $\bisimrtbrc$.
    
    Let $\equa$ be a manifestly well-guarded recursive specification with free variables from $V_\equa$ only, and $\rho, \nu \in \closed^{V_\equa}$ two of its solutions up to $\bisimrtbrc$. We are going to show that the symmetric closure of 
    \begin{align*}
        \tbisim := \{(H[\equa[\rho]],H[\equa[\nu]]) \mid H \in \expr \mbox{ is without $\tau_I$ and with free variables from } V_\equa \mbox{ only}\}
    \end{align*}
    is a \tb time-out bisimulation up to $\bisimrtbrc$. Here $\equa[\rho] := \{x = \equa_x[\rho] \mid x \in V_\equa\}\in\closed^{V_\equa}$ is employed as a substitution. Let $P,Q \in \closed$ such that $P \tbisim Q$. Then there exists $H \in \expr$ with free variables from $V_\equa$ only such that $P = H[\equa[\rho]]$ and $Q = H[\equa[\nu]]$, the other case being symmetrical. Note that $H[\equa[\rho]] = H[\equa][\rho]$. Since $H$ and $\equa$ have free variables from $V_\equa$ only, so does $H[\equa]$. Moreover, since $\equa$ is manifestly well-guarded, $H[\equa]$ is well-guarded.
    \begin{enumerate}
        \item Let $H[\equa[\rho]] \pathtau P_1 \step{\alpha} P_2$. By Corollary \ref{cor:guarded path}, there exists $H_1, H_2 \in \expr$ with free variables from $V_\equa$ only such that $P_1 = H_1[\rho]$, $P_2 = H_2[\rho]$ and $H[\equa][\nu] \pathtau H_1[\nu] \step{\alpha} H_2[\nu]$. Furthermore, since $\bisimrtbrc$ is a congruence and $\rho$ and $\nu$ are solutions of $\equa$ up to $\bisimrtbrc$\,, $H_1[\rho] \bisimrtbrc H_1[\equa[\rho]]$, $H_1[\nu] \bisimrtbrc H_1[\equa[\nu]]$, $H_2[\rho] \bisimrtbrc H_2[\equa[\rho]]$ and $H_2[\nu] \bisimrtbrc H_2[\equa[\nu]]$. Therefore, by definition of $\tbisim$, $H_1[\rho] \upto[\bisimrtbrc] H_1[\nu]$ and $H_2[\rho] \upto[\bisimrtbrc] H_2[\nu]$.
        \item Let $H[\equa[\rho]] \pathtau P_1 \step{\rt} P_2$ with $\deadend{P_1}{X}$. By Corollary \ref{cor:guarded path}, there exists $H_1,H_2 \in \expr$ with free variables from $V_\equa$ only such that $P_1 = H_1[\rho]$, $P_2 = H_2[\rho]$ and $H[\equa][\nu] \pathtau H_1[\nu] \step{\rt} H_2[\nu]$. Since $H_1$ is well-guarded, by Lemma \ref{lem:guarded}, $\deadend{H_1[\nu]}{X}$. Furthermore, since $\bisimrtbrc$ is a congruence and $\rho$ and $\nu$ are solutions of $\equa$ up to $\bisimrtbrc$\,, $H_2[\rho] \bisimrtbrc H_2[\equa[\rho]]$ and $H_2[\nu] \bisimrtbrc H_2[\equa[\nu]]$; therefore, by definition of $\tbisim$, $\theta_X(H_2[\rho]) \upto[\bisimrtbrc] \theta_X(H_2[\nu])$ (notice that $\theta_X(H_2[\equa][\rho]) = \theta_X(H_2)[\equa][\rho]$).
        \item Let $H[\equa[\rho]] \pathtau P_0 \nsteptau$. By Lemma \ref{lem:guarded}, there exists a well-guarded $H_1\in \expr$ with free variables from $V_\equa$ only such that $P_0 = H_0[\rho]$ and $H[\equa][\nu] \pathtau H_0[\nu]$. Since $H_0$ is well-guarded and $P_0 \nsteptau$, according to Lemma \ref{lem:guarded}.1, $H_0[\nu] \nsteptau$.
    \end{enumerate}

\noindent
    Next, we will prove that $\tbisim$ is a rooted \tb time-out bisimulation. Let $P,Q \in \closed$ such that $P \tbisim Q$. Then there exists $H \in \expr$ with free variables from $V_\equa$ only such that $P = H[\equa[\rho]]$ and $Q = H[\equa[\nu]]$, the other case being symmetrical. Note that $H[\equa[\rho]] = H[\equa][\rho]$. Since $H$ and $\equa$ have free variables from $V_\equa$ only, so does $H[\equa]$. Moreover, since $\equa$ is manifestly well-guarded, $H[\equa]$ is well-guarded.
    \begin{enumerate}
        \item Let $P \step{\alpha} P'$. By Lemma \ref{lem:guarded}, there exists $H' \in \expr$ with free variables from $V_\equa$ only such that $P' = H'[\rho]$ and $Q = H[\equa][\nu] \step{\alpha} H'[\nu]$. Furthermore, since $\bisimrtbrc$ is a congruence and $\rho$ and $\nu$ are solutions of $\equa$ up to $\bisimrtbrc\,$, $H'[\rho] \bisimrtbrc H'[\equa[\rho]]$ and $H'[\nu] \bisimrtbrc H'[\equa[\nu]]$. Therefore, by definition of $\tbisim$, $H'[\rho] \upto[\bisimrtbrc] H'[\nu]$. But, $\tbisim$ is a \tb time-out bisimulation up to $\bisimrtbrc\,$, thus, by Proposition~\ref{prop:up to b}, $H'[\rho] \bisimrtbrc H'[\nu]$.
        \item Let $\deadend{P}{X}$ and $P \step{\rt} P'$. By Lemma \ref{lem:guarded}, there exists $H' \in \expr$ with free variables from $V_\equa$ only such that $P' = H'[\rho]$ and $Q=H[\equa][\nu] \step{\rt} H'[\nu]$. Exactly as above, not even using $\deadend{P}{X}$, this implies $H'[\rho] \bisimrtbrc H'[\nu]$. Thus, since $\bisimrtbrc$ is a congruence, $\theta_X(H'[\rho]) \bisimrtbrc \theta_X(H'[\nu])$.
    \end{enumerate}
    By considering $H = x$ with $x \in V_\equa$, this yields $\equa_x[\rho] \bisimrtbrc \equa_x[\nu]$ and so $\rho(x) \bisimrtbrc \equa_x[\rho] \bisimrtbrc \equa_x[\nu] \bisimrtbrc \nu(x)$. Consequently, $\rho \bisimrtbrc \nu$.
\end{proof}

\section{Soundness of the Reactive Approximation Axiom} \label{app:RA}

\begin{lemma}\label{lem:theta twice}
    $\forall P \in \closed, \theta_X(P) \bisim \theta_X(\theta_X(P))$
\end{lemma}

\begin{proof}
    Trivial when considering the semantics of $\theta_X$.
\end{proof}

\begin{proof}[Proof of Proposition \ref{prop:soundness}]
    We show that $\tbisim := \{(P,Q),(Q,P) \mid \forall X \subseteq A, \; \psi_X(P) \bisimrtbrc \psi_X(Q)\}$ is a rooted \tb time-out bisimulation. Let $P,Q \in \closed$ such that $P \tbisim Q$. Thus, $\forall X \subseteq A, \psi_X(P) \bisimrtbrc \psi_X(Q)$.
    \begin{enumerate}
        \item If $P \step{\alpha} P'$ with $\alpha \in A_\tau$ then, by the semantics, $\psi_A(P) \step{\alpha} P'$. Since $\psi_A(P) \bisimrtbrc \psi_A(Q)$, there exists a transition $\psi_A(Q) \step{\alpha} Q'$ such that $P' \bisimtbrc Q'$. By the semantics, $Q \step{\alpha} Q'$.
        \item If $\deadend{P}{X}$ and $P \step{\rt} P'$ then, by the semantics, $\psi_X(P) \step{\rt} \theta_X(P')$. Since $\psi_X(P) \bisimrtbrc \psi_X(Q)$, there exists a transition $\psi_X(Q) \step{\rt} Q^\ddag$ with $\theta_X(\theta_X(P')) \bisimtbrc \theta_X(Q^\ddag)$. By the semantics, $Q^\ddag = \theta_X(Q')$ and $Q \step{\rt} Q'$. By Lemma~\ref{lem:theta twice}, $\theta_X(P') \bisim \theta_X(\theta_X(P')) \bisimtbrc \theta_X(\theta_X(Q')) \bisim \theta_X(Q')$.
    \popQED
    \end{enumerate}
\end{proof}

\section{Proofs of Completeness for Finite Processes} \label{app:completeness finite}

\begin{definition}\rm
    Call a process $P$ \emph{$\tau$-stable} if, for all transitions $P\steptau P^\dagger$, $P \not\bisimtbrc P^\dagger$.
\end{definition}

\begin{lemma}\label{lem:brb-stable}
If $P$ and $Q$ are $\tau$-stable and $P \bisimtbrc Q$ then $P \bisimrtbrc Q$.
\end{lemma}
\begin{proof}
    Assume that $P$ and $Q$ are $\tau$-stable and $P \bisimtbrc Q$. If $P \step{\alpha} P'$ with $\alpha \in A_\tau$, then there is a path $Q \pathtau Q_1 \step{\opt{\alpha}} Q_2$ with $P \bisimtbrc Q_1$ and $P' \bisimtbrc Q_2$. By symmetry and transitivity of $\bisimtbrc$ we have $Q_1 \bisimtbrc Q$, so by the $\tau$-stability of $Q$ it follows that $Q_1=Q$. Moreover, if $\alpha=\tau$ and $Q_2=Q_1$ then $P'\bisimtbrc P$, contradicting the $\tau$-stability of $P$. Thus $Q \step\alpha Q_2$. This argument also yields that $\init{P} = \init{Q}$.

    Furthermore, if $\init{P}\cap(X\cup\{\tau\}) = \emptyset$ and $P \step{\rt} P'$ then there is a path $Q \pathtau Q_1 \step{\rt} Q_2$ with $\theta_X(P') \bisimtbrc \theta_X(Q_2)$. By Lemma~\ref{lem:obvious}.1, Lemma~\ref{lem:obvious}.3 and the transitivity of $\bisimtbrc$\,, we have $P \bisimtbrc[X] Q_1$, $P \bisimtbrc Q_1$ and $Q \bisimtbrc Q_1$, respectively, so the $\tau$-stability of $Q$ yields $Q_1=Q$. Hence $Q \step{\rt} Q_2$. So indeed $P \bisimrtbrc Q$.
\end{proof}

\begin{proof}[Proof of Proposition \ref{prop:collapse}]
    We define the \textit{length} of a path $P_0 \step{\alpha_1} P_1 \step{\alpha_2} ... \step{\alpha_n} P_n$ to be $n$ and the \textit{depth} of a process $P$, denoted $d(P)$, to be the length of the longest path starting from $P$. It is well defined because $P$ is a recursion-free $\ccsp$ process. Note that $d(\theta_X(P))\leq d(P)$.

    $(\bisimtbrc)$ We will proceed by induction on $\max(d(P), d(Q))$. Let $n \mathbin\in\nat$ and suppose that the property holds for any recursion-free $\ccsp$ processes $P, Q$ such that $\max(d(P), d(Q)) < n$. Let $P,Q$ be two recursion-free $\ccsp$ processes such that $\max(d(P), d(Q)) = n$ and $P \bisimtbrc Q$.
    
    Since $P$ is recursion-free, there exists a path $P \pathtau P_0$ such that $P \bisimtbrc P_0$ and $P_0$ is $\tau$-stable. We are going to show that, for all $\alpha \in Act$, $Ax_r \vdash \alpha.\hat{P} = \alpha.\hat{P}_0$. If $P$ is $\tau$-stable then $P = P_0$ and this is trivial. Thus, suppose that $P$ is not $\tau$-stable, i.e., there exists $P \steptau P'$ with $P \bisimtbrc P'$. Then, as $P_0$ is $\tau$-stable, $P \ne P_0$ and so $d(P_0) < d(P)$.

    Let $I := \{(\alpha,P') \mid P \step{\alpha} P' \wedge \alpha \in A_\tau \wedge (\alpha \ne \tau \vee P \not\bisimtbrc P')\}$, listing the outgoing transitions of $P$ not labelled by $\rt$ and not elidable w.r.t.\ $\bisimtbrc\,$. Let $(\alpha,P') \in I$. Since $P \bisimtbrc P_0$, there exists a path $P_0 \pathtau P_1 \step{\opt{\alpha}} P_2$ such that $P \bisimtbrc P_1$ and $P' \bisimtbrc P_2$. Since $P_0$ is $\tau$-stable and $P_1 \bisimtbrc P \bisimtbrc P_0$, $P_0 = P_1$ and $P_0 \step{\opt{\alpha}} P_2$. If $\alpha = \tau$ then $P_0 \bisimtbrc P \not\bisimtbrc P' \bisimtbrc P_2$ so $P_0 \ne P_2$ and $P_0 \step{\alpha} P_2$. Since $\max(d(P_2),d(P')) < d(P)$, by induction, $Ax_r \vdash \alpha.\hat{P}' = \alpha.\hat{P}_2$, so by Lemma~\ref{lem:head-normal form} $Ax_r \vdash \alpha.P' = \alpha.P_2$. As a result, $Ax_r \vdash \hat{P}_0 = \sum_{(\alpha,P')\in I}\alpha.{P'} + \hat{P}_0$.

    Let $J := \{(\tau,P') \mid P \steptau P' \wedge P \bisimtbrc P'\}$, listing the outgoing $\tau$-transitions of $P$ elidable w.r.t.\ $\bisimtbrc\,$. So $J\neq\emptyset$. Let $(\tau,P') \in J$. Since $P' \bisimtbrc P \bisimtbrc P_0$ and $\max(d(P'),d(P_0)) < d(P)$, by induction, $Ax_r \vdash \tau.\hat{P'} = \tau.\hat{P_0}$, so by Lemma~\ref{lem:head-normal form} $Ax_r \vdash \tau.P' = \tau.\hat{P_0}$.

    Now, using $\hyperlink{Lt}{\mbox{\bf L}\tau}$, the following equality can be derived from $Ax_r$, for all $\beta \in Act$.
    \begin{align*}
        Ax_r \vdash \beta.\hat{P} = &~ \beta.(\sum_{\{(\alpha,P') \mid P\step{\alpha} P' \wedge \alpha \in A_\tau\}}\alpha.{P'}) = \beta.(\sum_{(\tau,P') \in J}\tau.{P}' + \sum_{(\alpha,P') \in I}\alpha.{P}') \\
        = &~ \beta.(\tau.\hat{P}_0 + \sum_{(\alpha,P') \in I}\alpha.{P'}) = \beta.(\tau.(\hat{P}_0 + \sum_{(\alpha,P') \in I}\alpha.{P'}) + \sum_{(\alpha,P') \in I}\alpha.{P'}) \\
        = &~ \beta.(\hat{P}_0 + \sum_{(\alpha,P') \in I}\alpha.{P'}) = \beta.\hat{P}_0
    \end{align*}
    
    Likewise, since $Q$ is recursion-free, a similar $\tau$-stable $Q_0$ can be defined. By the same reasoning, it can be proved that, for all $\alpha \in Act$, $Ax_r \vdash \alpha.\hat{Q} = \alpha.\hat{Q}_0$. Since $P_0$ and $Q_0$ are $\tau$-stable and $P_0 \bisimtbrc P \bisimtbrc Q \bisimtbrc Q_0$, $P_0 \bisimrtbrc Q_0$ according to Lemma \ref{lem:brb-stable}. 
    
    To end the proof, it suffices to show that, for all $\alpha \in Act$, $Ax_r \vdash \alpha.\hat{P}_0 = \alpha.\hat{Q}_0$, but \hypertarget{here}{we are going to prove the stronger statement $Ax_r \vdash \hat{P}_0 = \hat{Q}_0$}. Using the reactive approximation axiom, it suffices to prove that, for all $X \subseteq A$, $Ax_r \vdash \psi_X(P_0) = \psi_X(Q_0)$.

    Let $(\alpha,P_0') \in \{(\alpha,P_0') \mid P_0 \step{\alpha} P_0' \wedge \alpha\ne\rt\}$. Since $P_0 \bisimrtbrc Q_0$, there exists a transition $Q_0 \step{\alpha} Q_0'$ such that $P_0' \bisimtbrc Q_0'$. By induction, $Ax_r \vdash \alpha.\hat{P}_0' = \alpha.\hat{Q}_0'$.
    
    Let $X \subseteq A$ and $(\rt,P_0')$ such that $\deadend{P_0}{X}$ and $P_0 \step{\rt} P_0'$. Since $P_0 \bisimrtbrc Q_0$, there exists a transition $Q_0 \step{\rt} Q_0'$ such that $\theta_X(P_0') \bisimtbrc \theta_X(Q_0')$. By induction, $Ax_r \vdash \rt.\widehat{\theta_X(P_0')} = \rt.\widehat{\theta_X(Q_0')}$.

    Let $X \subseteq A$. If $\init{P_0}\cap(X\cup\{\tau\}) \ne \emptyset$ then $\init{Q_0}\cap(X\cup\{\tau\}) \ne \emptyset$ and, using Lemma~\ref{lem:head-normal form},
    \begin{align*}
        Ax_r \vdash \psi_X(Q_0) = & ~\sum_{\{(\alpha,Q_0') \mid Q_0 \step{\alpha} Q_0' \wedge \alpha \ne \rt\}}\alpha.Q_0' \\
        = & ~\sum_{\{(\alpha,Q_0') \mid Q_0 \step{\alpha} Q_0' \wedge \alpha \ne\rt\}}\alpha.Q_0' + \sum_{(\alpha,P_0') \in \{(\alpha,P_0') \mid P_0 \step{\alpha} P_0' \wedge \alpha\ne\rt\}}\alpha.P_0' \\
        = & ~\psi_X(P_0 + Q_0)
    \end{align*}
    If $\deadend{P_0}{X}$ then $\init{Q_0}\cap(X\cup\{\tau\}) = \emptyset$ and
    \begin{align*}
        Ax_r \vdash \psi_X(Q_0) = &~ \sum_{\{(\alpha,Q_0') \mid Q_0 \step{\alpha} Q_0' \wedge \alpha \ne \rt\}}\alpha.Q_0' + \sum_{\{(t,Q_0') \mid Q_0 \step{\rt} Q_0'\}}\rt.\theta_X(Q_0') \\
        = &~ \sum_{\{(\alpha,Q_0') \mid Q_0 \step{\alpha} Q_0' \wedge \alpha \ne\rt\}}\alpha.Q_0' + \sum_{(\alpha,P_0') \in \{(\alpha,P_0') \mid P_0 \step{\alpha} P_0' \wedge \alpha\ne\rt\}}\alpha.P_0' \\
        &+ \sum_{\{(\rt,Q_0') \mid Q_0 \step{\rt} Q_0'\}}\rt.\theta_X(Q_0') + \sum_{\{(\rt,P_0') \mid P_0 \step{\rt} P_0'\}}\rt.\theta_X(P_0') \\
        = &~ \psi_X(P_0 + Q_0)
    \end{align*}
    As a result, for all $X \subseteq A$, $Ax_r \vdash \psi_X(Q_0) = \psi_X(P_0+Q_0)$, and so,
    $Ax_r \vdash Q_0 = P_0+Q_0$. Symmetrically, $Ax_r \vdash P_0 = P_0+Q_0$. Therefore, $Ax_r \vdash P_0=Q_0$.

\bigskip

    $(\bisimb)$ We will proceed by induction on $\max(d(P), d(Q))$. Let $n \mathbin\in\nat$ and suppose that the property holds for any recursion-free $\ccsp$ processes $P, Q$ such that $\max(d(P), d(Q)) < n$. Let $P,Q$ be two recursion-free $\ccsp$ processes such that $\max(d(P), d(Q)) = n$ and $P \bisimb Q$.

    Since $P$ is recursion-free, there exists a path $P \pathtau P_0$ such that $P \bisimb P_0$ and, for all $P_0 \steptau P'$, $P_0 \not\bisimb P'$. We are going to show that, for all $\alpha \in Act$, $Ax \vdash \alpha.\hat{P} = \alpha.\hat{P_0}$. If for all $P \steptau P'$, $P \not\bisimb P'$ then $P = P_0$ and it is trivial. Thus, suppose there exists $P \steptau P'$ such that $P \bisimb P'$. Then $P_0 \ne P$ and $d(P_0) < d(P)$.

    Let $J := \{(\tau,P') \mid P \steptau P' \wedge P \bisimb P'\}$, listing the outgoing $\tau$-transitions of $P$ that can be elided w.r.t.\ $\bisimb\,$. Let $(\tau,P') \in J$. Since $P' \bisimb P \bisimb P_0$ and $\max(d(P'),d(P_0)) < d(P)$, by induction, $Ax \vdash \tau.\hat{P'} = \tau.\hat{P}_0$.
    
    Let $I := \{(\alpha,P') \mid P \step{\alpha} P'\} \setminus J$, listing the outgoing transitions of $P$ that cannot be elided. Let $(\alpha,P') \in I$. Since $P \bisimb P_0$, there exists a path $P_0 \pathtau P_1 \step{\opt{\alpha}} P_2$ such that $P \bisimb P_1$ and $P' \bisimb P_2$. Thus, $P_1 \bisimb P \bisimb P_0$, but, for all $P_0 \steptau P^\dag$, $P_0 \,\not\!\bisimb P^\dag$, so $P_0 = P_1$ and $P \step{\opt{\alpha}} P_2$. Since $(\alpha,P') \not\in J$, $\alpha \in A$ or $P_0 \bisimb P \not\bisimb P' \bisimb P_2$ so $P_0 \step{\alpha} P_2$ and $\max(d(P_2),d(P')) < d(P)$. Thus, by induction, $Ax \vdash \alpha.\hat{P'} = \alpha.\hat{P}_2$. As a result, $Ax \vdash \hat{P}_0 + \sum_{(\alpha,P') \in I}\alpha.\hat{P'} = \hat{P}_0$.

    Since there exists $P \steptau P'$ with $P \bisimb P'$, $J \ne \emptyset$.
    \begin{align*}
        Ax \vdash \alpha.\hat{P} = &~ \alpha.(\sum_{(\tau,P') \in J}\tau.\hat{P}' + \sum_{(\alpha,P') \in I}\alpha.\hat{P}') = \alpha.(\tau.\hat{P}_0 + \!\sum_{(\alpha,P') \in I}\!\alpha.\hat{P}') \\
        = &~ \alpha.(\tau.(\hat{P}_0 + \sum_{(\alpha,P') \in I}\alpha.\hat{P}') + \sum_{(\alpha,P') \in I}\alpha.\hat{P}') = \alpha.\hat{P}_0 
    \end{align*}
    Similarly, since $Q$ is recursion-free, there exists a recursion-free $\ccsp$ process $Q_0$ such that $Q \pathtau Q_0$, $Q \bisimb Q_0$ and, for all $Q_0 \steptau Q^\dag$, $\neg(Q_0 \bisimb Q^\dag)$. Moreover, for all $\alpha \in Act$, $Ax \vdash \alpha.\hat{Q} = \alpha.\hat{Q}_0$. Notice that $P_0 \bisimb P \bisimb Q \bisimb Q_0$ and, since, for all $Q_0 \steptau Q^\dag$, $\neg(Q_0 \bisimb Q^\dag)$ and, for all $P_0 \steptau P^\dag$, $\neg(P_0 \bisimb P^\dag)$, $P_0 \bisimrb Q_0$.

    Let $(\alpha,P_0')$ such that $P_0 \step{\alpha} P_0'$. Since $P_0 \bisimrb Q_0$, there exists a path $Q_0 \step{\alpha} Q_2$ such that $P_0' \bisimb Q_2$. Since $\max(d(P_0'),d(Q_2)) < n$, by induction, $Ax \vdash \alpha.P_0' = \alpha.Q_2$. As a result, $Ax \vdash \hat{P}_0 + \hat{Q}_0 = \hat{Q}_0$. Symmetrically, $Ax \vdash \hat{P}_0 + \hat{Q}_0 = \hat{P}_0$, and so, $Ax \vdash \hat{P}_0 = \hat{Q}_0$. Finally, for all $\alpha \in Act$, $Ax \vdash \alpha.\hat{P} = \alpha.\hat{Q}$.
\end{proof}

\begin{proof}[Proof of Theorem \ref{thm:completeness finite}]
    Let $P,Q \in \closed$ be two recursion-free $\ccsp$ processes. Let $P \bisimrtbrc Q$. $P$ and $Q$ can be equated in the same manner as \hyperlink{here}{$P_0$ and $Q_0$} in the proof of Proposition \ref{prop:collapse}. 

    Suppose that $P \bisimrb Q$. Let $(\alpha,P')$ such that $P \step{\alpha} P'$ with $\alpha \in Act$. Since $P \bisimrb Q$, there exists a transition $Q \step{\alpha} Q'$ such that $P' \bisimb Q'$. According to the previous proposition, $Ax \vdash \alpha.P' = \alpha.Q'$, thus, 
    \begin{align*}
        Ax \vdash Q = \sum_{\{(\alpha,Q') \mid Q \step{\alpha} Q'\}}\alpha.Q' = \sum_{\{(\alpha,Q') \mid Q \step{\alpha} Q'\}}\alpha.Q' + \sum_{\{(\alpha,P') \mid P \step{\alpha} P'\}}\alpha.P' = Q + P
    \end{align*}
    As a result, $Ax \vdash Q = P+Q$. Symmetrically, $Ax \vdash P = P+Q$. Therefore, $Ax \vdash P=Q$.
\end{proof}

\section{Proof of Completeness by Equation Merging} \label{app:completeness}

\begin{proof}[Proof of Theorem \ref{thm:completeness}]
    Let $E_0$ and $F_0$ two strongly guarded $\ccsp$ processes such that $E_0 \bisimrb F_0$. We are going to build a recursive specification $\equa$ such that $E_0$ and $F_0$ will be components of solutions of $\equa$ in the same variable. Let $\mathcal{E}_{E_0}$ (resp.\ $\mathcal{E}_{F_0}$) be the set of reachable expressions from $E_0$ (resp.\ $F_0$). Let $V_\equa$ be a set of fresh variables $\{x_{EF} \mid (E,F) \in \mathcal{E}_{E_0}\times\mathcal{E}_{F_0} \wedge E \bisimb F\}$. We denote $x_0 = x_{E_0F_0} \in V_\equa$ and we define the following set of equations $\equa$, for all $x_{EF} \in V_\equa$, with $\alpha\in A \cup\{\tau,\rt\}$.
    \begin{align*}
        \equa_{x_{EF}} := & \sum_{E \step{\alpha} E', F \step{\alpha} F', E' \bisimb F'}\hspace{-8pt}\alpha.x_{E'F'}  \\
        &+ \hspace{-8pt}\sum_{E \steptau E', E' \bisimb F, x_{EF} \neq x_0}\hspace{-8pt}\tau.x_{E'F} + \hspace{-8pt}\sum_{F \steptau F', E \bisimb F', x_{EF} \neq x_0}\hspace{-8pt}\tau.x_{EF'}
    \end{align*}
    Note that $\equa$ is well-guarded since $E_0$ and $F_0$ are strongly guarded $\ccsp$ processes. For $x_{EF} \in V_\equa$, we define $H_{EF}, G_{EF} \in \expr$ such that 
    \begin{align*}
        H_{EF} := & \sum_{E \step{\alpha} E', F \step{\alpha} F', E' \bisimb F'}\hspace{-8pt}\alpha.E' + \hspace{-8pt}\sum_{E \steptau E', E' \bisimb F, x_{EF} \neq x_0}\hspace{-8pt}\tau.E' \\
        G_{EF} := &  \begin{cases}  
                        H_{EF} + \tau.E & \mbox{if }x_0 \mathbin{\neq} x_{EF} \mbox{, } \exists F \steptau F', E \bisimb F'  \\
                        E & \mbox{otherwise}
                    \end{cases}
    \end{align*}
    According to Lemma~\ref{lem:head-normal form}, for all $(E,F) \in \mathcal{E}_{E_0}\times\mathcal{E}_{F_0}$, $Ax^\infty \vdash E + H_{EF} = \widehat{(E + H_{EF})} = \hat{E} = E$. Let $(E,F) \in \mathcal{E}_{E_0}\times\mathcal{E}_{F_0}$. If $x_0 \ne x_{EF}$ and $\exists F\steptau F', E \bisimb F'$ then, for all $\alpha \in Act$, $Ax^\infty \vdash \alpha.G_{EF} = \alpha.(H_{EF} + \tau.E) = \alpha.E$ using the branching axiom. In any case, for all $\alpha \in Act$, $Ax^\infty \vdash \alpha.G_{EF} = \alpha.E$. \hfill $(*)$
    
    If we prove that the family $(G_{EF})_{(E,F) \in \mathcal{E}_{E_0}\times\mathcal{E}_{F_0}}$ is a solution of $\equa$ then, by definition of $G_{E_0F_0}$, there would exist a solution whose value for the variable $x_0$ is $E$. According to $(*)$, we need to prove that, for all $x_{EF} \in V_\equa$, 
    \begin{align*}
        Ax^\infty \vdash G_{EF} = & \sum_{E \step{\alpha} E', F \step{\alpha} F', E' \bisimb F'}\hspace{-8pt}\alpha.G_{E'F'} \\ 
        & + \hspace{-8pt}\sum_{E \steptau E', E' \bisimb F, x_{EF} \neq x_0}\hspace{-8pt}\tau.G_{E'F} + \hspace{-8pt}\sum_{F \steptau F', E \bisimb F', x_{EF} \neq x_0}\hspace{-8pt}\tau.G_{EF'} \\
        = & \sum_{E \step{\alpha} E', F \step{\alpha} F', E' \bisimb F'}\hspace{-8pt}\alpha.E' \\ 
        & + \hspace{-8pt}\sum_{E \steptau E', E' \bisimb F, x_{EF} \neq x_0}\hspace{-8pt}\tau.E' + \hspace{-8pt}\sum_{F \steptau F', E \bisimb F', x_{EF} \neq x_0}\hspace{-8pt}\tau.E \\
        = &~ H_{EF} + \hspace{-8pt}\sum_{x_{EF} \neq x_0, F \steptau F', E \bisimb F'}\hspace{-8pt}\tau.E 
    \end{align*}
    \begin{itemize}
        \item If $x_{EF} \neq x_0$ and $\exists F \steptau F', E \bisimb F'$ then this follows from the definition of $G_{EF}$.
        \item If $x_{EF} \neq x_0$ and $\forall F\steptau F', E \,\not\!\bisimb F'$ then, by definition of $G_{EF}$, we have to prove $Ax^\infty \vdash E = H_{EF}$. Let $(\alpha,E')$ such that $E \step{\alpha} E'$ and $\alpha \in Act$. Since $E \bisimb F$, there exists a path $F \pathtau F_1 \step{\opt{\alpha}} F_2$ such that $E \bisimb F_1$ and $E' \bisimb F_2$. Since $\forall F\steptau F', \neg(E \bisimb F')$, $F = F_1$, so there exists a transition $F \step{\opt{\alpha}} F_2$ such that either $F \step{\alpha} F_2$ and $E' \bisimb F_2$, or $\alpha = \tau$ and $E' \bisimb F$. In either case, $H_{EF} \step{\alpha} E'$. \\
        As a result, $Ax^\infty \vdash E = \hat{E} = \widehat{E + H}_{EF} = \hat{H}_{EF} = H_{EF}$.
        \item If $x_{EF} = x_0$ then $E = E_0$, $F = F_0$ and we have to show that $Ax^\infty \vdash E_0 = H_{E_0F_0} = \sum_{E_0 \step{\alpha} E', F_0 \step{\alpha} F', E' \bisimb F'}\alpha.E'$. Let $(\alpha,E')$ such that $E_0 \step{\alpha} E'$. Since $E_0 \bisimrb F_0$, there exists a transition $F_0 \step{\alpha} F'$ such that $E' \bisimb F'$. $Ax^\infty \vdash E = \hat{E} = \hat{H}_{EF} = H_{EF}$.
    \popQED
    \end{itemize}
Note that we could define $H'_{EF}$ and $G'_{EF}$ by reverting the role of $E$ and $F$ and also get a solution whose value for the variable $x_0$ is $F_0$. Consequently, RSP yields $Ax^\infty \vdash E_0 = F_0$.
\end{proof}

\section{The Canonical Representative} \label{app:canonical rep}

We are going to start by proving some lemmas facilitating the handling of classes.

\begin{lemma} \label{lem:transition class}
    Let $P \in \closed^g$.
    \begin{enumerate}
        \item $\forall \alpha \in A_\tau, ([P] \step{\alpha} R' \Leftrightarrow \exists P \pathtau P_1 \step{\alpha} P_2, (P_1,P_2) \in [P]\times R' \wedge (\alpha \in A \vee [P] \ne R'))$.
        \item ${[P] \nsteptau} \Leftrightarrow \exists P_0 \in [P], P \pathtau P_0 \nsteptau$.
        \item Let $X\subseteq A$. Then $\deadend{[P]}{X} \Leftrightarrow \exists P \pathtau P_0, P_0 \in [P] \wedge \deadend{P_0}{X}$.
        \item If $[P] \step{\rt} R'$ and $\deadend{[P]}{X}$ then $\exists P \pathtau P_1 \step{\rt} P_2$, $P_1 \in [P] \wedge \deadend{P_1}{X} \wedge \theta_X(P_2) \in [\theta_X(\chi(R'))]$.
        \item If $\exists X \mathbin\subseteq A,\, \exists P \pathtau P_1 \step{\rt} P_2, P_1 \in [P] \wedge \deadend{P_1}{X}$ then there exists an $R'$ with $[P] \step{\rt} R' \wedge \theta_X(P_2) \in [\theta_X(\chi(R'))]$.
    \end{enumerate}
\end{lemma}

\begin{proof}
    Let $P \in \closed^g$.
    \begin{enumerate}
        \item Let $\alpha \in A_\tau$.
        \begin{itemize}
            \item If $[P] \step{\alpha} R'$ then, by definition of $\rightarrow$, there exists a path $\chi([P]) \pathtau P_1 \step{\alpha} P_2$ such that $P_1 \in [P]$, $P_2 \in R'$ and $\alpha \in A \vee [P] \ne R'$. Since $\chi([P]) \bisimtbrc P$, there exists a path $P \pathtau P_1' \step{\opt{\alpha}} P_2'$ such that $P_1 \bisimtbrc P_1'$ and $P_2 \bisimtbrc P_2'$, thus, $P_1' \in [P]$ and $P_2' \in R'$. If $\alpha = \tau$ then $[P] \ne R'$, so $P_1' \,\not\!\bisimtbrc P_2'$ and so $P_1' \step{\alpha} P_2$, otherwise, $P_1' \step{\alpha} P_2'$. 
            \item If there exists a path $P \pathtau P_1 \step{\alpha} P_2$ such that $P_1 \in [P]$, $P_2 \in R'$ and $\alpha \in A \vee [P] \ne R'$ then, since $P \bisimtbrc \chi([P])$, there exists a path $\chi([P]) \pathtau P_1' \step{\opt{\alpha}} P_2'$ such that $P_1 \bisimtbrc P_1'$ and $P_2 \bisimtbrc P_2'$, thus, $P'_1 \in [P]$ and $P'_2 \in R'$. If $\alpha = \tau$ then $[P] \ne R'$, therefore, $P_1' \,\not\!\bisimtbrc P_2'$ and so $P_1' \step{\alpha} P_2'$, otherwise, $P_1' \step{\alpha} P_2'$. By definition of $\rightarrow$, $[P] \step{\alpha} R'$.
        \end{itemize}
        \item 
        \begin{itemize}
            \item If $[P] \steptau$ then, according to the previous point, there exists a path $P \pathtau P_1 \steptau P_2$ such that $P_1 \in [P]$ and $P_2 \not\in [P]$. Suppose that there is a path $P \pathtau P_0 \nsteptau$ with $P_0 \in [P]$. Then $P_0 \bisimtbrc P$, so there exists a path $P_0 \pathtau P^\dag \step{\opt{\tau}} P^\ddag$ such that $P_1 \bisimtbrc P^\dag$ and $P_2 \bisimtbrc P^\ddag$. Since $P_0 \nsteptau$, $P_0 = P^\ddag \bisimtbrc P_2 \not\in [P]$, but that's impossible.
            \item Suppose that, for all paths $P \pathtau P_0 \nsteptau$, $P_0 \not\in [P]$. Since $P$ is strongly guarded, there exists a path $P \pathtau P_1$ such that $P_1 \in [P]$ and, for all $P_1 \steptau P'$, $P_1 \,\not\!\bisimtbrc P'$. Since $P \pathtau P_1$ and $P_1 \in [P]$, there exists a transition $P_1 \steptau P_2$, and $P_1 \not\!\bisimtbrc P_2$. Thus, there exists a path $P \pathtau P_1 \steptau P_2$ such that $P_1 \in [P]$ and $P_2 \not\in [P]$. According to the previous point, $[P] \steptau [P_2]$.
        \end{itemize}
        \item This a corollary of the two previous points.
        \item Suppose that $[P] \step{\rt} R'$ and $\deadend{[P]}{X}$, then, by definition of $\rightarrow$, there exists a path $\chi([P]) \pathtau P'_1 \step{\rt} P'_2$ such that $P'_1 \in [P]$, $P'_1 \nsteptau$ and $P'_2 \in R'$. Since $\deadend{[P]}{X}$, there exists a path $P \pathtau P_0$ such that $P_0 \in [P]$ and $\deadend{P_0}{X}$, thus, since $P_0 \bisimtbrc P'_1$ and $P'_1 \nsteptau$, $\deadend{P'_1}{X}$. Since $P'_1 \bisimtbrc P$ and $P'_1 \step{\rt} P'_2$, there exists a path $P \pathtau P_1 \step{\rt} P_2$ such that $P_1 \bisimtbrc P_1'$, $\deadend{P_1}{X}$ and $\theta_X(P_2) \bisimtbrc \theta_X(P_2')$. As a result, there exists a path $P \pathtau P_1 \step{\rt} P_2$ such that $P_1 \in [P]$, $\deadend{P_1}{X}$ and $\theta_X(P_2) \in [\theta_X(P_2')] = [\theta_X(\chi(R'))]$.
         \item Suppose there exists $X \subseteq A$ and a path $P \pathtau P_1 \step{\rt} P_2$ such that $P_1 \in [P]$ and $\deadend{P_1}{X}$. Since $P \bisimtbrc \chi([P])$, there exists a path $\chi([P]) \pathtau P_1' \step{\rt} P_2'$ such that $P_1 \bisimtbrc P_1'$, $\deadend{P'_1}{X}$ and $\theta_X(P_2) \bisimtbrc \theta_X(P'_2)$ and so $P_1' \in [P]$. Therefore, by definition of $\rightarrow$, $[P] \step{\rt} [P_2']$ and $\theta_X(P_2) \in [\theta_X(\chi([P'_2]))]$.
\popQED
    \end{enumerate}
\end{proof}

\begin{corollary}  \label{cor:transition class}
    Let $P \in \closed^g$ and $X \subseteq A$.
    \begin{enumerate}
        \item If $[P] \pathtau R'$ and $\init{R'}\cap (X \cup \{\tau\})=\emptyset$ then $\exists P \pathtau P'\mathbin\in R'$ with $\init{P'}\cap (X \cup \{\tau\})=\emptyset$.
          % \wedge [P] \ne R')$.
        \item If $\theta_X(P) \in [\theta_X(\chi(R))]$, $R \pathtau R'$ and $\init{R'}\cap (X \cup \{\tau\})=\emptyset$ then $\exists P \pathtau P' \in R'$ with $\init{P'}\cap (X \cup \{\tau\})=\emptyset$.
    \end{enumerate}
\end{corollary}

\begin{proof}
    The first statement follows directly from Lemma~\ref{lem:transition class}.1--3. For the second, suppose $\theta_X(P) \in [\theta_X(\chi(R))]$, $R \pathtau R'$ and $\init{R'}\cap (X \cup \{\tau\})=\emptyset$. By the first statement, $\chi(R)\pathtau Q'\in R'$ for some $Q'$ with $\init{Q'}\cap (X \cup \{\tau\})=\emptyset$.  By the semantics of $\theta_X$, there is a path $\theta_X(\chi(R)) \pathtau \theta_X(Q')\nsteptau$. Since $\theta_X(P) \bisimtbrc \theta_X(\chi(R))$, there is a path $\theta_X(P) \pathtau P^\dag \nsteptau$ with $P^\dag \bisimtbrc \theta_X(Q')$. By the semantics, $P^\dag = \theta_X(P')$ for some $P'$ with $P\pathtau P'\nsteptau$. So $P' \bisimtbrc[X] Q'$ by Proposition~\ref{prop:time-out bisim}.2, and Lemma~\ref{lem:obvious}.3 yields $P' \bisimtbrc Q'$. Thus $P'\in R'$ and Lemma~\ref{lem:obvious}.2 gives $\init{P'}\cap (X \cup \{\tau\})=\emptyset$.
\end{proof}

\begin{remark}\label{tortau}
    Let $R\in[\closed^g]$. If $R \step\rt$ then $R\nsteptau$.
\end{remark}

\begin{proof}
    Suppose $R \step\rt R'$. By the definition in Section~\ref{subsec:canonical}, there is a path $\chi(R)\pathtau P_1 \nsteptau$ with $P_1\in R$. So by Lemma~\ref{lem:transition class}.2 $R\nsteptau$.
\end{proof}

\begin{definition}\rm\label{def:uptoRT}
    A \emph{concrete branching time-out bisimulation up to reflexivity and transitivity} is a symmetric relation ${\B}$ on $\closed^g \uplus [\closed^g] \uplus \{\theta_X([P]) \mathbin| X \mathbin\subseteq A \wedge P\mathbin\in \closed^g\}$, such that, for all $P^\dag\!\B Q$,
    \begin{itemize}\itemsep 0pt \parsep 0pt
        \item if $P^\dag \step{\alpha} P^\ddagger$ with $\alpha\mathbin\in A_\tau$, then $\exists$ path $Q\pathtau Q^\dag \step{\opt{\alpha}} Q^\ddagger$ with $P^\dag \B^* Q^\dag$ and $P^\ddagger \B^* Q^\dagger$,
        \item if $P^\dag \step\rt P^\ddagger$ with $\init{P^\dag}\cap (X\cup\{\tau\})=\emptyset$, then there is a path $Q\pathtau Q^\dag \step{\rt} Q^\ddagger$ with $\init{Q^\dag}\cap (X\cup\{\tau\})=\emptyset$ and $\theta_{X}(P^\ddagger) \B^* \theta_X(Q^\ddagger)$,
        \item if $P^\dag \nsteptau$ then there is a path $Q\pathtau Q^\dag \nsteptau$.
    \end{itemize}
    Here $\B^*  := \{(P^\dag,Q^\dag) \mid \exists n\geq 0. ~ \exists P_0,\dots, P_n.~ P^\dag = P_0 \B P_1 \B \dots \B P_n = Q^\dag\}$.
\end{definition}

\begin{proposition}\label{prop:uptoRT}
    If $P \B Q$ for a concrete branching time-out bisimulation $\B$ up to reflexivity and transitivity, then $P \bisimtbrc Q$.
\end{proposition}

\begin{proof}
    It suffices to show that $\B^*$ is a \tb time-out bisimulation. Clearly this relation is symmetric.
    \begin{itemize}
        \item Suppose $P_0 \B P_1 \B \dots \B P_n$ for some $n\geq 0$ and $P \pathtau P^\dag_0 \step{\opt{\alpha}} P^\ddagger_0$ with $\alpha\in A_\tau$. It suffices to find $P^\dag_n,P^\ddagger_n$ such that $P_n \pathtau P^\dag_n \step{\opt{\alpha}} P^\ddag_n$, $P^\dag_0 \B^* P^\dag_n$ and $P^\ddag_0 \B^* P^\ddag_n$. (In fact, we need this only in the special case where $P_0 = P^\dag_0 \neq P^\ddag_0$, but establish the more general claim.) We proceed with induction on $n$. The case $n=0$ is trivial.

        Fixing an $n>0$, by Definition~\ref{def:uptoRT} there are $P^\dag_1,P^\ddag_1$ such that $P_1 \pathtau P^\dag_1 \step{\opt{\alpha}} P^\ddag_1$, $P^\dag_0 \B^* P^\dag_1$ and $P^\ddag_0 \B^* P^\ddag_1$. Now by induction there are $P^\dag_n,P^\ddagger_n$ such that $P_n \pathtau P^\dag_n \step{\opt{\alpha}} P^\ddag_n$, $P^\dag_1 \B^* P^\dag_n$ and $P^\ddag_1 \B^* P^\ddag_n$. Hence $P^\dag_0 \B^* P^\dag_n$ and $P^\ddag_0 \B^* P^\ddag_n$.
        
        \item Suppose $P_0 \B P_1 \B \dots \B P_n$ for some $n\geq 0$ and there is a path $P_0 \pathtau P^\dag_0 \step\rt P^\ddagger_0$ with $\init{P^\dag_0}\cap(X\cup\{\tau\})=\emptyset$. It suffices to find $P^\dag_n,P^\ddagger_n$ such that $P_n \pathtau P^\dag_n \step\rt P^\ddagger_n$, $\init{P^\dag_n}\cap(X\cup\{\tau\})=\emptyset$ and $\theta_X(P^\ddagger_0) \B^* \theta_X(P^\ddagger_n)$. (In fact, we need this only in the special case where $P^\dag_0 = P_0$, but establish the more general claim.) We proceed with induction on $n$. The case $n=0$ is trivial.

        Fixing an $n>0$, by Definition~\ref{def:uptoRT} there exist $P^\dag_1,P^\ddagger_1$ such that $P_1 \pathtau P^\dag_1 \step\rt P^\ddagger_1$, $\init{P^\dag_1}\cap(X\cup\{\tau\})=\emptyset$ and $\theta_X(P^\ddagger_0) \B^* \theta_X(P^\ddagger_1)$. By induction there are $P^\dag_n,P^\ddagger_n$ with $P_n \mathbin{\pathtau} P^\dag_n \mathbin{\step\rt} P^\ddagger_n$, $\init{P^\dag_n}\!\cap\!(X\!\cup\!\{\tau\})\mathbin=\emptyset$ and $\theta_X(P^\ddagger_1) \B^* \theta_X(P^\ddagger_n)$. Hence $\theta_X(P^\ddagger_0) \B^* \theta_X(P^\ddagger_n)$.
    
        \item Suppose $P_0 \B P_1 \B \dots \B P_n$ for some $n\geq 0$ and there is a path $P_0 \pathtau P^\dag_0\nsteptau$. It suffices to find a path $P_n \pathtau P^\dag_n\nsteptau$. (In fact, we need this only in the special case where $P^\dag_0 = P_0$, but establish the more general claim.) We proceed with induction on $n$. The case $n=0$ is trivial.

        Fixing an $n>0$, by Definition~\ref{def:uptoRT} there exists a path $P_1 \pathtau P^\dag_1\nsteptau$. By induction, there exists a path $P_n \pathtau P^\dag_n\nsteptau$.
        \popQED          
    \end{itemize}
\end{proof}

\begin{lemma} \label{lem:class deadlock theta}
    Let $P \in \closed^g$ and $X \subseteq A$. If $\deadend{[P]}{X}$ then $[P] \bisim \theta_X([P])$.
\end{lemma}

\begin{proof}
    It suffices to see that $\tbisim := \{([P],\theta_X([P])),(\theta_X([P]),[P])\} \cup Id$ is a strong bisimulation thanks to the semantics of $\theta_X$.
\end{proof}

\begin{lemma} \label{lem:theta class}
    Let $P \in \closed^g$ and $X \subseteq A$. Then $\theta_X([P]) \bisimtbrc[] [\theta_X(P)]$.
\end{lemma}

\begin{proof}
    We will show that ${\tbisim} := {{\bisim} \cup \!\{(\theta_X([P]),[\theta_X(\!P)]), ([\theta_X(\!P)],\theta_X([P])) \!\mid\! P \mathbin\in \closed^g\! \wedge X \mathbin\subseteq A\}}$ is a concrete branching time-out bisimulation up to reflexivity and transitivity.
     \begin{itemize}
        \item If $\theta_X([P]) \steptau R'$ then $[P] \steptau R^\dag$ with $R' = \theta_X(R^\dag)$. According to Lemma \ref{lem:transition class}.1, there exists a path $P \pathtau P_1 \steptau P_2$ such that $P_1 \in [P]$ and $P_2 \in R^\dag = [P_2]$ and $[P] \ne R^\dag$. Thus, there exists a path $\theta_X(P) \pathtau \theta_X(P_1) \steptau \theta_X(P_2)$ such that $\theta_X(P_1) \in [\theta_X(P)]$. If $\theta_X(P_1) \bisimtbrc \theta_X(P_2)$ then $[\theta_X(P)] = [\theta_X(P_2)]$. Otherwise, $[\theta_X(P)] \ne [\theta_X(P_2)]$, thus, according to Lemma \ref{lem:transition class}.1, $[\theta_X(P)] \steptau [\theta_X(P_2)]$. In either case, there exists a transition $[\theta_X(P)] \step{\opt{\tau}} [\theta_X(P_2)]$ such that, by definition of $\tbisim$, $[\theta_X(P_2)] \tbisim \theta_X([P_2]) = R'$.
        \item If $[\theta_X(P)] \steptau R'$ then, according to Lemma \ref{lem:transition class}.1, there exists a path $\theta_X(P) \pathtau P_1 \steptau P_2$ such that $P_1 \in [\theta_X(P)]$, $P_2 \in R'$ and $[\theta_X(P)] \ne R'$. Thus, there exists a path $P \pathtau P^\dag \steptau P^\ddag$ such that $P_1 = \theta_X(P^\dag)$ and $P_2 = \theta_X(P^\ddag)$. Notice that, since $[P_1] \ne [P_2]$, $[P^\dag] \ne [P^\ddag]$. According to Lemma \ref{lem:transition class}.1, there exists a path $[P] \pathtau [P^\dag] \steptau [P^\ddag]$. Thus, $\theta_X([P]) \pathtau \theta_X([P^\dag]) \steptau \theta_X([P^\ddag])$. Moreover, by definition of $\tbisim$, $\theta_X([P^\dag]) \tbisim [\theta_X(P^\dag)] = [\theta_X(P)]$ and $\theta_X([P^\ddag]) \tbisim [\theta_X(P^\ddag)] = R'$.
        \item If $\theta_X([P]) \step{a} R'$ with $a \in A$ then $[P] \step{a} R'$ and $a \in X \vee \deadend{[P]}{X}$. If $\deadend{[P]}{X}$, according to Lemma \ref{lem:transition class}.3, there exists a path $P \pathtau P_0$ such that $P_0 \in [P]$ and $\deadend{P_0}{X}$. Otherwise, set $P_0 := P$. Since $[P_0] \step{a} R'$, according to Lemma \ref{lem:transition class}.1, there exists a path $P_0 \pathtau P_1 \step{a} P_2$ such that $P_1 \in [P_0]$ and $P_2 \in R'$. Notice that $\deadend{[P]}{X} \Rightarrow \deadend{P_0}{X} \wedge P_0 = P_1$. Thus, there exists a path $\theta_X(P_0) \pathtau \theta_X(P_1) \step{a} P_2$ such that $\theta_X(P_1) \in [\theta_X(P)]$. According to Lemma \ref{lem:transition class}.1, there exists a transition $[\theta_X(P)] \step{a} R'$.
        \item If $[\theta_X(P)] \step{a} R'$ with $a \in A$ then, according to Lemma \ref{lem:transition class}.1, there exists a path $\theta_X(P) \pathtau P_1 \step{a} P_2$ such that $P_1 \in [\theta_X(P)]$ and $P_2 \in R'$. Thus, there exists a path $P \pathtau P^\dag \step{a} P_2$ such that $P_1 = \theta_X(P^\dag)$ and $a \in X \vee \deadend{P^\dag}{X}$. According to Lemma \ref{lem:transition class}.1, there exists a path $[P] \pathtau [P^\dag] \step{a} [P_2]$. Thus, $\theta_X([P]) \pathtau \theta_X([P^\dag]) \step{a} [P_2]$ since $\deadend{P^\dag}{X} \Rightarrow \deadend{[P^\dag]}{X}$ by Lemma \ref{lem:transition class}.3. Moreover, by definition of $\tbisim$, $\theta_X([P^\dag]) \tbisim [\theta_X(P^\dag)] = [\theta_X(P)]$ and $[P_2] = R' \tbisim^* R'$.
        \item If $\deadend{\theta_X([P])}{Y}$ and $\theta_X([P]) \step{\rt} R'$ then $\deadend{[P]}{X}$; thus, according to Lemma \ref{lem:transition class}.3, there exists a path $P \pathtau P_0$ such that $P_0 \in [P]$ and $\deadend{P_0}{X}$. Since $\deadend{P_0}{X}$, $P \bisimtbrc P_0 \bisim \theta_X(P_0) \bisimtbrc \theta_X(P)$ and so $[P] = [\theta_X(P)]$. Since $\deadend{[P]}{X}$, according to Lemma \ref{lem:class deadlock theta}, $\theta_X([P]) \bisim [P] = [\theta_X(P)]$.
        \item If $\deadend{[\theta_X(P)]}{Y}$ and $[\theta_X(P)] \step{\rt} R'$ then, according to Lemma \ref{lem:transition class}.4, there exists a path $\theta_X(P) \pathtau \theta_X(P_1) \step{\rt} P_2$ such that $\theta_X(P_1) \in [\theta_X(P)]$, $\deadend{\theta_X(P_1)}{Y}$ and $\theta_Y(P_2) \in [\theta_Y(\chi(R'))]$. Since $\theta_X(P_1) \step{\rt} P_2$, $P_1 \step{\rt} P_2$ and $\deadend{P_1}{X}$, thus, $\theta_X(P) \bisimtbrc \theta_X(P_1) \bisim P_1$. Therefore, there exists a path $P \pathtau P_1 \step{\rt} P_2$ such that $\deadend{P_1}{Y}$. According to Lemma \ref{lem:transition class}.1 and~\ref{lem:transition class}.5, there exists a path $[P] \pathtau [P_1] \step{\rt} R''$ for some $R''\in[\closed^g]$ with $\theta_Y(P_2) \in [\theta_Y(\chi(R''))]$. Thus, $\theta_X([P]) \pathtau \theta_X([P_1]) \step{\rt} R''$ since $\deadend{[P_1]}{X}$. Moreover, $\deadend{\theta_X([P_1])}{Y}$ since $\deadend{P_1}{X \cup Y}$ and $\theta_Y(R') \tbisim [\theta_Y(\chi(R'))] = [\theta_Y(P_2)] = [\theta_Y(\chi(R''))] \tbisim \theta_Y(R'')$.
        \item If $\theta_X([P]) \nsteptau$ then $[P] \nsteptau$. According to Lemma \ref{lem:transition class}.2, there exists a path $P \pathtau P_0 \nsteptau$ such that $P_0 \in [P]$. Thus, there exists a path $\theta_X(P) \pathtau \theta_X(P_0) \nsteptau$ such that $\theta_X(P_0) \in [\theta_X(P)]$. According to Lemma \ref{lem:transition class}.2, $[\theta_X(P)] \nsteptau$.
        \item If $[\theta_X(P)] \nsteptau$ then, according to Lemma \ref{lem:transition class}.2, there exists a path $\theta_X(P) \pathtau P_0 \nsteptau$ such that $P_0 \in [\theta_X(P)]$. Thus, there exists a path $P \pathtau P^\dag \nsteptau$ such that $P_0 = \theta_X(P^\dag)$. According to Lemma \ref{lem:transition class}.1--2, there exists a path $[P] \pathtau [P^\dag]\nsteptau$. Thus, $\theta_X([P]) \pathtau \theta_X([P^\dag]) \nsteptau$. 
\popQED          
    \end{itemize}
\end{proof}

\begin{proof}[Proof of Proposition \ref{prop:class}]
    We are going to show that $\tbisim := \{(P,[P]),([P],P) \mid P \in \closed^g\}$ is a \tb time-out bisimulation up to $\bisimtbrc$ (see Definition~\ref{def:up to b}).
    \begin{enumerate}
        \item 
        \begin{itemize}
            \item Let $P \pathtau P' \step{\alpha} P''$ with $\alpha \in A_\tau$ and $P \bisimtbrc P'$. If $\alpha \in A \vee P \,\not\!\bisimtbrc P''$ then, according to Lemma \ref{lem:transition class}.1, $[P] \step{\alpha} [P'']$ and, by definition of $\tbisim$, $P' \tbisim [P'] = [P]$ and $P'' \tbisim [P'']$. Otherwise, $\alpha = \tau \wedge P \bisimtbrc P''$ thus, by definition of $\tbisim$, $P'' \tbisim [P''] = [P]$. In either case, there exists a path $[P] \step{\opt{\alpha}} [P'']$ such that $P' \tbisim [P]$ and $[P''] \tbisim P''$.
            \item If $[P] \pathtau R' \step{\alpha} R''$ with $\alpha \in A_\tau$ and $[P] \bisimtbrc R'$ then, according to Lemma \ref{lem:transition class}.1, $P \pathtau P_1 \step{\alpha} P_2$ such that $P_1 \in R'$ and $P_2 \in R''$. Thus, by definition of $\tbisim$, $P_1 \tbisim [P_1] = R'$ and $P_2 \tbisim [P_2] = R''$. 
        \end{itemize}
        \item
        \begin{itemize}
            \item Let $P \pathtau P_1 \step{\rt} P_2$ with $P \bisimtbrc P_1 \wedge \deadend{P_1}{X}$. By Lemma \ref{lem:transition class}.5 there exists a transition $[P] \mathbin{\step{\rt}} R'$ such that $\theta_X(P_2) \bisimtbrc \theta_X(\chi(R')) \tbisim [\theta_X(\chi(R'))] \bisimtbrc \theta_X(R')$. Moreover, by Lemma~\ref{lem:transition class}.3, $\deadend{[P]}{X}$ since $\deadend{P_1}{X}$ and $P_1 \in [P]$.
            \item Let $[P] \pathtau R_1 \step{\rt} R_2$ with $[P] \bisimtbrc R_1 \wedge \deadend{R_1}{X}$. Then, according to Lemma \ref{lem:transition class}.4 and Corollary~\ref{cor:transition class}, there exists a path $P \pathtau P_1 \step{\rt} P_2$ with $\deadend{P_1}{X}$ and $\theta_X(P_2) \in [\theta_X(\chi(R_2))]$. Thus, applying Lemma~\ref{lem:theta class}, $\theta_X(P_2) \tbisim [\theta_X(P_2)] = [\theta_X(\chi(R_2))] \bisimtbrc \theta_X(R_2)$.
        \end{itemize}
        \item
        \begin{itemize}
            \item If $P \pathtau P_0 \nsteptau$ with $P \bisimtbrc P_0$ then, according to Lemma \ref{lem:transition class}.2, $[P] \nsteptau$.
            \item If $[P] \pathtau R' \nsteptau$ with $[P] \bisimtbrc R'$ then, according to Lemma \ref{lem:transition class}.1--2, there exists a path $P \pathtau P' \pathtau P_0 \nsteptau$ such that $P',P_0 \in R'$.
            \popQED
        \end{itemize}
    \end{enumerate}
\end{proof}

\section{Completeness Proof by Canonical Representatives} \label{app:canonical}

\begin{lemma} \label{lem:simplication}
    Let $P,Q \in \closed^g$.
    \begin{itemize}
        \item $[P] \bisimtbrc[] [Q] \Rightarrow [P] = [Q]$.
        \item $\theta_X([P]) \bisimtbrc \theta_X([Q]) \Rightarrow \theta_X([P]) \bisimb \theta_X([Q])$.
    \end{itemize}
\end{lemma}

\begin{proof}
    ~
    \begin{itemize}
        \item If $[P] \bisimtbrc[] [Q]$ then, by Proposition~\ref{prop:class}, $P \bisimtbrc[] [P] \bisimtbrc[] [Q] \bisimtbrc Q$. Thus, $[P] = [Q]$.
        \item We are going to show that ${\tbisim} := {\it Id} \cup \{(\theta_X([P]),\theta_X([Q])) \mid \theta_X([P]) \bisimtbrc \theta_X([Q])\}$ is a stability respecting branching bisimulation. Suppose $\theta_X([P]) \bisimtbrc \theta_X([Q])$. If $\theta_X([P]) \step{\alpha} R'$ with $\alpha \in A_\tau$ then the first clause of Definition \ref{def:intuitive} suffices. If $\theta_X([P]) \step\rt R'$ then $\deadend{[P]}{X}$ and $[P]\step\rt R'$. As $[P] \bisimtbrc[X] [Q]$, by Clause 2.c of Definition~\ref{def:intuitive} there is a path $[Q] \pathtau R$ for some $R\in[\closed^g]$ with $[P] \bisimtbrc R$. By the previous statement of this lemma, $R=[P]$. Thus $\theta_X([Q])\pathtau \theta_X([P]) \step\rt R'$, which suffices to satisfy the first clause of  Definition~\ref{def:non-reactive}. If $\theta_X([P]) \nsteptau$ then the stability-respecting clause of Definition \ref{def:intuitive} suffices.
\popQED
    \end{itemize}
\end{proof}

\begin{proof}[Proof of Proposition \ref{prop:canonical}]
    Let $\equa'$ be a recursive specification such that $V_{\equa'} := \{y_{P'} \mid P' \in \reach{P}\}$ and, for all $P' \in \reach{P}$, $\equa_{y_{P'}} := \sum_{\{(\alpha,P'') \mid P' \step{\alpha} P''\}}\alpha.y_{P''}$. Note that $\equa'$ is strongly guarded since $P$ is. We are going to show that $P$ and $\langle x_P|\equa\rangle$ are both $y_P$-components of solutions of $\equa'$, so that the proposition follows by RSP.

    First of all, consider $\rho: V_{\equa'} \rightarrow \closed$ such that $\forall P' \in \reach{P},\; \rho(y_{P'}) := P'$. For all $P' \in \reach{P}$, $Ax^\infty_r \vdash \rho(y_{P'}) = P' = \sum_{\{(\alpha,P'') \mid P' \step{\alpha} P''\}}\alpha.P'' = \sum_{\{(\alpha,P'') \mid P' \step{\alpha} P''\}}\alpha.\rho(y_{P''})$ is a direct application of Lemma \ref{lem:head-normal form}. Thus, for all $P' \in \reach{P}$, $Ax^\infty_r \vdash \rho(y_{P'}) = \equa'_{y_{P'}}[\rho]$, i.e., $\rho$ is a solution of $\equa'$ up to $\bisimtbrc$\,, and $\rho(y_P) = P$.

    Next, consider $\nu: V_{\equa'} \rightarrow \closed$ such that, for all $P' \in \reach{P}$, 
    \begin{align*}
        \nu(y_{P'}) := \sum_{\{(\alpha,P'') \mid P' \step{\alpha} P''\}}\alpha.\langle x_{[P'']}\mid\equa\rangle
    \end{align*}
    We are going to show that, for all $\alpha \in Act$ and all $P' \in \reach{P}$, $Ax^\infty_r \vdash \alpha.\nu(y_{P'}) = \alpha.\langle x_{[P']}\mid\equa\rangle$. Let $P' \in \reach{P}$.
    \begin{itemize}
        \item If $\exists P' \steptau P'',\; P' \bisimtbrc P''$ then $\{(\alpha,[P'']) \mid P' \step{\alpha} P'' \wedge (\alpha \in A \vee (\alpha=\tau \wedge P' \,\not\!\bisimtbrc P''))\} \subseteq \{(\alpha,R) \mid [P'] \step{\alpha} R\}$. Thus,
        \begin{align*}
            Ax^\infty_r \vdash \alpha.\nu(y_{P'}) & = \alpha.(\sum_{\{(\alpha,P'') \mid P' \step{\alpha} P'' \wedge \alpha \ne \rt\}}\alpha.\langle x_{[P'']}\mid\equa\rangle) \qquad\qquad(\hyperlink{Lt}{\mbox{\bf L}\tau})\\
            & = \alpha.(\sum_{\{(\alpha,P'') \mid P' \step{\alpha} P'' \wedge (\alpha \in A \vee (\alpha=\tau \wedge P \,\,\not\!\scriptrbis{c}{\!br}\, P'))\}}\hspace{-50pt}\alpha.\langle x_{[P'']}\mid\equa\rangle + \tau.\langle x_{[P']}\mid\equa\rangle) \\
            & = \alpha.\langle x_{[P']}\mid\equa\rangle \qquad\qquad\qquad\qquad\mbox{(branching axiom and RDP)}
        \end{align*}
        \item If $\forall P \steptau P',\, P \,\not\!\bisimtbrc P'$ then, for all $\alpha \in A_\tau$, $P \step{\alpha} P' \wedge P' \in R' \iff [P] \step{\alpha} R'$ and $\init{P} = \init{[P]}$. Moreover, if $\deadend{P}{X}$ and $[P] \step{\rt} R'$ then there exists a transition $P \step{\rt} P'$ with $\theta_X(P') \in [\theta_X(\chi(R'))]$. Thus $\theta_X([P']) \bisimtbrc \theta_X(R')$ so $\theta_X([P']) \bisimb \theta_X(R')$ by Lemma~\ref{lem:simplication}, and therefore $Ax^\infty_r \vdash \rt.\theta_X([P']) = \rt.\theta_X(R')$. Conversely, if $\deadend{P}{X}$ and $P \step{\rt} P'$ then there exists a transition $[P] \step{\rt} R'$ such that $\theta_X(P') \in [\theta_X(\chi(R'))]$ and thus $Ax^\infty_r \vdash  \rt.\theta_X([P']) =  \rt.\theta_X(R')$. Using the reactive approximation axiom, $Ax^\infty_r \vdash \nu(y_{P'}) = \langle x_{[P']} \mid\equa\rangle$ and so, for all $\alpha \in Act$, $Ax^\infty_r \vdash \alpha.\nu(y_{P'}) = \alpha.\langle x_{[P']} \mid\equa\rangle$.
    \end{itemize}
    As a result, for all $P' \in \reach{P}$, $ Ax^\infty_r \vdash \nu(y_{P'}) = \sum_{\{(\alpha,P'') \mid P' \step{\alpha} P''\}}\alpha.\langle x_{[P'']}|\equa\rangle = \sum_{\{(\alpha,P'') \mid P' \step{\alpha} P''\}}\alpha.\nu(y_{P''}) = \equa_{y_{P'}}[\nu]$, so $\nu$ is a solution of $\equa'$ up to $\bisimtbrc$\,. Moreover, $\nu(y_P) = \sum_{\{(\alpha,P') \mid P \step{\alpha} P'\}}\alpha.\langle x_{[P']}|\equa\rangle$ which can be equated to $\langle x_P | \equa\rangle$ by a single application of RDP.
\end{proof}

\begin{proof}[Proof of Theorem \ref{thm:canonical}]
    According to Proposition \ref{prop:canonical}, it suffices to establish that $Ax^\infty_r \vdash \langle x_P|\equa\rangle = \langle x_Q|\equa\rangle$. By applying RDP, this amounts to proving that
    \begin{align*}
        Ax^\infty_r \vdash \sum_{\{(\alpha,P') \mid P\step{\alpha}P'\}}\alpha.\langle x_{[P']}\mid\equa\rangle = \sum_{\{(\alpha,Q') \mid Q\step{\alpha}Q'\}}\alpha.\langle x_{[Q']}\mid\equa\rangle
    \end{align*}
    Let $(\alpha,P')$ such that $P \step{\alpha} P'$ and $\alpha \in A_\tau$. Since $P \bisimrtbrc Q$, there exists a transition $Q \step{\alpha} Q'$ such that $P' \bisimtbrc Q'$. Thus, $[P'] = [Q']$ and so $\langle x_{[P']}\mid\equa\rangle = \langle x_{[Q']}\mid\equa\rangle$. The same observation can be made for all $(\alpha,Q')$ such that $Q \step{\alpha} Q'$ and $\alpha \in A_\tau$. As a result, $\init{P} = \init{Q}$ and
    \begin{align*}
        Ax^\infty_r \vdash \sum_{\{(\alpha,P') \mid P\step{\alpha}P' \wedge \alpha \in A_\tau\}}\alpha.\langle x_{[P']}\mid\equa\rangle = \sum_{\{(\alpha,Q') \mid Q\step{\alpha}Q' \wedge \alpha \in A_\tau\}}\alpha.\langle x_{[Q']}\mid\equa\rangle
    \end{align*}
    Let $(\rt,P')$ be such that $P \step{\rt} P'$. Since $P \bisimrtbrc Q$, for all $X \subseteq A$ such that $\deadend{P}{X}$, there exists a transition $Q \step{\rt} Q'$ such that $\theta_X(P') \bisimtbrc \theta_X(Q')$. Thus, \hyperlink{recall}{recalling that $\langle x_R |\equa\rangle \bisim R$ for all $R$}, $\theta_X(\langle x_{[P']}\mid\equa\rangle) \bisimtbrc \theta_X(\langle x_{[Q']}\mid\equa\rangle)$, so, by Lemma \ref{lem:simplication}, $\theta_X(\langle x_{[P']}\mid\equa\rangle)\linebreak[2] \bisimb \theta_X(\langle x_{[Q']}\mid\equa\rangle)$ and hence $\rt.\theta_X(\langle x_{[P']}\mid\equa\rangle) \bisimrb \rt.\theta_X(\langle x_{[Q']}\mid\equa\rangle)$. Since $Ax^\infty$ is a subset of $Ax^\infty_r$, according to Theorem \ref{thm:completeness}, $Ax^\infty_r \vdash \rt.\theta_X(\langle x_{[P']}\mid\equa\rangle) = \rt.\theta_X(\langle x_{[Q']}\mid\equa\rangle)$. The same observation can be made for all $(\rt,Q')$ such that $Q \step{\rt} Q'$. Let $X \subseteq A$. If $P \step{\alpha}$ with $\alpha \in X \cup\{\tau\}$ then 
    \begin{align*}
        Ax^\infty_r \vdash \psi_X(\langle x_{P}|\equa\rangle) & = \sum_{\{(\alpha,P') \mid P\step{\alpha}P' \wedge \alpha \in A_\tau\}}\alpha.\langle x_{[P']}\mid\equa\rangle \\
        & = \sum_{\{(\alpha,Q') \mid Q\step{\alpha}Q' \wedge \alpha \in A_\tau\}}\alpha.\langle x_{[Q']}\mid\equa\rangle \\
        & = \psi_X(\langle x_{Q}\mid\equa\rangle)
    \end{align*}
    Otherwise, $\deadend{P}{X}$ so $\deadend{Q}{X}$, thus,
    \begin{align*}
        Ax^\infty_r \vdash \psi_X(\langle x_{P}|\equa\rangle) & = \sum_{\{(\alpha,P') \mid P\step{\alpha}P' \wedge \alpha \in A_\tau\}}\!\!\alpha.\langle x_{[P']}\mid\equa\rangle + \sum_{\{(t,P') \mid P\step{\rt}P'\}}\rt.\theta_X(\langle x_{[P']}\mid\equa\rangle) \\
        & = \sum_{\{(\alpha,Q') \mid Q\step{\alpha}Q' \wedge \alpha \in A_\tau\}}\!\!\alpha.\langle x_{[Q']}\mid\equa\rangle + \sum_{\{(t,Q') \mid Q\step{\rt}Q'\}}\rt.\theta_X(\langle x_{[Q']}\mid\equa\rangle) \\
        & = \psi_X(\langle x_{Q}\mid\equa\rangle)
    \end{align*}
    Using the reactive approximation axiom, $Ax^\infty_r \vdash \langle x_{P}|\equa\rangle = \langle x_{Q}|\equa\rangle$.
\end{proof}

%% file: main.bbl
\begin{thebibliography}{10}

\bibitem{BW90}
Jos~C.M. Baeten and W.~Peter Weijland.
\newblock {\em Process Algebra}.
\newblock Cambridge Tracts in Theoretical Computer Science 18. Cambridge
  University Press, 1990.
\newblock \href {https://doi.org/10.1017/CBO9780511624193}
  {\path{doi:10.1017/CBO9780511624193}}.

\bibitem{BHR84}
Stephen~D. Brookes, Tony~(C.A.R.) Hoare, and Bill~(A.W.) Roscoe.
\newblock A theory of communicating sequential processes.
\newblock {\em Journal of the ACM}, 31(3):560--599, 1984.
\newblock \href {https://doi.org/10.1145/828.833} {\path{doi:10.1145/828.833}}.

\bibitem{BGKLNVWWW19}
Olav Bunte, Jan~Friso Groote, Jeroen J.~A. Keiren, Maurice Laveaux, Thomas
  Neele, Erik~P. de~Vink, Wieger Wesselink, Anton Wijs, and Tim A.~C. Willemse.
\newblock The {mCRL2} toolset for analysing concurrent systems---improvements
  in expressivity and usability.
\newblock In Tom\'a\v{s} Vojnar and Lijun Zhang, editors, {\em {\rm Proc. 25th
  International Conference on} Tools and Algorithms for the Construction and
  Analysis of Systems, {\rm TACAS'19, held as part of the} European Joint
  Conferences on Theory and Practice of Software, {\rm {ETAPS}'19, Prague,
  Czech Republic}}, volume 11428 of {\em \rm LNCS}, pages 21--39. Springer,
  2019.
\newblock \href {https://doi.org/10.1007/978-3-030-17465-1_2}
  {\path{doi:10.1007/978-3-030-17465-1_2}}.

\bibitem{Fok00}
Wan~J. Fokkink.
\newblock {\em Introduction to Process Algebra}.
\newblock Texts in Theoretical Computer Science, An EATCS Series. Springer,
  2000.
\newblock \href {https://doi.org/10.1007/978-3-662-04293-9}
  {\path{doi:10.1007/978-3-662-04293-9}}.

\bibitem{modalstab}
Wan~J. Fokkink, Rob~J. van Glabbeek, and Bas Luttik.
\newblock Divide and congruence {III}: From decomposition of modal formulas to
  preservation of stability and divergence.
\newblock {\em Information and Computation}, 268:104435, 2019.
\newblock \href {https://doi.org/10.1016/j.ic.2019.104435}
  {\path{doi:10.1016/j.ic.2019.104435}}.

\bibitem{GLMS11}
Hubert Garavel, Fr\'ed\'eric Lang, Radu Mateescu, and Wendelin Serwe.
\newblock {CADP} 2010: {A} toolbox for the construction and analysis of
  distributed processes.
\newblock In Parosh~Aziz Abdulla and K.~Rustan~M. Leino, editors, {\em {\rm
  Proceedings} Tools and Algorithms for the Construction and Analysis of
  Systems, {\rm TACAS '11}}, volume 6605 of {\em \rm LNCS}, pages 372--387.
  Springer, 2011.
\newblock \href {https://doi.org/10.1007/978-3-642-19835-9_33}
  {\path{doi:10.1007/978-3-642-19835-9_33}}.

\bibitem{vG93}
Rob J.~van Glabbeek.
\newblock The linear time -- branching time spectrum {II}; the semantics of
  sequential systems with silent moves (extended abstract).
\newblock In E.~Best, editor, {\em {\rm Proceedings} CONCUR'93, {\rm 4$^{\it
  th}$ International Conference on} Concurrency Theory, {\rm Hildesheim,
  Germany, August 1993}}, volume 715 of {\em \rm LNCS}, pages 66--81. Springer,
  1993.
\newblock \href {https://doi.org/10.1007/3-540-57208-2_6}
  {\path{doi:10.1007/3-540-57208-2_6}}.

\bibitem{vG17b}
Rob J.~van Glabbeek.
\newblock Lean and full congruence formats for recursion.
\newblock In {\em {\rm Proceedings $32^{nd}$ Annual ACM/IEEE Symposium on}
  Logic in Computer Science, {\rm LICS'17, Reykjavik, Iceland, June 2017}}.
  IEEE Computer Society Press, 2017.
\newblock \href {https://doi.org/10.1109/LICS.2017.8005142}
  {\path{doi:10.1109/LICS.2017.8005142}}.

\bibitem{vG21}
Rob J.~van Glabbeek.
\newblock Failure trace semantics for a process algebra with time-outs.
\newblock {\em Logical Methods in Computer Science}, 17(2), 2021.
\newblock \href {https://doi.org/10.23638/LMCS-17(2:11)2021}
  {\path{doi:10.23638/LMCS-17(2:11)2021}}.

\bibitem{vG23a}
Rob J.~van Glabbeek.
\newblock Modelling mutual exclusion in a process algebra with time-outs.
\newblock {\em Information and Computation}, 294, 2023.
\newblock \href {https://doi.org/10.1016/j.ic.2023.105079}
  {\path{doi:10.1016/j.ic.2023.105079}}.

\bibitem{strongreactivebisimilarity}
Rob J.~van Glabbeek.
\newblock Reactive bisimulation semantics for a process algebra with timeouts.
\newblock {\em Acta Informatica}, 60(1):11–57, 2023.
\newblock \href {https://doi.org/10.1007/s00236-022-00417-1}
  {\path{doi:10.1007/s00236-022-00417-1}}.

\bibitem{GH19}
Rob J.~van Glabbeek and Peter H{\"o}fner.
\newblock Progress, justness and fairness.
\newblock {\em ACM Computing Surveys}, 52(4), August 2019.
\newblock \href {https://doi.org/10.1145/3329125} {\path{doi:10.1145/3329125}}.

\bibitem{branching}
Rob J.~van Glabbeek and W.~Peter Weijland.
\newblock Branching time and abstraction in bisimulation semantics.
\newblock {\em Journal of the ACM}, 43(3):555–600, 1996.
\newblock \href {https://doi.org/10.1145/233551.233556}
  {\path{doi:10.1145/233551.233556}}.

\bibitem{GF20}
Clemens Grabmayer and Wan~J. Fokkink.
\newblock A complete proof system for 1-free regular expressions modulo
  bisimilarity.
\newblock In H.~Hermanns, L.~Zhang, N.~Kobayashi, and D.~Miller, editors, {\em
  {\rm Proc.\ 35th Annual {ACM/IEEE} Symposium on} Logic in Computer Science,
  {\rm {LICS}'20}}, pages 465--478. {ACM}, 2020.
\newblock \href {https://doi.org/10.1145/3373718.3394744}
  {\path{doi:10.1145/3373718.3394744}}.

\bibitem{HM85}
Matthew Hennessy and Robin Milner.
\newblock Algebraic laws for nondeterminism and concurrency.
\newblock {\em Journal of the ACM}, 32(1):137--161, 1985.
\newblock \href {https://doi.org/10.1145/2455.2460}
  {\path{doi:10.1145/2455.2460}}.

\bibitem{LY20}
Xinxin Liu and Tingting Yu.
\newblock Canonical solutions to recursive equations and completeness of
  equational axiomatisations.
\newblock In I.~Konnov and L.~Kovacs, editors, {\em {\rm Proceedings 31st
  International Conference on} Concurrency Theory {\rm (CONCUR 2020)}}, volume
  171 of {\em Leibniz International Proceedings in Informatics (LIPIcs)}.
  Schloss Dagstuhl -- Leibniz-Zentrum f{\"{u}}r Informatik, 2020.
\newblock \href {https://doi.org/10.4230/LIPIcs.CONCUR.2020.35}
  {\path{doi:10.4230/LIPIcs.CONCUR.2020.35}}.

\bibitem{Mi90ccs}
Robin Milner.
\newblock Operational and algebraic semantics of concurrent processes.
\newblock In J.~van Leeuwen, editor, {\em Handbook of Theoretical Computer
  Science}, chapter~19, pages 1201--1242. Elsevier Science Publishers B.V.
  (North-Holland), 1990.
\newblock Alternatively see{ \em Communication and Concurrency}, Prentice-Hall,
  Englewood Cliffs, 1989, of which an earlier version appeared as{ \em A
  Calculus of Communicating Systems}, LNCS 92, Springer, 1980,
  doi:\href{http:dx.doi.org/10.1007/3-540-10235-3}{10.1007/3-540-10235-3}.

\bibitem{OH86}
Ernst-Ruediger Olderog and Tony~(C.A.R.) Hoare.
\newblock Specification-oriented semantics for communicating processes.
\newblock {\em Acta Informatica}, 23:9--66, 1986.
\newblock \href {https://doi.org/10.1007/BF00268075}
  {\path{doi:10.1007/BF00268075}}.

\bibitem{Pohlmann}
Maximilian Pohlmann.
\newblock Reducing strong reactive bisimilarity to strong bisimilarity.
\newblock Bachelor's thesis, Technische Universität Berlin, 2021.
\newblock URL:
  \url{https://maxpohlmann.github.io/Reducing-Reactive-to-Strong-Bisimilarity/thesis.pdf}.

\bibitem{RG24}
Gaspard Reghem and Rob J.~van Glabbeek.
\newblock Branching bisimilarity for processes with time-outs.
\newblock Technical report, 2024.
\newblock Available at \url{http://arxiv.org/abs/2408.10117}. Extended abstract
  in Rupak Majumdar and Alexandra Silva, editors: Proceedings 35th
  International Conference on \emph{Concurrency Theory}, CONCUR'24,
  \textsl{Leibniz International Proceedings in Informatics} (LIPIcs) 311,
  Schloss Dagstuhl -- Leibniz-Zentrum f\"ur Informatik, 2024, doi:
  \href{https://doi.org/10.4230/LIPIcs.CONCUR.2024.36}{10.4230/LIPIcs.CONCUR.2024.36}.

\end{thebibliography}
